\documentclass[11pt,letterpaper]{article}
\usepackage[margin=1in]{geometry}
\usepackage{tikz} % main package
\usetikzlibrary{arrows.meta,positioning,calc} % optional, very useful

\usepackage[compatibility=false]{caption}
\usepackage{amsmath,amsfonts}
\usepackage{algorithm}
\usepackage{array}
\usepackage[caption=false,font=normalsize,labelfont=sf,textfont=sf]{subfig}
\usepackage{afterpage}
\usepackage{textcomp}
\usepackage{stfloats}
\usepackage{pdflscape}
\usepackage{url}
\usepackage{verbatim}
\usepackage{cite}
\usepackage{parskip}

\usepackage{graphicx}
\graphicspath{{figs/}{matlab/}}
\usepackage{ifpdf}
\usepackage[table]{xcolor}
\usepackage{pgfplots}
\usepackage{mathrsfs}  
\usepackage{amsmath,amssymb,amsthm,bm,dsfont}  %数学公式字体   %\usepackage{mathtools} provide {dcases}
\usepackage{enumerate}  %枚举
\usepackage{color}
\usepackage{setspace}%调整行间距
\usepackage[hidelinks]{hyperref} %生成目录，公式引用自动括号
\usepackage{cleveref} % multiple refference
\usepackage{cite}
\usepackage{subcaption}
\allowdisplaybreaks[2]  %公式分页显示，数字0-4表示分开程度
\usepackage{multirow}
\usepackage{cases}
\usepackage{balance}
\usepackage{algpseudocode}
\usepackage[dvipsnames]{xcolor}

\newif\ifEnv

\Envtrue
\ifEnv
    
\else
    \excludecomment{Env}
\fi

\usepackage{optidef}

\newtheorem{theorem}{\textbf{Theorem}}
\newtheorem{prop}{\textbf{Proposition}}

\newtheorem{definition}{\textbf{Definition}}%[section]
\newtheorem{lemma}{\textbf{Lemma}}
\newcommand{\supp}{\mathrm{supp}}

	\newcommand{\cX}{\mathcal{X}}
	
    \newcommand{\cZ}{\mathcal{Z}}

\newcommand{\parttwo}[1]{%
  \begingroup
  \setlength{\fboxsep}{1pt}
  \boxed{#1}%
  \endgroup
}

\newcommand{\revision}[1]{{#1}}

\begin{document}

\title{Optimal Multi-bit Generative Watermarking Schemes Under Worst-Case False-Alarm Constraints}

\author{Yu-Shin Huang, Chao Tian, and Krishna Narayanan}
        % <-this % stops a space
%\thanks{}% <-this % stops a space
%\thanks{}}

\markboth{Journal of \LaTeX\ Class Files,~Vol.~1, No.~2, December~2023}%
{Shell \MakeLowercase{\textit{et al.}}: A Sample Article Using IEEEtran.cls for IEEE Journals}

%\IEEEpubid{0000--0000~\copyright~2023 IEEE}

\maketitle

\begin{abstract}
We consider the problem of multi-bit generative watermarking for large language models under a worst-case false-alarm constraint. Prior work established a lower bound on the achievable miss-detection probability in the finite-token regime and proposed a scheme claimed to achieve this bound. We show, however, that the proposed scheme is in fact suboptimal. We then develop two new encoding-decoding constructions that attain the previously established lower bound, thereby completely characterizing the optimal multi-bit watermarking performance. Our approach formulates the watermark design problem as a linear program and derives the structural conditions under which optimality can be achieved. In addition, we identify the failure mechanism of the previous construction and compare the tradeoffs between the two proposed schemes.
\end{abstract}

%\begin{IEEEkeywords}
%Article submission, IEEE, IEEEtran, journal, \LaTeX, paper, template, typesetting.
%\end{IEEEkeywords}

\section{Introduction}
%\IEEEPARstart{T}{his} 

As Large Language Models (LLMs) continue to scale and their outputs become increasingly indistinguishable from human-written text, they raise serious concerns such as the spread of fabricated news and false academic content. Consequently, the need for robust watermarking techniques \cite{cox2008digital} for LLM-generated text is becoming ever more critical. Watermarking seeks to insert hidden information into generated content to reliably indicate its source and authorship, while remaining undetectable so that users cannot remove the concealed message components easily.

Recent LLM watermarking approaches generally fall into two broad categories: training-time \cite{sun2023codemark,xu2024learning, gu2023learnability,block2025gaussmark,zhao2025marktune} and inference-time approaches \cite{liu2023unforgeable,aaronson2023watermarking, he2024theoretically,he2025distributional,kirchenbauer2023watermark,kuditipudi2023robust,zhao2023provable,liu2024adaptive,hu2023unbiased,wu2023resilient,chao2024watermarking, dathathri2024scalable,long2025optimized,zhu2024duwak,takezawa2023necessary,chen2025improved,christ2024undetectable,wouters2023optimizing,kirchenbauer2023reliability,an2026reinforcement}. Training-time approaches embed watermarking behaviour during training or fine-tuning, causing the model to generate watermarked text directly. In contrast, this work targets inference-time methods, which embed messages by altering the sampling process during generation. These methods generally operate under a model-agnostic detection setting: while the generator and detector share auxiliary side information (secret keys), the detector lacks access to the generator’s parameters or architectural details. More discussion of the categorization can be found in \cite{liu2024survey, yang2025watermarking}

Current LLM watermarking methods are mostly ``zero-bit": they only allow binary detection, indicating whether a piece of text was generated by AI or not, without encoding any extra information in the text. Most existing zero-bit approaches \cite{takezawa2023necessary,zhu2024duwak,hu2023unbiased,chen2025improved,wouters2023optimizing,kirchenbauer2023reliability} are extensions of a state-of-the-art technique known as the Green-Red list \cite{kirchenbauer2023watermark}. This approach uses side information to select a subset of tokens marked as green and then modifies token probabilities so that the generated text assigns higher probability mass to these green tokens. Alternatively, distortionless watermarking \cite{kuditipudi2023robust, hu2023unbiased,wu2023resilient,chao2024watermarking,long2025optimized,chen2025improved,christ2024undetectable,wouters2023optimizing} maintains the LLM next-token distribution. It constructs a joint distribution between tokens and side information such that the marginal over tokens coincides with the original LLM. By sampling tokens conditionally on the side information, this approach embeds watermarks without altering the text's statistical properties. 

While research on zero-bit watermarking has reached a certain level of maturity, existing multi-bit schemes \cite{jiang2025stealthink,wang2023towards,yoo2024advancing,boroujeny2024multi} are limited, usually relying on heuristic extensions of zero-bit methods. In practice, encoding more information into the text tends to either degrade the naturalness of the watermarked content or induce higher detection error probability, making it harder to resolve through empirical trial-and-error alone. Therefore, a rigorous analysis from an information-theoretic perspective \cite{moulin2000information,moulin2001role,cohen2002gaussian,merhav2002random,steinberg2002identification} is necessary to determine the information-rate limits underlying these trade-offs and to establish a principled foundation for watermarking.

Recently, He et al. proposed a theoretical framework for zero-bit watermarking \cite{he2024theoretically}, which was later extended to the multi-bit setting \cite{he2025distributional} under the worst-case false-alarm constraint. Under an information-theoretic framework, an asymptotic analysis of the maximum achievable information rate under given false-alarm constraints for multi-bit generative watermarking was given in \cite{he2025distributional} ; for the setting with a finite number of tokens, a detection-theory framework was used to study the optimal miss-detection performance. 

We made a critical observation that the encoding–decoding scheme introduced for the finite-token setting in \cite{he2025distributional}  unfortunately does not actually achieve the claimed optimal performance. Therefore, in this work, we revisit the multi-bit watermarking framework. We identify the fundamental cause for the construction proposed in \cite{he2025distributional}  to be suboptimal, which is the overly restrictive decoder design, and adopt a more general class of decoding functions. Two novel encoding-decoding constructions are proposed that can indeed achieve the lower bound established in \cite{he2025distributional} , thereby completely settling the problem of multi-bit generative watermarking under the worst-case false-alarm constraint. We formulate the problem naturally as a linear program and analyze the necessary conditions for the constructions to match the lower bound. 
The first construction utilizes a decomposition approach, where these conditions are carefully maintained, and a key concept we introduce is the $T$-hot representable vectors. The second construction instead relies on a pseudo-token approach, which is conceptually straightforward at the expense of a larger key set than the first one.

The remainder of this paper is organized as follows: Section 2 defines the valid watermarking system; Section 3 presents the main theorem on optimal detection error under worst-case false-alarm constraints; Section 4 introduces a generalized expression for deterministic decoders; Sections 5 and 6 provide our two optimal constructions; and Section 7 analyzes the failure cases of the previous construction in \cite{he2025distributional} , and also compares the two proposed constructions. Section 8 finally concludes the paper.

Notation: The set $\{a,a+1,\ldots,b\}$ for integers $a,b$ will be denoted as $[a:b]$. Denote the LLM’s vocabulary set as $\cX$, and write $N = |\cX|$ for its cardinality, where $|\cdot|$ represents the size of a set. Let $m \in [0:T]$ denote the watermark message that we want to encode into the generated text sequences, where $M = 0$ indicates that no watermark message is present, and  $T$ is the total number of messages that this watermarking system can embed. \revision{When $T=2^b$, each nonzero watermark message can represent a $b$-bit binary string, where $b=\log_2 T$. } 
We make the assumption that $T \leq N$. This is justified since $N$ is usually large, while the number of watermark messages employed in practice is significantly smaller.

\section{System Description}

\subsection{Generative Watermarking Encoder, Sampler, and Decoder}

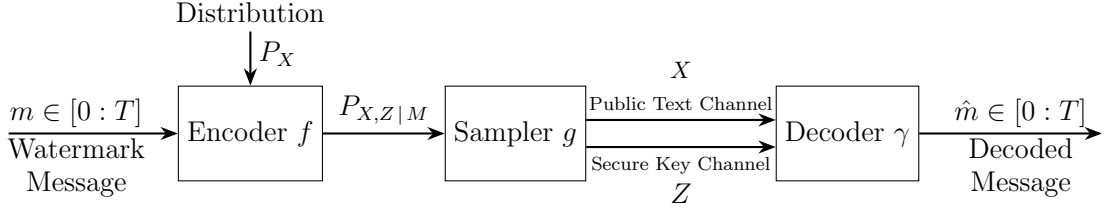
\begin{figure}[thb]
\centering
\begin{tikzpicture}[
  scale=0.9, transform shape,
  font=\large,
  block/.style={draw, minimum width=20mm, minimum height=14mm, align=center},
  >={Stealth[length=3mm]},
  line/.style={-Stealth, thick},
  inarrow/.style={Stealth-Stealth, thick} 
]

    % --- Blocks ---
    \node[block] (enc) {Encoder $f$};
    \node[block, right=18mm of enc] (sam) {Sampler $g$};
    \node[block, right=28mm of sam] (dec) {Decoder $\gamma$};
    
    % --- Main horizontal arrows ---
    \draw[line] ([xshift=-25mm]enc.west) -- (enc.west); % incoming line to encoder
    \draw[line] (enc.east) -- (sam.west);
    % \draw[line] (sam.east) -- (dec.west);
    \draw[line] (dec.east) -- ([xshift=27mm]dec.east); % outgoing line
    
    % --- Labels on horizontal arrows (use midway to place text) ---
    \node[align=center] at ([xshift=-15mm, , yshift=-6mm]enc.west) {Watermark\\Message};

    \node[align=center] at ([xshift=-15mm, , yshift=3mm]enc.west) {$m \in [0:T]$};
    
    \node[above] at ($(enc.east)!0.5!(sam.west)$) {$P_{X,Z\,|\,M}$};
   
    \draw[line] ($(sam.east)+(0,2mm)$) --
    node[above, align=center]{\small $X$\\[-0.5pt]\scriptsize Public Text Channel}
    ($(dec.west)+(0,2mm)$);

    \draw[line] ($(sam.east)+(0,-2mm)$) --
    node[below, align=center]{\small \scriptsize Secure Key Channel \\$Z$}
    ($(dec.west)+(0,-2mm)$);
    \node[align=center] at ([xshift=15mm, yshift=3mm]dec.east) {$\hat{m} \in [0:T] $};
    \node[align=center] at ([xshift=15mm, yshift=-6mm]dec.east) {Decoded \\ Message};
    
    \node[above=8mm of enc.north, align=center] (dist) {Distribution};
    \draw[line] (dist.south) -- (enc.north);
    \node[above] at ($(dist.south)!0.8!(enc.north)+ (4mm,0)$) {$P_{X}$};
\end{tikzpicture}
\caption{Diagram of a generative watermarking system. \label{fig:systemdiagram}}
\end{figure}

A generative watermark system consists of an encoder $f$, a sampler $g$, and a decoder $\gamma$; see Fig. \ref{fig:systemdiagram}. \revision{We adopt the multi-bit watermarking framework proposed by \cite{he2025distributional}  and briefly review the relevant watermark definitions below for completeness.}

\begin{definition}
A generative watermarking encoding function $f: M \times \Delta(\cX) \rightarrow \Delta(\cX\times\cZ)$ is a deterministic mapping, where $\Delta(\cdot)$ denotes the probability simplex on the given space. 
\end{definition}
Essentially, the function $f$ takes a message $M$ in the set $[0:T]$, where $M=0$ means no watermark is to be inserted, and a distribution $P_X \in \Delta(\cX)$ as input, and outputs a joint distribution on the space $\cX\times\cZ$, which we denote as $P_m(\cdot,\cdot) := f(m, P_X)$ when $M=m$; $T$ is the total number of messages allowed. $M$ can be viewed as a random variable, and then $P_m$ can be viewed as the conditional distribution $P(X,Z|M=m)$. The encoding procedure has a sampling module, which will use the random key $Z$ to produce a sample $X$ in $\cX$ according to the joint distribution given by the encoder, which is isolated out of the encoder in Fig. \ref{fig:systemdiagram}. The sampler is fixed and not up for design. 

\begin{definition}
A generative watermarking sampler is a stochastic mapping $g:\Delta(\cX\times\cZ)\rightarrow \cX\times\cZ$, which samples in the set $\cX\times \cZ$ according to the distribution $P_m$.
\end{definition}

\begin{definition} \label{def:decoder}
A watermarking decoding function $\gamma: \cX \times \cZ \rightarrow [0:T]$ is a deterministic mapping, which outputs a message based on the observed symbol in $\cX$ and a shared secret key in $\cZ$.
\end{definition}

Notice that the decoder and the sampler share the same random key $Z$, however, the encoder only determines the coupling structure between $X$ and $Z$, but does not utilize the common randomness $Z$ directly. \revision{We assume the shared key $Z$ remains secret from the adversary; therefore, attacks involving key compromise and adversarial modifications of $X$ performed with knowledge of $Z$ fall outside the scope of the problem setting considered here.} We impose certain requirements on the encoder to guarantee the secrecy of the embedded message. Specifically, the generated text $x \in \cX$ must be statistically indistinguishable regardless of which message $m$ is embedded, when the secret key $Z$ is not available. Furthermore, the secret key $Z$ itself, without $X$, should be independent of the message, due to the nature of the generative watermarking sampler. Formally, we state this requirement as follows:

\begin{definition}\label{req:encoder_requirement}
A generative watermarking encoder is called valid if both the marginal distribution $P_m(X)$ and the marginal distribution $P_m(Z)$ are invariant to $m\in[1:T]$. 
\end{definition}

When the system needs to produce non-watermarked output, that is, when $m = 0$, the distribution $P_0$ for a valid watermarking encoder is given by $P_0=P_X\otimes P_m(Z)$, for any $m\in[1:T]$ given the invariance above, i.e., the secret key $Z$ is independent of $X$ when $m=0$. We will only consider in this work the setting where there is no statistical difference between the marginal distribution $P_m(X)$ and $P_X$, which is usually referred to as distortionless watermarking.
\begin{definition}\label{req:distortionless}
A valid generative watermarking encoder is called distortionless if the marginal distribution $P_m(X)$ is the same as $P_X$ for any $m\in [1:T]$.
\end{definition}

\subsection{Decoding Errors}

An effective watermarking scheme should let the decoder correctly infer the message that the encoder intends to transmit, with sufficiently high probability. Let $M$ denote the embedded message, and let $\hat{M}$ be the message recovered by the decoder, which can be written as $\hat{M} = \gamma(X, Z)$. If $\hat{M} = 0$, the decoder interprets the signal as unwatermarked; if instead $\hat{M} = m$ for some $m \in [1:T]$, the decoder interprets it as watermarked with message $m$. 

Because there are multiple possible watermark messages, the problem can be cast as $(T+1)$-ary hypothesis testing. We denote by $H_0$ the hypothesis that the text is not watermarked, and by $H_m$ the hypothesis that the text is watermarked with message $m$.

The decoding performance is characterized by the errors made in the hypothesis test. We denote by $\beta_m$ the error probability corresponding to hypothesis $H_m$, meaning that the decoder should have recovered message $m$ but failed to do so. 
\begin{align}
    \beta_m(P_m, \gamma) = \sum_{x\in \cX, \zeta \in \cZ} P_m(x,\zeta)\mathds{1}\{\gamma(x,\zeta) \neq m \}, \quad \forall m \in [0:T].
\end{align}
In particular, $\beta_m$ is commonly referred to as the miss-detection error for $m \in [1:T]$, while $\beta_0$ is known as the false-alarm error, which occurs when the decoder reports that a watermark is present even though the text in fact contains no watermark. The worst-case false-alarm probability is expressed as $\sup_{Q_X} \beta_0(Q_X \otimes P_Z, \gamma)$ where $Q_X\in\Delta(\cX)$, representing the largest false-alarm error over all possible choices of any token distribution $Q_X$; we will denote this supremum as $\beta_0^*[f,\gamma]$. Since $P_m$ is the output of the encoder $f$, we will also use the notation $\beta_m(f,\gamma)$ to represent the miss-detection error. Throughout this article, we will use these two expressions interchangeably to refer to the same error quantity.

\section{Main Theorem} \label{sec:main-thm}

\revision{Following the problem formulation in \cite{he2025distributional} }, given a token distribution $P_X$, we define $\beta^*[\alpha, P_X]$ as the smallest possible value of the worst-case miss-detection probability achievable by any watermarking scheme $(f,\gamma)$, under the constraints that its worst-case false-alarm probability does not exceed a specified threshold $\alpha$, and that the encoder is distortionless, i.e., 
\begin{align*}
    \beta^*[\alpha, P_X] & := \min_{f,\gamma} \max_{m \in [1:T]} \beta_m(f, \gamma) \\
    \text{subject to: } & ~ 1.  \quad \beta_0^*[f,\gamma] \leq \alpha \\
    & ~ 2.  \quad \text{ $f$ is distortionless}
\end{align*}
\revision{The worst-case false-alarm constraint provides a
distribution-independent guarantee, ensuring that the
false-alarm probability remains bounded for any
distribution of non-watermarked text.
Under this constraint, the optimal miss-detection
performance is characterized by the following theorem.
}

\begin{theorem}\label{thm:main_thm}
    For any $P_X \in \Delta(\cX)$ and $\alpha \in [0,1)$, 
    \begin{align}
        \beta^*[\alpha, P_X] = 1 - \sum_{x\in \cX} \min\left(\tfrac{\alpha}{T}, P_X(x)\right). \label{eq:min_error}
    \end{align}
\end{theorem}

To prove this theorem, it needs to be shown that the minimum worst-case missed-detection error is lower-bounded by the right-hand side of (\ref{eq:min_error}); furthermore, we must design an encoder–decoder pair that achieves this minimum error. For this purpose, we rewrite the optimization as the following equivalent form (\ref{eq:main_opt_obj})–(\ref{eq:main_opt_condi_3}), which is now parameterized by $P_m$, $m\in[1:T]$, and $\gamma$.
\begin{align}
    \min_{\gamma,P_Z\in \Delta(\cZ),\{P_m\in \Delta(\cX\times\cZ),m\in[1:T]\}} \quad & \max_{m\in[1:T]} \beta_m(P_m, \gamma)\label{eq:main_opt_obj} \\
    \text{s.t.}  \quad & \sup_{Q_X\in \Delta(\cX)} \beta_0 (Q_X\otimes P_Z,\gamma) \leq \alpha \label{eq:main_opt_condi_1}\\
    & \sum_{\zeta \in \cZ} P_m(x, \zeta) = P_X(x), ~ \forall x \in X, \forall m \in [1:T] \label{eq:main_opt_condi_2}\\
    & \sum_{x \in \cX} P_m (x, \zeta) = P_{Z}(\zeta), ~ \forall \zeta \in Z, \forall m \in [1:T] \label{eq:main_opt_condi_3}
\end{align}

Let $P_{Z}(\zeta)$ denote the marginal distribution of $P_m$ over the key space. The encoder $f$ is now implicitly parametrized by the joint distribution $P_m$. Conditions (\ref{eq:main_opt_condi_2}) and (\ref{eq:main_opt_condi_3}) follow from the validity condition in Definition \ref{req:encoder_requirement} and \ref{req:distortionless}. The following two theorems together establish Theorem \ref{thm:main_thm}.

\begin{theorem} [Converse Bound]\label{thm:converse_bound}
    For any $P_X \in \Delta(\cX)$ and $\alpha \in [0,1)$, $\beta^*[\alpha, P_X]$ is lower bounded as 
    \begin{align}
        \beta^*[\alpha, P_X]\geq 1 - \sum_{x\in \cX} \min\left(\tfrac{\alpha}{T}, P_X(x)\right).
    \end{align}
\end{theorem}

\begin{theorem} [Achievability]\label{thm:forward}
     For any $P_X \in \Delta(\cX)$ and $\alpha \in [0,1)$, there exists a distortionless encoder and a corresponding decoder $(f, \gamma)$ such that 
    \begin{align}
    \max_{m\in [1:T]}\beta_m(P_m, \gamma)\leq 1- \sum_{x\in \cX} \min\left(\tfrac{\alpha}{T}, P_X(x)\right),\quad \sup_{Q_X\in\Delta(\cX)} \beta_0 (Q_X\otimes P_Z,\gamma) \leq \alpha
    \end{align}
\end{theorem}

Theorem \ref{thm:converse_bound} was proved in \cite{he2025distributional} . 
Theorem \ref{thm:forward} was purportedly also proved in \cite{he2025distributional} , however, there was a critical mistake in the construction, which leads to either invalid or suboptimal solutions. Therefore, we seek to show that there is a pair of $(f, \gamma)$ that fulfills the required distortionless property along with the needed performance. The distortionless properties for $P_m$ are already given in (\ref{eq:main_opt_obj})–(\ref{eq:main_opt_condi_3}) but we state them explicitly as the following properties on the distributions $P_m$'s: (i) \textbf{Column-sum invariance}: the column sums equal $P_X$, for all $m \in [1:T]$, which is the constraint (\ref{eq:main_opt_condi_2});  (ii) \textbf{Row-sum invariance}: all $P_m$ have the same row-sum for all $m \in [1:T]$, which is the constraint (\ref{eq:main_opt_condi_3}); and (iii) $P_m$ is in the probability simplex, for all $m\in[1:T]$, which is specified as the domain of the optimization variables.

In addition, to attain the minimum objective value $\beta^*[\alpha,P_X]$, the following two properties will be shown to lead to the desired performance:
\begin{itemize}
\item \textbf{$\alpha$-bounded total sum: }
\begin{align}
    \sum_{m\in [1:T]} \sum_{\zeta \in \cZ} P_{Z}(\zeta)\,\mathds{1}\{\gamma(x,\zeta) = m\}\leq \alpha, \quad \forall x \in \cX. \label{eq:alpha_bounded_total_sum}
\end{align}
\item \textbf{$\frac{\alpha}{T}$-capped column sum: } 
\begin{align}
    \sum_{\zeta \in \cZ} P_m(x,\zeta)\,\mathds{1}\{\gamma(x,\zeta) = m \} \geq \min\left(\frac{\alpha}{T}, P_X(x)\right),  \forall m \in [1:T],\ \forall x \in \cX\label{eq:alpha/M_capped_column_sum}
\end{align}
\end{itemize}

Formally, we have the following propositions. 
\begin{prop}\label{prop:alpha_bound}
If a valid $(f,\gamma)$ pair satisfies the $\alpha$-bounded total sum condition, then it satisfies the condition (\ref{eq:main_opt_condi_1}), i.e., 
\begin{align}
\sup_{Q_X\in\Delta(\cX)} \beta_0 (Q_X\otimes P_Z,\gamma) \leq \alpha. \label{eqn:alpha_bound}
\end{align}
\end{prop}

\begin{proof}[Proof of Proposition \ref{prop:alpha_bound}]
Because $(f,\gamma)$ satisfies the $\alpha$-bounded total sum condition, we have
\begin{align}
    & \sum_{m\in[1:T]} \sum_{\zeta \in \cZ} P_{Z}(\zeta)\,\mathds{1}\{\gamma(x,\zeta)=m\} \leq \alpha. 
\end{align}
For any $Q_X \in \Delta(\cX)$, each entry $Q_X(x)\in[0,1]$, and it follows that
\begin{align}
    \sup_{Q_X \in \Delta(\cX)} \sum_{x \in \cX} Q_X(x)\left(\sum_{\zeta \in \cZ}\sum_{m\in[1:T]}  P_{Z}(\zeta)\,\mathds{1}\{\gamma(x,\zeta)=m\} \right)\leq \alpha
\end{align}
However, the left-hand side can be rewritten as
\begin{align}
&\sup_{Q_X \in \Delta(\cX)} \sum_{x \in \cX}  \sum_{\zeta \in \cZ}Q_X(x)P_{Z}(\zeta) \left( \sum_{m\in[1:T]}\mathds{1}\{\gamma(x,\zeta)=m\} \right)\notag\\
    &=\sup_{Q_X \in \Delta(\cX)} \sum_{x \in \cX}\sum_{\zeta \in \cZ}Q_X(x)P_{Z}(\zeta) \mathds{1}\{\gamma(x,\zeta)\neq 0\}\notag\\
    &=\sup_{Q_X \in \Delta(\cX)} \beta_0 (Q_X\otimes P_Z,\gamma),
\label{eq:pf_prop_alpha_bound_goal} 
\end{align}
which is exactly what we need for (\ref{eqn:alpha_bound}), and the proof is complete.
\end{proof}

\begin{prop} \label{prop:alpha/T-capped}
    If a valid encoder-decoder pair $(f,\gamma)$ satisfies the condition of $\frac{\alpha}{T}$-capped column sum, then the corresponding $P_m$'s satisfy
    \begin{align}
    \max_{m\in [1:T]}\beta_m(P_m, \gamma)\leq 1 - \sum_{x\in \cX} \min\left(\tfrac{\alpha}{T}, P_X(x)\right).
    \end{align}
\end{prop}

\begin{proof}[Proof of Proposition \ref{prop:alpha/T-capped}]
For any $m \in [1:T] $, by the definition of $\beta_m(P_m, \gamma)$, we can write 
\begin{align}
    \beta_m(P_m, \gamma) & = \sum_{x\in \cX} \sum_{ \zeta \in \cZ} P_m(x,\zeta)\mathds{1}\{\gamma(x,\zeta) \neq m \} \\
    & = 1 - \sum_{x\in \cX} \left( \sum_{ \zeta \in \cZ}  P_m(x,\zeta)\mathds{1}\{\gamma(x,\zeta) = m \} \right) \\
    & \leq 1- \sum_{x\in \cX} \min(\frac{\alpha}{T}, P_X(x)), \label{eq:pf_alpha/M_capped}
\end{align}
where (\ref{eq:pf_alpha/M_capped}) follows directly from the $\frac{\alpha}{T}$-capped column sum in (\ref{eq:alpha/M_capped_column_sum}). This completes the proof.
\end{proof}

We shall introduce two constructions of $(P_m, \gamma)$ under which the watermarking system achieves the optimal value $1 - \sum_{x\in \cX} \min\left(\tfrac{\alpha}{T}, P_X(x)\right)$ in Sections \ref{sec:construction_A} and \ref{sec:construction_B}. Given the two propositions above, we only need to show that these constructions indeed satisfy the properties described earlier.

\section{Decoder Representation via Key Patterns}
Since the decoder $\gamma$ is deterministic, for each key value the decoder assigns a unique decision value $\hat{m} \in [0:T]$ to every pair $(x,\zeta)$ . Consequently, we can directly use a length $N = |\cX|$ vector to denote each key, whose $x$-th entry indicates the corresponding decision value $\hat{m}$. For instance, suppose a key is $\zeta=(0, 1, 2)$, then the decoder $\gamma$ maps the pair $(x,\zeta)=(1, (0, 1, 2))$ to $\hat{m} = 0$ and the pair $(x,\zeta)=(2, (0, 1, 2))$ to $\hat{m} = 1$.

Table \ref{tab:deterministic_dec_example} gives a simple illustration of how a deterministic decoder $\gamma$ maps a pair $(x, \zeta)$ to the decoded message $\hat{m}$. We further denote by $\zeta_i$ the $i$-th component of $\zeta$, which also implies that $\gamma(x = i, \zeta) = \zeta_i$. To account for all possible decoding functions, we consider every vector $\zeta$ of length $N$ whose components can take any value in $[0:T]$. Consequently, the secret key set $\cZ$ contains at most $(T+1)^N$ distinct $\zeta$. \revision{In Section \ref{sec:construction_A}, we show that the optimal value can be achieved using the reduced key set defined below.}

\begin{table}
\centering
\begin{tabular}{c|c|c|c}
$\gamma(x,\zeta)$   & ~$x = 1$ & $x = 2$ & ~$x = 3$  \\ 
\hline
$\zeta = (0, 1, 2)$ & $0$     & $1$     & $2$       \\ 
\hline
$\zeta = (0, 2, 1)$ & $0$     & $2$     & $1$           \\ 
\hline
$\zeta = (2, 0, 1)$ & $2$     & $0$     & $1$           \\ 
\hline
$\zeta = (1, 0, 2)$ & $1$     & $0$     & $2$             \\ 
\hline
$\zeta = (1, 2, 0)$ & $1$      & $2$     & $0$           \\ 
\hline
$\zeta = (2, 1, 0)$ & $2$     & $1$     & $0$            \\ 
\hline
$\zeta = (0, 0, 0)$ & $0$     & $0$     & $0$          
\end{tabular}
\caption{Example of deterministic decoder.}
\label{tab:deterministic_dec_example}
\end{table}

\begin{definition}[Reduced Key Set] \label{def:reduced_key_set}
    The canonical key vector is 
    \begin{align}
        b_T^N := (1, 2, \ldots, T, \mathbf{0}_{N-T}),
    \end{align}
    where $\mathbf{0}_{N-T}$ denotes the all-zero vector of length $N-T$. The $(N,T)$ reduced key set is given by the set of permutations of the canonical key, together with an all-zero key
    \begin{align}
        \mathcal{Z}_T^N \!:=\! \left\{ (b_T^N(\pi(1)), b_T^N(\pi(2)), \ldots, b_T^N(\pi(N))) \! \mid\! \pi : [1\!:\!N] \rightarrow [1\!:\!N] \text{ is any permutation} \} \cup \{ \mathbf{0}_N \right\},
    \end{align}
    where $b_T^N(i)$ is the value of $i$-th position in the standard key vector.
\end{definition}

The size of the key set is then reduced from $({T+1})^N$ to 
\begin{align*}
    |\cZ^N_T| = \left( \begin{gathered}N\\ T\end{gathered} \right)   \times T! + 1 =\frac{N!}{\left( N-T\right)  !} + 1.
\end{align*}

\begin{definition}[Preimage and preimage slice]\label{def:preimage}

Let $\gamma$ denote a decoder that fulfills Definition \ref{def:decoder}. The preimage of a message $m\in[0:T]$ is
\[
\gamma^{-1}(m) := \{(x,\zeta):\ \gamma(x,\zeta)=m\} = \{(x,\zeta):\ \zeta_x = m \}.
\]

For a fixed $x \in \mathcal{X}$, the preimage of $m$ under $\gamma$ at $x$ is
\[
\gamma^{-1}_x(m):= \{\zeta : \gamma(x,\zeta)=m\} = \{\zeta : \zeta_x = m\}. 
\]
Likewise, for fixed $\zeta \in \mathcal{Z}$, the preimage of $m$ at $\zeta$ is
\[
\gamma^{-1}_\zeta(m):= \{x : \gamma(x,\zeta)=m\} = \{x : \zeta_x = m\}.
\]
\end{definition}

In Definition \ref{def:preimage}, we denote the preimage of a message $m$, i.e., the set of all pairs $(x,\zeta)$ that the decoder maps to $m$, as $\gamma^{-1}(m)$. For a fixed $x$, the preimage slice of $m$ is $\gamma_x^{-1}(m)$, which is the set of all keys $\zeta$ such that the decoder maps $(x,\zeta)$ to $m$. We will frequently rely on this notation in the sequel.

\section{Construction A: A Decomposition Approach} \label{sec:construction_A}

\subsection{Decomposition of $P_X$}
We shall apply a decomposition and split $P_X$ into three components: $P_X^{(1)} + P_X^{(2)} + P_X^{(3)} = P_X$, where each $P_X^{(i)}$ induces a corresponding $P_m^{(i)}$ for $m \in [1:T]$, and the resulting construction is given as $P_m = P_m^{(1)} + P_m^{(2)} + P_m^{(3)}$. The reduced key set $\cZ = \cZ^N_T$ is used in this construction. Without loss of generality, we assume the elements of the vector $P_X$ to be nondecreasing. We will need the following notion of $T$-hot representability from Definition \ref{def:T-hot-representable} for the decomposition. \revision{Intuitively, it writes a vector as a weighted sum of $T$-hot vectors, each matching the $T$ nonzero entries of a key $\zeta$. This enables assigning probability mass to pairs $(x,\zeta)$ decoding to message $m$ while satisfying the marginal constraints. Since $P_X$ might not be $T$-hot representable, we introduce the decomposition of $P_X$, as detailed in Proposition \ref{prop:split_px}.}

\begin{definition}[$T$-hot Representable] \label{def:T-hot-representable}
    A vector $\boldsymbol{a}$ of length $N$ is said to be $T$-hot representable if $\boldsymbol{a}$ can be written as a non-negative weighted sum of all $T$-hot vectors $\boldsymbol{\omega}$ of length $N$; that is
    \begin{align}
        \boldsymbol{a}= \sum_{\boldsymbol{\omega} \in \Omega_T} \lambda_{\boldsymbol{\omega}}\boldsymbol{\omega}, \label{eq:T-hot_representable}
    \end{align}
    where $\Omega_T$ is the collection of all $T$-hot vectors
    \begin{align}
        \Omega_T = \{\boldsymbol{\omega} \in \{0, 1\}^N : \sum_{i=1}^N \boldsymbol{\omega}(i) = T\}.
    \end{align} 
\end{definition}

A necessary and sufficient condition for a vector to be $T$-hot representable is given below. 

\begin{lemma}\label{lemma:T-hot-representable-criterion} 
     A length-$N$ non-negative vector $\boldsymbol{a}$ is $T$-hot representable if and only if  
    \begin{align}
        T \cdot \max_i \boldsymbol{a}(i) \leq \sum_{i = 1}^N \boldsymbol{a}(i).\label{eq:T-hot-representable}
    \end{align}
\end{lemma}
\begin{proof}
    The proof of the ``if” direction is via Algorithm \ref{alg:T-hot-representable}, which constructs a set of $\lambda_{\boldsymbol{\omega}}$'s for any $T$-hot representable vector; we defer the proof of correctness of this algorithm to Appendix A. Note that in Algorithm \ref{alg:T-hot-representable}, we set $\min \emptyset = +\infty$, which occurs when $T = N$. %\ref{sec:correctness_T-hot-representable}.

\begin{algorithm}[h]
\caption{$T$-hot representable}
\label{alg:T-hot-representable}
\begin{algorithmic}[1]
\State \textbf{Input:} $N$ length $T$-hot representable vector $\boldsymbol{a}$; Index set $A = \emptyset $
\State $\boldsymbol{\nu} \gets \boldsymbol{a}$
\While{$\boldsymbol{\nu} \neq \mathbf{0}_N$}
    \State $A \gets \{ \text{index of the $T$ largest entries in } \boldsymbol{\nu} \}$
    \State $\boldsymbol{\omega} \gets \mathbf{1}_A$
    \State $\lambda_{\boldsymbol{\omega}} \gets \frac{1}{T}\min \left( \min_{j \in \supp(\boldsymbol{\omega})} \left( T \boldsymbol{\nu}(j)\right)  ,\min_{j\notin \supp(\boldsymbol{\omega})} \left( \sum^{N}_{i=1} \boldsymbol{\nu}(i)-T \boldsymbol{\nu}(j)\right)  \  \right)$ \label{alg-eq:lambda}
    \State $\boldsymbol{\nu} \gets \boldsymbol{\nu} - \lambda_{\boldsymbol{\omega}}\boldsymbol{\omega}$ \label{alg-eq:update_v}
\EndWhile
\State \textbf{Output:} $\{\lambda_{\boldsymbol{\omega}}, \omega\in \Omega_T\}$
\end{algorithmic}
\end{algorithm}

    For the ``only if"    direction, we can rewrite, using (\ref{eq:T-hot_representable}), the $i$-th entry of $\boldsymbol{a}$ as follows
    \begin{align}
        \boldsymbol{a}(i) = \sum_{\boldsymbol{\omega} \in \Omega_T} \lambda_{\boldsymbol{\omega}}\boldsymbol{\omega}(i).
    \end{align}
    Thus, the sum of $\boldsymbol{a}$ is 
    \begin{align}
        \sum_{i=1}^N \boldsymbol{a}(i) & = \sum_{i=1}^N \sum_{\boldsymbol{\omega} \in \Omega_T} \lambda_{\boldsymbol{\omega}}\boldsymbol{\omega}(i)=  \sum_{\boldsymbol{\omega} \in \Omega_T} \lambda_{\boldsymbol{\omega}} \left( \sum_{i=1}^N \boldsymbol{\omega}(i) \right) = T \sum_{\boldsymbol{\omega} \in \Omega_T} \lambda_{\boldsymbol{\omega}}
    \end{align}
    For every entry $i \in [1:N]$, the value of $T\cdot \boldsymbol{a}(i)$ is 
    \begin{align}
        T\cdot \boldsymbol{a}(i) & = T\cdot \sum_{\boldsymbol{\omega} \in \Omega_T} \lambda_{\boldsymbol{\omega}}\boldsymbol{\omega}(i) \leq T\cdot \sum_{\boldsymbol{\omega} \in \Omega_T} \lambda_{\boldsymbol{\omega}} = \sum_{i=1}^N \boldsymbol{a}(i), \label{eq:T-hot_rep_proof_ineq}
    \end{align}
    where the inequality in (\ref{eq:T-hot_rep_proof_ineq}) follows from $\boldsymbol{\omega}(i)\leq 1$ for all $i$. Therefore, (\ref{eq:T-hot-representable}) follows from (\ref{eq:T-hot_rep_proof_ineq}), which completes the proof. 
\end{proof}

The precise decomposition of $P_X$ is stated in Proposition \ref{prop:split_px}, which also ensures that $P_X^{(1)} + P_X^{(2)}$ is bounded above by $\tfrac{\alpha}{T}$, conforming with the $\tfrac{\alpha}{T}$-capped column sum condition in our construction.

\begin{prop}[$P_X$ Decomposition] \label{prop:split_px}
    Let $P_X$ be a distribution with $P_X(1) \leq P_X(2) \leq \dots \leq P_X(N)$. There exists a decomposition $P_X^{(1)}+P_X^{(2)}+P_X^{(3)} = P_X$ such that 
    \begin{itemize}
        \item[1.] $P_X^{(1)}$ is $T$-hot representable;
        \item[2.] $P_X^{(2)}$ is a non-negative non-decreasing vector with at most $T-1$  positive components, all of which are located in the final entries;
        \item[3.] $P_X^{(3)}(x) = \max(P_X(x) -\tfrac{\alpha}{T}, 0), \forall x$.
    \end{itemize}
    Moreover, either $P^{(2)}_X(x)=0$ for all $x\in\cX$, or 
    \begin{align}
    T\cdot\max_{x\in \cX}P_X^{(1)}(x)=\sum_{x\in\cX}P_X^{(1)}(x).\label{eqn:PX2}
    \end{align}
\end{prop}
Note that (\ref{eqn:PX2}) is essentially for $P_X^{(1)}$ to satisfy (\ref{eq:T-hot-representable}) with an equality. The proof of this proposition is given in Appendix B. %\ref{appendix:prop3}.

The intuition to construct $P_m^{(1)}(x,\zeta)$ and $P_m^{(2)}(x,\zeta)$ is as follows: to satisfy the $\tfrac{\alpha}{T}$-capped column sum property in (\ref{eq:alpha/M_capped_column_sum}), we need to distribute $\min(\tfrac{\alpha}{T}, P_X(x)) = P_X^{(1)}(x) + P_X^{(2)}(x)$ over those $(x,\zeta)$ pairs that the decoder associates with message $m$.  \revision{This is equivalent to specifying values for all pairs $(x,\zeta) \in \gamma^{-1}(m)$ to satisfy the $\tfrac{\alpha}{T}$-capped sum property and the $\alpha$-boundedness property. If the vector $\min\left(\tfrac{\alpha}{T}, P_X(x)\right)$ is already $T$-hot representable, then it is straightforward to construct a measure such that the column-sum (for this vector) and row-sum invariance are preserved; see section \ref{sec:construct_Pm1}. However, if this vector is not $T$-hot representable itself, we need to peel off its $T$-hot representable part $P_X^{(1)}(x)$ and address the remaining part $P_X^{(2)}(x)$ separately. The leftover component $P_X^{(3)}(x)=P_X(x)-\min(\tfrac{\alpha}{T}, P_X(x))$ does not affect the error probabilities and can be viewed as the slackness in the design. Consequently, the third component, $P_m^{(3)}(x,\zeta)$ is used to guarantee that $P_m$ satisfies the row-sum invariance condition, while also implicitly ensuring that the overall construction meets the $\alpha$-bounded total sum condition in (\ref{eq:alpha_bounded_total_sum}).} In the following three sections, we describe in detail how $P_m^{(i)}$ is constructed from the decomposed components of $P_X^{(i)}$. 

For convenience of notation, we will use the same properties, such as column-sum and row-sum invariance, to characterize any matrix $P^{(i)}_m$ (having the same dimensions as $P_m$).
In particular, when the column-sum target is a general non-negative vector $Q_X$ (instead of $P_X$), we say that
$P^{(i)}_m$ satisfies \textbf{$Q_X$-column-sum invariance}, defined by
\begin{equation}
    \label{eq:QX_column_sum_invariance}
    \sum_{\zeta \in \cZ}P^{(i)}_m(x,\zeta) = Q_X(x),\; \forall x \in \cX, \;  \forall m\in[1:T].
\end{equation}

\subsection{Construction of $P_m^{(1)}$ } \label{sec:construct_Pm1}
In the decomposition in Proposition \ref{prop:split_px}, we chose $P_X^{(1)}$ to be $T$-hot representable. We shall use this representation to construct $P_m^{(1)}$.

\subsubsection{Structural-Basis via $T$-hot Representative}

By Definition \ref{def:T-hot-representable}, $P_X^{(1)}$ can be expressed as a non-negative weighted sum of all $T$-hot vectors $\boldsymbol{\omega}$. Leveraging this representation, we construct $P_m^{(1)}$ by assigning appropriate weights to the pairs  $(x,\zeta) \in \gamma^{-1}(m)$.

\begin{prop} [Construction of $P_m^{(1)}$] \label{prop:construct_pm1}
    If $P_X^{(1)}$ is $T$-hot representable, then there exists a non-negative matrix $P_m^{(1)}$ such that $P_X^{(1)}$ can be perfectly allocated to the pairs $(x, \zeta) \in \gamma^{-1}(m)$ via the method of $T$-hot representation of $P_X^{(1)}$, and moreover $P_m^{(1)}$ itself satisfies $P_X^{(1)}$-column-sum and row-sum invariance.
\end{prop}

We prove Proposition \ref{prop:construct_pm1} in Appendix C %\ref{sec:proof_pm1}
using Algorithm \ref{alg:iterative_construction}. The procedure builds on the expression
\begin{align}
    P_X^{(1)} = \sum_{\boldsymbol{\omega} \in \Omega_T} \lambda_{\boldsymbol{\omega}}\boldsymbol{\omega}.
\end{align}
For each $T$-hot vector $\boldsymbol{\omega} \in \Omega_T$, we identify a structurally equivalent key set $\cZ_{\boldsymbol{\omega}}$, consisting of all elements whose non-zero entries occur in exactly the same positions as those of $\boldsymbol{\omega}$. Formally, we define
\begin{align}
    \cZ_{\boldsymbol{\omega}} = \{\zeta \in \cZ^N_T : \supp(\zeta) = \supp(\boldsymbol{\omega}) \},
\end{align}
where $\supp(\cdot)$ denotes the support of a vector, i.e., the set of indices of its non-zero components. 
Next, we distribute the weight $\lambda_{\boldsymbol{\omega}}$ uniformly across all rows $P_m^{(1)}(\cdot, \zeta)$ with $\zeta \in \cZ_{\boldsymbol{\omega}}$. More precisely, for each $\zeta \in \cZ_{\boldsymbol{\omega}}$, we uniformly assign a value to every pair $(x,\zeta)$ for which the decoder $\gamma$ outputs the message $m$. Thus, given each $x \in [1:N]$ and $m \in [1:T]$, we assign
\begin{align}
    P_m^{(1)}(x,\zeta) =  \frac{\lambda_{\boldsymbol{\omega}}}{\bigl|\cZ_{\boldsymbol{\omega}} \cap \gamma_x^{-1}(m)\bigr|}, \; \forall \zeta \in \cZ_{\boldsymbol{\omega}} \cap \gamma_x^{-1}(m).
\end{align}
Note that the value of $|\cZ_{\boldsymbol{\omega}} \cap \gamma_x^{-1}(m)\bigr| = (T-1)!$ if $x \in \supp(\omega)$. 
We iterate this procedure until all positive $\lambda_{\boldsymbol{\omega}}$'s have been handled. This naturally guarantees that the full mass of $P_X^{(1)}(x)$ is assigned to the pairs $(x,\zeta) \in \gamma^{-1}(m)$ within $P_m^{(1)}$, while simultaneously maintaining both the row-sum and column-sum invariance. 

\begin{algorithm}[h]
\caption{$P_m^{(1)}$: Structural-Basis via $T$-hot Representative}
\label{alg:iterative_construction}
\begin{algorithmic}[1]
\State \textbf{Input:} $N$ length $T$-hot representable $P_X^{(1)} = \sum_{\boldsymbol{\omega} \in \Omega_T} \lambda_{\boldsymbol{\omega}}\boldsymbol{\omega}$, key set $\cZ^N_T$, and initialized $P_m^{(1)} = \mathbf{0}, m \in [1:T]$.
\For{$\boldsymbol{\omega} \in \Omega_T$}
    \State $P^{(1)}_m(x, \zeta) \gets \frac{\lambda_{\boldsymbol{\omega}}}{(T-1)!}, \text{if} \; \zeta \in \cZ_{\boldsymbol{\omega}} \; \text{and}\; (x,\zeta) \in \gamma^{-1}(m)$
\EndFor
\end{algorithmic}
\end{algorithm}

\subsubsection{Example of $P_m^{(1)}$ Construction}\label{sec:A_pm1_ex}

We present an example with $N=4$ and $T=3$, where $P_X^{(1)} = (0.05, 0.1, 0.15, 0.15)$. The $T$-hot representable decomposition of this $P_X^{(1)}$ is
\begin{align}
    P_X^{(1)} = (0.05, 0.1, 0.15, 0.15) 
    = 0.1 \times (0,1,1,1) + 0.05 \times (1,0,1,1).
\end{align}
From this representation, we obtain the positive coefficients $\lambda_{(0,1,1,1)} = 0.1$ and $\lambda_{(1,0,1,1)} = 0.05$. Next, we apply Algorithm \ref{alg:iterative_construction} to build the corresponding $P_m^{(1)}$, which is presented in Table VI %\ref{tab:pm1_example2} 
in the Appendix. The red keys are in the set $\cZ_{(1,0,1,1)}$, while the green keys are in the set $\cZ_{(0,1,1,1)}$; the highlighted cells in the table indicate the pairs $(x,\zeta) \in \gamma^{-1}(m)$.

\subsection{Construction for $P_m^{(2)}$}
In this section, we construct $P_m^{(2)}$ by distributing $P_X^{(2)}$ to the pairs $(x, \zeta) \in \gamma^{-1}(m)$ in $P_m^{(2)}$, so that $P_m^{(1)} + P_m^{(2)}$ satisfies the $\tfrac{\alpha}{T}$-capped column sum condition, and the column-sum invariance property.

\subsubsection{Layered Step-Vector Allocation of $P_m^{(2)}$}

 To introduce the algorithm, we utilize a subset of keys, which are called anchored keys.

\begin{definition}[Anchored Key Set]
   Let $K := |\supp(P_X^{(2)})|$. Given that $P_X^{(2)}$ is non-decreasing and has $K \leq T-1$ positive components at its tail, we define the anchored key set as
\begin{align}
        S(P_X^{(2)}) 
        &:= \left\{ \zeta \in \cZ^{N}_{T} : \; \supp\big(P_X^{(2)}\big) \subset \supp(\zeta)\right\} \notag \\
        &= \left\{ \zeta \in \cZ^{N}_{T} : \; [N-K+1:N] \subset \supp(\zeta)\right\} 
        := S(K).
\end{align}
\end{definition}
In other words, $S(P_X^{(2)})$ (or equivalently $S(K)$) consists of all keys $\zeta$ whose last $K$ components contain only positive values in the set $[1:T]$.

Since $P_X^{(2)}$ is a non-decreasing vector with $K \leq T-1$ positive entries at its tail, it can be represented as a sum of step-vectors of varying lengths:
\begin{align}
    P_X^{(2)} = (\mathbf{0}_{N-K},\eta_{N-K+1}, \cdots, \eta_N) 
    & = \Delta\eta_{N-K+1} \boldsymbol{u}_{K} + \Delta\eta_{N-K+2} \boldsymbol{u}_{K-1} + \ldots + \Delta\eta_{N} \boldsymbol{u}_{1} \notag \\
    & = \sum_{j=1}^K \Delta\eta_{N-j+1} \boldsymbol{u}_{j}\label{eq:express_eta_i_same_eta},
\end{align}
where we list nonzero components of $P_X^{(2)}$ with positive value $\eta_i$, and $\Delta\eta_i = \eta_i - \eta_{i-1} \geq 0 $ denotes the increment relative to the previous component, where the vector $\boldsymbol{u}_j$ is a step-vector with $j$ nonzero entries in the tail:
\begin{align}
    \boldsymbol{u}_j = (\mathbf{0}_{N-j}, \mathbf{1}_{j}).
\end{align}

The algorithm sequentially assigns the components $\Delta\eta_{N-j+1}\boldsymbol{u}_j$ in $P_X^{(2)}$ to $P_m^{(2)}$. The procedure is given in Algorithm \ref{alg:fill_eta_layered}.

\begin{algorithm}[h]
\caption{$P_m^{(2)}$: Layered Step-vector Allocation}
\label{alg:fill_eta_layered}
\begin{algorithmic}[1]
\State \textbf{Input:} Decomposition $P_X^{(2)} = \sum_{j=1}^K \Delta\eta_{N-j+1}\boldsymbol{u}_j$; Initialized $P^{(2)}_m = \mathbf{0}, m \in [1:T]$ 
\For{$j = K, K-1,\ldots ,1$}
    \For{$x = N-j+1, \ldots ,N$}
        \State $P_m^{(2)}(x, \zeta) \gets P_m^{(2)}(x, \zeta) + \frac{\Delta \eta_{N-j+1}}{|S(K) \cap \gamma_x^{-1}(m)|}$,   $ \forall \ \zeta \in S(K) \cap \gamma_x^{-1}(m)$
    \EndFor
\EndFor
\end{algorithmic}
\end{algorithm}

\begin{prop} \label{prop:construct_pm2}
    The construction $P_m^{(2)}$ given in Algorithm \ref{alg:fill_eta_layered} assigns only positive values to pairs $(x, \zeta) \in \gamma^{-1}(m)$, and the combined distribution $P_m^{(1)} + P_m^{(2)}$ satisfies the $\tfrac{\alpha}{T}$-capped column-sum constraint, while $P_m^{(2)}$ itself preserves the $P_X^{(2)}$-column-sum invariance.
\end{prop}

The proof of Proposition $\ref{prop:construct_pm2}$ is given in Appendix D. %\ref{sec:proof_pm2_prop}.
Note that the row-sum property is not preserved in this construction. The reason is that with the restriction of using only pairs $(x, \zeta) \in \gamma^{-1}(m)$, it is impossible to do so. 
To see this, consider $P_1^{(2)}$ shown in Table VII
%\ref{tab:pm2_simple}
in the Appendix. If we assign a positive value to $\zeta = (3,2,0,1)$ or $\zeta = (2,3,0,1)$, then for $m = 2,3$, these $\zeta$ values have no gray cells in the third and fourth columns to assign any positive value. This would break the row-sum property. In fact, there is no choice of $\zeta$ that can entirely avoid the row-sum imbalance. The algorithm selects $\zeta$ properly so that this imbalance is well controlled by utilizing only the anchor key positions. We will regain the row-sum invariance with $P_m^{(3)}$.

\subsubsection{Example of $P_m^{(2)}$ Construction }

In this section, we present two examples illustrating how to construct $P_m^{(2)}$: one based on the step-vector $P_X^{(2)}$ in Table VII
%\ref{tab:pm2_simple} 
in the Appendix, and another using a more general $P_X^{(2)}$, reported in Table \ref{tab:pm2_general} along with the construction of $P_m^{(1)}$ given in the example in Section \ref{sec:A_pm1_ex}.
In both cases, we choose $N = 4$ and $T = 3$, which yields three non-negative matrices of $P_m^{(2)}$ in total for each example. The pairs $(x,\zeta)$ that are decoded as $m$ are highlighted in gray.

In Table 
VII %\ref{tab:pm2_simple} 
in the Appendix, we choose $P_X^{(2)} = (0,0,0.1,0.1) = 0.1\boldsymbol{u}_2$, which contains at most $T-1 = 2$ positive entries at the tail. We then apply Algorithm \ref{alg:fill_eta_layered} to obtain the corresponding $P_m^{(2)}$, performing the procedure only for $j = 2$. In this table, the first two columns sum up to the values $P_X^{(2)}(1)=P_X^{(2)}(2)=0$, so we leave these two columns blank. Note that in the third column, we do not spread $P_X^{(2)}(3) = 0.1$ evenly over all pairs $(x,\zeta) \in \gamma^{-1}(m)$, i.e., over all gray cells. The reason is that some choices of $\zeta$ would create more imbalance in the row-sums of the remaining $P_m^{(2)}$ tables. Assigning the value to row $(3, 0, 1,2)$ is preferable than using $(3,2,0,1)$ or $(2,3,0,1)$, because in this configuration the row-sum imbalance occurs in only one table rather than in both tables. The assigned keys are exactly the anchored key set $S(K)$.

In Table \ref{tab:pm2_general}, we form $P_m^{(1)}+P_m^{(2)}$ by applying Algorithms \ref{alg:iterative_construction} and \ref{alg:fill_eta_layered} to $\!P_X^{(1)}\!\! = (0.05, 0.1, 0.15, 0.15)$ and $\!P_X^{(2)}\!\! = \!\!(0,0,0.1,0.15)$, respectively. We again use the $P_m^{(1)}$ example from Section \ref{sec:A_pm1_ex}, which is the allocation reported in Table VI %\ref{tab:pm1_example2} 
in the Appendix. To build $P_m^{(2)}$, the quantity $P_X^{(2)}$ may be written as a weighted sum of step vectors:
\begin{align}
    P_X^{(2)} = (0,0,0.1,0.15) = 0.1\times(0,0,1,1) + 0.05\times(0,0,0,1) = 0.1\boldsymbol{u}_2 + 0.05\boldsymbol{u}_1.
\end{align}
Accordingly, algorithm \ref{alg:fill_eta_layered} first allocates $0.1\boldsymbol{u}_2$,
followed by $0.05\boldsymbol{u}_1$.
In Table \ref{tab:pm2_general}, the contributions from
$0.1\boldsymbol{u}_2$ and $0.05\boldsymbol{u}_1$
are shown in normal and bold fonts within boxes, respectively,
while the values outside the boxes correspond to $P_m^{(1)}$.

\begin{table}[]
\resizebox{\linewidth}{!}{%
\begin{tabular}{c||c|c|c|c||c|c|c|c||c|c|c|c}
            & \multicolumn{4}{c||}{$P_1^{(1)}+P_1^{(2)} (m=1)$} & \multicolumn{4}{c||}{$P_2^{(1)}+P_2^{(2)} (m=2)$} & \multicolumn{4}{c}{$P_3^{(1)}+P_3^{(2)} (m=3)$} \\ \hline
$\zeta$     & $x=1$   & $x=2$   & $x=3$   & $x=4$   & $x=1$   & $x=2$   & $x=3$   & $x=4$   & $x=1$   & $x=2$   & $x=3$   & $x=4$   \\ \hline\hline
$(3,2,1,0)$ &         &         &    \cellcolor{gray!15}     &         &         &     \cellcolor{gray!15}    &         &         &   \cellcolor{gray!15}      &         &         &         \\\hline
$(2,3,1,0)$ &         &         &   \cellcolor{gray!15}      &         &  \cellcolor{gray!15}       &         &         &         &         &   \cellcolor{gray!15}      &         &         \\\hline
$(3,1,2,0)$ &         &     \cellcolor{gray!15}    &         &         &         &         &     \cellcolor{gray!15}    &        &     \cellcolor{gray!15}    &         &         &         \\\hline
$(1,3,2,0)$ &   \cellcolor{gray!15}      &         &         &         &         &         &     \cellcolor{gray!15}    &         &         &   \cellcolor{gray!15}      &         &         \\\hline
$(2,1,3,0)$ &         &   \cellcolor{gray!15}      &         &         &      \cellcolor{gray!15}   &         &         &         &         &         &      \cellcolor{gray!15}   &         \\\hline
$(1,2,3,0)$ &  \cellcolor{gray!15}       &         &         &         &         &      \cellcolor{gray!15}   &         &         &         &         &   \cellcolor{gray!15}      &         \\\noalign{\hrule height 1pt}
$(3,2,0,1)$ &         &         &         &    \cellcolor{gray!15}     &         &       \cellcolor{gray!15}  &         &         &     \cellcolor{gray!15}    &         &         &         \\\hline
$(2,3,0,1)$ &         &         &         &     \cellcolor{gray!15}    &   \cellcolor{gray!15}      &         &         &         &         &    \cellcolor{gray!15}     &         &         \\\hline
$(3,0,2,1)$ &         &         &         &     \cellcolor{gray!15}$\tfrac{0.05}{2}+$$\parttwo{\tfrac{0.1}{4}+\mathbf{\tfrac{0.05}{4}}}$&         &         &      \cellcolor{gray!15}$\tfrac{0.05}{2}+$$\parttwo{\tfrac{0.1}{4}}$&         &      \cellcolor{gray!15}$\tfrac{0.05}{2}$    &         &         &         \\\hline
$(0,3,2,1)$ &         &         &         &       \cellcolor{gray!15}$\tfrac{0.1}{2}+$$\parttwo{\tfrac{0.1}{4}+\mathbf{\tfrac{0.05}{4}}}$&         &         &    \cellcolor{gray!15}$\tfrac{0.1}{2}+$$\parttwo{\tfrac{0.1}{4}}$&         &         &     \cellcolor{gray!15}$\tfrac{0.1}{2}$     &         &         \\\hline
$(2,0,3,1)$ &         &         &         &   \cellcolor{gray!15}$\tfrac{0.05}{2}+$$\parttwo{\tfrac{0.1}{4}+\mathbf{\tfrac{0.05}{4}}}$&  \cellcolor{gray!15}$\tfrac{0.05}{2}$       &         &         &         &         &         &    \cellcolor{gray!15}$\tfrac{0.05}{2}+$$\parttwo{\tfrac{0.1}{4}}$&         \\\hline
$(0,2,3,1)$ &         &         &         &  \cellcolor{gray!15}$\tfrac{0.1}{2}+$$\parttwo{\tfrac{0.1}{4}+\mathbf{\tfrac{0.05}{4}}}$&         &    \cellcolor{gray!15}$\tfrac{0.1}{2}$      &         &         &         &         &    \cellcolor{gray!15}$\tfrac{0.1}{2}+$$\parttwo{\tfrac{0.1}{4}}$&         \\\noalign{\hrule height 1pt}
$(3,1,0,2)$ &         &    \cellcolor{gray!15}     &         &         &         &         &         &    \cellcolor{gray!15}     &   \cellcolor{gray!15}      &         &         &    \\\hline 
$(1,3,0,2)$ &    \cellcolor{gray!15}     &         &         &         &         &         &         &    \cellcolor{gray!15}     &         &   \cellcolor{gray!15}      &         &    \\\hline 
$(3,0,1,2)$ &         &         &  \cellcolor{gray!15}$\tfrac{0.05}{2}+$$\parttwo{\tfrac{0.1}{4}}$&         &         &         &         &   \cellcolor{gray!15}$\tfrac{0.05}{2}+$$\parttwo{\tfrac{0.1}{4}+\mathbf{\tfrac{0.05}{4}}}$&      \cellcolor{gray!15}$\tfrac{0.05}{2}$    &         &         &    \\\hline 
$(0,3,1,2)$ &         &         &  \cellcolor{gray!15}$\tfrac{0.1}{2}+$$\parttwo{\tfrac{0.1}{4}}$&         &         &         &         &    \cellcolor{gray!15}$\tfrac{0.1}{2}+$$\parttwo{\tfrac{0.1}{4}+\mathbf{\tfrac{0.05}{4}}}$&         & \cellcolor{gray!15}$\tfrac{0.1}{2}$        &         &    \\\hline 
$(1,0,3,2)$ & \cellcolor{gray!15}$\tfrac{0.05}{2}$        &         &         &         &         &         &         & \cellcolor{gray!15}$\tfrac{0.05}{2}+$$\parttwo{\tfrac{0.1}{4}+\mathbf{\tfrac{0.05}{4}}}$&         &         &  \cellcolor{gray!15}$\tfrac{0.05}{2}+$$\parttwo{\tfrac{0.1}{4}}$&    \\\hline 
$(0,1,3,2)$ &         &  \cellcolor{gray!15}$\tfrac{0.1}{2}$       &         &         &         &         &         &\cellcolor{gray!15}$\tfrac{0.1}{2}+$$\parttwo{\tfrac{0.1}{4}+\mathbf{\tfrac{0.05}{4}}}$&         &         &   \cellcolor{gray!15}$\tfrac{0.1}{2}+$$\parttwo{\tfrac{0.1}{4}}$&    \\\noalign{\hrule height 1pt}
$(2,1,0,3)$ &         &  \cellcolor{gray!15}       &         &         &      \cellcolor{gray!15}   &         &         &         &         &         &         &  \cellcolor{gray!15}  \\\hline 

$(1,2,0,3)$ &  \cellcolor{gray!15}       &         &         &         &         & \cellcolor{gray!15}        &         &         &         &         &         & \cellcolor{gray!15}   \\\hline 

$(2,0,1,3)$ &         &         &  \cellcolor{gray!15}$\tfrac{0.05}{2}+$$\parttwo{\tfrac{0.1}{4}}$&         & \cellcolor{gray!15}$\tfrac{0.05}{2}$        &         &         &         &         &         &         & \cellcolor{gray!15}$\tfrac{0.05}{2}+$$\parttwo{\tfrac{0.1}{4}+\mathbf{\tfrac{0.05}{4}}}$  \\\hline 

$(0,2,1,3)$ &         &         & \cellcolor{gray!15}$\tfrac{0.1}{2}+$$\parttwo{\tfrac{0.1}{4}}$&         &         &  \cellcolor{gray!15}$\tfrac{0.1}{2}$       &         &         &         &         &         &\cellcolor{gray!15}$\tfrac{0.1}{2}+$$\parttwo{\tfrac{0.1}{4}+\mathbf{\tfrac{0.05}{4}}}$  \\\hline 

$(1,0,2,3)$ & \cellcolor{gray!15}$\tfrac{0.05}{2}$        &         &         &         &         &         &  \cellcolor{gray!15}$\tfrac{0.05}{2}+$$\parttwo{\tfrac{0.1}{4}}$&         &         &         &         & \cellcolor{gray!15}$\tfrac{0.05}{2}+$$\parttwo{\tfrac{0.1}{4}+\mathbf{\tfrac{0.05}{4}}}$   \\\hline 

$(0,1,2,3)$ &         & \cellcolor{gray!15}$\tfrac{0.1}{2}$        &         &         &         &         &   \cellcolor{gray!15}$\tfrac{0.1}{2}+$$\parttwo{\tfrac{0.1}{4}}$&         &         &         &         & \cellcolor{gray!15}$\tfrac{0.1}{2}+$$\parttwo{\tfrac{0.1}{4}+\mathbf{\tfrac{0.05}{4}}}$ \\ \noalign{\hrule height 1pt}
$(0,0,0,0)$ &         &        &         &         &         &         &        &         &         &         &         & 
\end{tabular}
}
\caption{Example of $P_m^{(1)} + P_m^{(2)}$.
The construction of $P_m^{(1)}$ follows the example in
Section \ref{sec:A_pm1_ex}, with
$P_X^{(1)} = (0.05, 0.1, 0.15, 0.15)$.
The boxed terms represent the contributions from $P_m^{(2)}$,
obtained through layered step-vector allocation with
$P_X^{(2)} = (0, 0, 0.1, 0.15)$.}
\label{tab:pm2_general}
\end{table}

\subsubsection{Row-sum Imbalance in $P_m^{(2)}$}

We express the row-sum imbalance in a mathematical form given in the following definition.

\begin{definition}[Row-sum Imbalance]
    The row-sum imbalance generated by \(P_m^{(2)}\) at each row \(\zeta \in \cZ^N_T\) is denoted as \(U_m(\zeta)\), and is defined as:
\[
    U_m(\zeta) = \max_{\mu \in [1:T]}\left( \sum_{x \in \cX} P_{\mu}^{(2)}(x, \zeta) \right) - \sum_{x \in \cX} P_{m}^{(2)}(x, \zeta), \quad \forall m \in [1:T],\, \zeta \in \cZ^N_T.
\]
The total row-sum imbalance associated with \(P_m^{(2)}\) is then given by $U_m = \sum_{\zeta \in \cZ^N_T} U_m(\zeta)$.
\end{definition}

In this definition, the row-sum imbalance $U_m(\zeta)$ for a given row $\zeta$ is defined as the difference between the row sum of the current table $P_m^{(2)}$ and the maximum row sum attained by that row in any table $P_{\mu}^{(2)}$ with $\mu \in [1:T]$. The total row-sum imbalance $U_m$ is then obtained by summing the imbalances over all rows $\zeta$.

\begin{prop} \label{prop:U_value}
Let $P_m^{(2)}$ be produced by Algorithm \ref{alg:fill_eta_layered}. The total row-sum imbalance induced by $P_m^{(2)}$ is $U_m = \sum_{j=1}^{K} \Delta\eta_{N-j+1}(T-j)$, for any $m \in [1:T]$.
\end{prop}

Since the imbalance value $U_m$ is independent of $m$, we simply write it as $U$.  The proof of Proposition \ref{prop:U_value} is given in Appendix E, %\ref{sec:prop_U_proof}, 
where we show that the total imbalance generated by $P_m^{(2)}$ is the sum of the imbalances resulting from assigning each step-vector $\Delta\eta_{N-j+1}\boldsymbol{u}_j$ to $P_m^{(2)}$ in Algorithm \ref{alg:fill_eta_layered}.

For instance, the $P_m^{(2)}$ imbalance produced by the step-vector $0.1\boldsymbol{u}_2$ allocation in Table VII %\ref{tab:pm2_simple} 
in the Appendix is
\[ U = 0.1 \times (3-2) = 0.1, \]
while the $P_m^{(2)}$ imbalance resulting from  $0.1\boldsymbol{u}_2 + 0.05\boldsymbol{u}_1$ allocation in Table \ref{tab:pm2_general} is
\[ U = 0.1 \times (3-2) + 0.05 \times (3-1) = 0.2. \]

\subsection{Construction of $P_m^{(3)}$}

In this section, we develop an algorithm to obtain $P_m^{(3)}$ from $P_X^{(3)}$ so that we can correct the row-sum imbalance arising from the construction of $P_m^{(2)}$ in Algorithm \ref{alg:fill_eta_layered}.

\subsubsection{Regaining Row-sum Invariance}

 The procedure for constructing $P_m^{(3)}$ is presented in Algorithm \ref{alg:pm3_proportional_fill} below. We first observe that, in the $P_X$ decomposition of Proposition \ref{prop:split_px}, the vector $P_X^{(3)}$ has at most $\tilde{K}$ positive entries at its tail, where $\tilde{K}$ is the number of components of $P_X$ that are greater or equal to $\tfrac{\alpha}{T}$. We then define $R$ as the sum of all elements of $P_X^{(3)}$, i.e.,
\begin{align}
    R = \sum_{i = 1}^N P_X^{(3)}(i) = \sum_{i = N-\tilde{K}+1}^N P_X^{(3)}(i).
\end{align}
If $R = 0$, we simply define all $P_m^{(3)}(x, \zeta)$ to be $0$. Hence, Algorithm \ref{alg:pm3_proportional_fill} addresses only the situation in which $R > 0$.

\begin{algorithm}[h]
\caption{$P_m^{(3)}$ Construction}
\label{alg:pm3_proportional_fill}
\begin{algorithmic}[1]
\State \textbf{Input:} $P_X^{(3)}$ with $\tilde{K}$ positive values; $R > 0$; Decomposition $P_X^{(2)} = \sum_{j=1}^K \Delta\eta_{N-j+1}\boldsymbol{u}_j$; $K$ number of nonzero elements in $P_X^{(2)}$; Total Imbalance $U$; Initialized $P_m^{(3)} =\mathbf{0}$, $m \in [1:T]$
\For{$j = K, \cdots, 1$}
    \For{$x = N-\tilde{K}+1, \ldots, N$}
        \State $P_m^{(3)}(x,\zeta) \!\gets \!P_m^{(3)}(x,\zeta) + \tfrac{\Delta\eta_{N-j+1}(T-j)}{|S(K) \setminus \cup_{l=N-j+1}^N \gamma^{-1}_l(m)|}\!\cdot\! \tfrac{P_X^{(3)}(x)}{R}$ , if $\zeta \in S(K) \setminus \cup_{l=N-j+1}^N \gamma^{-1}_l(m)  $ \label{alg-eq:allocate_imbalance}
    \EndFor
\EndFor
\State $P_m^{(3)}(x, \mathbf{0}_N) \gets P_X^{(3)}(x) \times (1-\frac{U}{R}), \forall x$
\end{algorithmic}
\end{algorithm}

Intuitively, this step uses the remaining entries of $P_X^{(3)}$ to offset the row-sum imbalance introduced when constructing $P_m^{(2)}$. Thus, for each $m \in [1:T]$, we must allocate the precise imbalance amount $U_m(\zeta)$ to the $\zeta$ row in $P_m^{(3)}$ using the available mass in $P_X^{(3)}$. To justify this, we first need to show that the total sum of $P_X^{(3)}$ exceeds the total imbalance $U$, ensuring that $P_X^{(3)}$ contains enough mass to fully correct the imbalance. This leads us to the following proposition. 

\begin{prop} \label{prop:R_larger_than_U}
    Let $P_X^{(2)}$ and $P_X^{(3)}$ be the decomposition of $P_X$ in Proposition \ref{prop:split_px}. Then the sum of $P_X^{(3)}$, denoted by $R$, is always greater or equal to the total row-sum imbalance $U$ induced by Algorithm \ref{alg:fill_eta_layered}.
\end{prop}

The proof of Proposition \ref{prop:R_larger_than_U} is given in Appendix F. %\ref{sec:proof_R_larger_than_U}.
Recall Algorithm \ref{alg:fill_eta_layered} applies step-vector allocation layer by layer. Accordingly, we also distribute the imbalance values layer by layer, inserting these values $U_m(\zeta)$ into the row $P_m^{(3)}(\cdot, \zeta)$ for those keys $\zeta$. 
Algorithm \ref{alg:pm3_proportional_fill} assigns values only to keys $\zeta \in S(K) \setminus \cup_{l=N-j+1}^N \gamma^{-1}_l(m)$, since $U_m(\zeta)$ can be positive only for those $\zeta$, as established in Appendix E. %\ref{sec:prop_U_proof}.

\begin{prop} \label{prop:pm3_contruction}
    The construction of $P_m^{(3)}$ in Algorithm \ref{alg:pm3_proportional_fill} leads to $P_m = P_m^{(1)} + P_m^{(2)} + P_m^{(3)}$ satisfying the following properties: row-sum and column-sum invariance, $\tfrac{\alpha}{T}$-capped column sum and the $\alpha$-bounded total sum condition.
\end{prop}
Proposition \ref{prop:pm3_contruction} is proved in Appendix G. %\ref{sec:proof_pm3_prop}. 
By Proposition \ref{prop:construct_pm1}, Proposition \ref{prop:construct_pm2}, and Proposition \ref{prop:pm3_contruction}, we have $P_m = P_m^{(1)} + P_m^{(2)} + P_m^{(3)}$ satisfies all the conditions given in section \ref{sec:main-thm}, therefore, Theorem \ref{thm:forward} is proved.

\subsubsection{Example of $P_m$ }
Table VIII %\ref{tab:pm1+pm2+pm3} 
in Appendix gives the combined \(P_m = P_m^{(1)} + P_m^{(2)} + P_m^{(3)}\) for \(N = 4\), \(T = 3\), \(\alpha = 0.9\), and \(P_X = (0.05, 0.1, 0.25, 0.6)\). The entries corresponding to \(P_m^{(1)}\) are shown in black, those corresponding to \(P_m^{(2)}\) are shown in blue, and those corresponding to \(P_m^{(3)}\) are shown in red. Decomposing \(P_X\) yields \(P_X^{(1)} = (0.05, 0.1, 0.15, 0.15)\), \(P_X^{(2)} = (0, 0, 0.1, 0.15)\), and \(P_X^{(3)} = (0, 0, 0, 0.3)\), which are precisely the values used in the previous example Table \ref{tab:pm2_general} and Table VI %\ref{tab:pm1_example2} 
in the Appendix. 
It is seen the imbalance is indeed corrected via Algorithm \ref{alg:pm3_proportional_fill}.

\section{Construction B: A Pseudo-Token Approach} \label{sec:construction_B}

In Construction B, we build $P_m$ using a method similar to that described in Algorithm \ref{alg:iterative_construction}, but here additional “pseudo-tokens” are introduced into the original token set $\cX$. Again, assume a non-decreasing token order on the distribution $P_X$. 
\revision{Similar to Construction A, we use the notion of $T$-hot representability to assign mass to pairs $(x,\zeta) \in \gamma^{-1}(m)$. Rather than decomposing $P_X$, we add pseudo-tokens to form a $T$-hot representable vector in an expanded token space, build the associated joint distributions, and finally translate the pseudo-tokens back into the original token space.}

\subsection{Construction in Augmented Space}
We apply the following lemma to construct an extended vector $P'_X$ from a given $P_X$.
\begin{lemma} \label{lemma:extend_PX}
    Let \(P_X\) be a non-decreasing distribution of length \(N\). Then there exists an integer \(n \ge 0\) such that we can construct an extended non-negative vector \(P'_X\) of length \(N + n\) that is \(T\)-hot representable. Specifically, when \(n = 0\), we have \(P'_X = \min \bigl(\tfrac{\alpha}{T}, P_X\bigr)\). When \(n > 0\), the components of \(P'_X\) add up to \(1\), i.e. \(P'_X\) itself forms a probability distribution with extra $n$ pseudo-tokens. 
\end{lemma}
\begin{proof}
    We define $\boldsymbol{a}$ and $\boldsymbol{r}$ as
    \begin{align}
        \boldsymbol{a}(x) = \min\!\left(\frac{\alpha}{T}, P_X(x)\right), \; \boldsymbol{r}(x) = (P_X - \boldsymbol{a})(x) = \max\!\left(P_X(x) - \frac{\alpha}{T}, 0\right),
    \end{align}
    where $\boldsymbol{a}$ is same as what we defined in the proof of Proposition~\ref{prop:split_px}, and $\boldsymbol{r}$ is the same definition as $P_X^{(3)}$. We use $R$ to denote the sum of the vector $\boldsymbol{r}$, consistent with the definition of $R$ provided in Algorithm \ref{alg:pm3_proportional_fill}.
    
    If $\boldsymbol{a}$ itself is $T$-hot representable, then set $n = 0$ and $P'_X = \boldsymbol{a}$. Otherwise, we choose $n = \lceil \frac{R}{\boldsymbol{a}(N)} \rceil$. We then demonstrate that
    \begin{align}
        P'_X = (\boldsymbol{a}, \tfrac{R}{n}\mathbf{1}_n)
    \end{align}
    is $T$-hot representable, where $\mathbf{1}_n$ denotes the all-ones vector of length $n$. The sum of $P_X'$ is 
    \begin{align*}
        \sum_{i=1}^{N+n} P_X'(i) & = \sum_{i=1}^{N} \boldsymbol{a}(i) + \sum_{i=N+1}^{N+n} \frac{R}{n} = \sum_{i=1}^{N} \boldsymbol{a}(i) + R = \sum_{i=1}^{N} \boldsymbol{a}(i) + \sum_{i=1}^{N}[P_X(i) - \boldsymbol{a}(i)]  = 1.
    \end{align*}
    
    Clearly, we have $\tfrac{R}{n} \leq \tfrac{R}{\tfrac{R}{\boldsymbol{a}(N)}} = \boldsymbol{a}(N)$. Furthermore, because $P_X$ is a non-decreasing vector, the vector $\boldsymbol{a}$ is also non-decreasing. Hence, $\max_i P_X'(i) = \boldsymbol{a}(N)$.

    The vector $P_X'$ is $T$-hot representable since
    \begin{align}
        T \max_i P_X'(i) = T\boldsymbol{a}(N) \leq T \frac{\alpha}{T} = \alpha \leq 1 = \sum_{i=1}^{N+n} P_X'(i),
    \end{align}
    which satisfies the inequality \eqref{eq:T-hot-representable}.
\end{proof}

Lemma \ref{lemma:extend_PX} provides an extended vector $P'_X$ that is $T$-hot representable. We now construct  $P_m$ under two scenarios: case 1 with $n = 0$ and case 2 with $n > 0$. For case 1, where $n = 0$, the lemma gives us $P'_X = \min(\tfrac{\alpha}{T}, P_X)$. In this setting, we adopt the same reduced key set $\cZ^N_T$, and then apply Algorithm \ref{alg:iterative_construction} to generate $P_m$ using $P'_X$ as the input vector. For the remaining mass represented by the vector $\boldsymbol{r} = P_X - P'_X$, we assign it to the dummy key $\zeta = \mathbf{0}_N$, namely
\begin{align}
    P_m(x, \mathbf{0}_N) = \boldsymbol{r}(x), \quad \forall x \in [1:N],\ \forall m \in [1:T].
\end{align} 

For case 2, where $n > 0$, $P'_X$ is now a vector with length $N+n$ and there are $n$ pseudo-tokens other than the original token set $\cX = [1:N]$. We denote the extended token space as $\cX_E = [1:N+n]$ and we adopt the key set $\cZ^{N+n}_T$. Algorithm \ref{alg:construction_B} summarizes the construction for both cases. 
\begin{prop} \label{prop:construct_big_Pm}
    Let $P_X$ be a non-decreasing distribution, and let $P'_X$ be the extended vector of length $N+n$ (with $n \geq 0$) constructed according to Lemma \ref{lemma:extend_PX}. The construction of $P_m$ in Algorithm \ref{alg:construction_B} is of dimension $|\cX|\times|\cZ^{N+n}_T|$ that satisfies the following properties: column-sum invariance, row-sum invariance, $\tfrac{\alpha}{T}$-capped column sum, and $\alpha$-bounded total sum.
\end{prop}

We prove Proposition \ref{prop:construct_big_Pm} in Appendix H. %\ref{sec:proof_big_Pm}. 
When $n > 0 $ in Algorithm \ref{alg:construction_B}, Algorithm \ref{alg:iterative_construction} with the input $P'_X$ leads to the extended joint distribution $P'_m$, which is with dimension $|\cX_E| \times |\cZ^{N+n}_T|$. However, the pseudo-tokens are not actual tokens in the generative model, and thus cannot be directly used. An additional step in Algorithm \ref{alg:construction_B} maps the values from the pseudo-token columns back into the original token space $\cX$, namely
\begin{align}
    P_m(x, \zeta) \gets P_m(x, \zeta) + \frac{\boldsymbol{r}(x)}{R} \left( \sum_{i=N+1}^{N+n} P'_m (i, \zeta) \right), \quad \forall x \in [1:N], \; \forall \zeta \in \cZ^{N+n}_T.
\end{align}
This mapping sums up the $P'_m$ values associated with the pseudo-tokens and redistributes them to the original tokens $x \in \cX$ with the proportion of $\boldsymbol{r}(x)$.

\begin{algorithm}[h]
\caption{$P_m$: Construction B}
\label{alg:construction_B}
\begin{algorithmic}[1]
\State \textbf{Input:} $P_X$ with length $N$; $P'_X$ with length $N+n$; Key set $\cZ^{N+n}_T$; Initialized $P'_m = \mathbf{0}$; Initialized $P_m = \mathbf{0}$, $m \in [1:T]$ 
\State Generate $P'_m$ by Algorithm \ref{alg:iterative_construction} with vector input $P'_X$
\State $P_m(x,\zeta) = P'_m(x,\zeta), \forall x \in [1:N], \; \forall \zeta \in \cZ^{N+n}_T$ 
\State  $\boldsymbol{r} = P_X - \min(\tfrac{\alpha}{T}, P_X)$ and $R = \sum_i\boldsymbol{r}(i)$
\If{$n = 0$}
\State $P_m(x, \mathbf{0}_N) \gets \boldsymbol{r}(x), ~\forall x \in [1:N]$  
\Else
    \State $P_m(x, \zeta) \gets  P_m(x, \zeta) + \frac{\boldsymbol{r}(x)}{R} \sum_{i=N+1}^{N+n} P'_m (i, \zeta), ~\forall x \in [1:N], \; \forall \zeta \in \cZ^{N+n}_T$
\EndIf
\end{algorithmic}
\end{algorithm}

Since $P_m$ satisfies all four conditions in proposition \ref{prop:construct_big_Pm}, we conclude that with the extended key set $\cZ^{N+n}_T$, the resultant $(P_m, \gamma)$ achieves the optimal solution. 

\subsection{Example of Construction B} 
Here we present two examples: one illustrates how to construct the extended $P_m'$, and the other demonstrates how to obtain $P_m$ from $P_m'$ by performing the final step of Algorithm \ref{alg:construction_B}, which maps the values in the pseudo-token columns back to the original columns. In this example, we choose $P_X = (0.1, 0.3, 0.6)$ with $N=3$ and $T=2$. For $\alpha = 0.9$, applying Lemma \ref{lemma:extend_PX} yields the extended distribution $P_X' = (0.1, 0.3, 0.45, 0.15)$ with $n=1$ . Consequently, the key set is taken to be $\cZ^4_2$, and the extended token space is $\cX_E = [1:4]$.

The $T$-hot representable form of $P'_X$ is 
\begin{align}
    P'_X & = (0.1, 0.3, 0.45, 0.15) \notag \\
    & = 0.3 \times (0,1,1,0) + 0.1 \times (0,0,1,1) + 0.05 \times (1,0,1,0) + 0.05 \times (1,0,0,1).
\end{align}
Algorithm \ref{alg:iterative_construction} with input $P_X'$ gives $P'_m$ in Table \ref{tab:construct_big_pm'}. 

For $P_m$ in Table \ref{tab:construct_pm_B}, the values in bold fonts correspond to the entries in the extended column $x = 4$ of the extended $P'_m$. By mapping them to their original token space $\cX = [1:3]$, we recover $P_m$.

\begin{table}[]
\centering
\begin{tabular}{c||cccc||cccc}

            & \multicolumn{4}{c||}{$P'_m (m=1)$}                                                            & \multicolumn{4}{c}{$P'_m (m=2)$}                                                                          \\ \hline
$\zeta$     & \multicolumn{1}{c|}{$x=1$}                                           & \multicolumn{1}{c|}{$x=2$}                    & \multicolumn{1}{c|}{$x=3$}                    & $x=4$                    & \multicolumn{1}{c|}{$x=1$}                    & \multicolumn{1}{c|}{$x=2$}                    & \multicolumn{1}{c|}{$x=3$}                    & $x=4$                    \\ \hline \hline
$(1,2,0,0)$ & \multicolumn{1}{c|}{\cellcolor{gray!15}} & \multicolumn{1}{c|}{}                         & \multicolumn{1}{c|}{}                         &                          & \multicolumn{1}{c|}{}                         & \multicolumn{1}{c|}{\cellcolor{gray!15}} & \multicolumn{1}{c|}{}                         &                          \\ \hline
$(2,1,0,0)$ & \multicolumn{1}{c|}{}                                                & \multicolumn{1}{c|}{\cellcolor{gray!15}} & \multicolumn{1}{c|}{}                         &                          & \multicolumn{1}{c|}{\cellcolor{gray!15}} & \multicolumn{1}{c|}{}                         & \multicolumn{1}{c|}{}                         &                          \\ \hline
$(0,1,2,0)$ & \multicolumn{1}{c|}{}                                                & \multicolumn{1}{c|}{\cellcolor{gray!15}$0.3$} & \multicolumn{1}{c|}{}                         &                          & \multicolumn{1}{c|}{}                         & \multicolumn{1}{c|}{}                         & \multicolumn{1}{c|}{\cellcolor{gray!15}$0.3$} &                          \\ \hline
$(0,2,1,0)$ & \multicolumn{1}{c|}{}                                                & \multicolumn{1}{c|}{}                         & \multicolumn{1}{c|}{\cellcolor{gray!15}$0.3$} &                          & \multicolumn{1}{c|}{}                         & \multicolumn{1}{c|}{\cellcolor{gray!15}$0.3$} & \multicolumn{1}{c|}{}                         &                          \\ \hline
$(1,0,2,0)$ & \multicolumn{1}{c|}{\cellcolor{gray!15}$0.05$}                        & \multicolumn{1}{c|}{}                         & \multicolumn{1}{c|}{}                         &                          & \multicolumn{1}{c|}{}                         & \multicolumn{1}{c|}{}                         & \multicolumn{1}{c|}{\cellcolor{gray!15}$0.05$} &                          \\ \hline
$(2,0,1,0)$ & \multicolumn{1}{c|}{}                                                & \multicolumn{1}{c|}{}                         & \multicolumn{1}{c|}{\cellcolor{gray!15}$0.05$} &                          & \multicolumn{1}{c|}{\cellcolor{gray!15}$0.05$} & \multicolumn{1}{c|}{}                         & \multicolumn{1}{c|}{}                         &                          \\ \hline
$(2,0,0,1)$ & \multicolumn{1}{c|}{}                                                & \multicolumn{1}{c|}{}                         & \multicolumn{1}{c|}{}                         & \cellcolor{gray!15}$0.05$& \multicolumn{1}{c|}{\cellcolor{gray!15}$0.05$} & \multicolumn{1}{c|}{}                         & \multicolumn{1}{c|}{}                         &                          \\ \hline
$(0,2,0,1)$ & \multicolumn{1}{c|}{}                                                & \multicolumn{1}{c|}{}                         & \multicolumn{1}{c|}{}                         & \cellcolor{gray!15} & \multicolumn{1}{c|}{}                         & \multicolumn{1}{c|}{\cellcolor{gray!15}} & \multicolumn{1}{c|}{}                         &                          \\ \hline
$(0,0,2,1)$ & \multicolumn{1}{c|}{}                                                & \multicolumn{1}{c|}{}                         & \multicolumn{1}{c|}{}                         & \cellcolor{gray!15}$0.1$& \multicolumn{1}{c|}{}                         & \multicolumn{1}{c|}{}                         & \multicolumn{1}{c|}{\cellcolor{gray!15}$0.1$} &                          \\ \hline
$(1,0,0,2)$ & \multicolumn{1}{c|}{\cellcolor{gray!15}$0.05$}                        & \multicolumn{1}{c|}{}                         & \multicolumn{1}{c|}{}                         &                          & \multicolumn{1}{c|}{}                         & \multicolumn{1}{c|}{}                         & \multicolumn{1}{c|}{}                         & \cellcolor{gray!15}$0.05$\\ \hline
$(0,1,0,2)$ & \multicolumn{1}{c|}{}                                                & \multicolumn{1}{c|}{\cellcolor{gray!15}} & \multicolumn{1}{c|}{}                         &                          & \multicolumn{1}{c|}{}                         & \multicolumn{1}{c|}{}                         & \multicolumn{1}{c|}{}                         & \cellcolor{gray!15} \\ \hline
$(0,0,1,2)$ & \multicolumn{1}{c|}{}                                                & \multicolumn{1}{c|}{}                         & \multicolumn{1}{c|}{\cellcolor{gray!15}$0.1$} &                          & \multicolumn{1}{c|}{}                         & \multicolumn{1}{c|}{}                         & \multicolumn{1}{c|}{}                         & \cellcolor{gray!15}$0.1$\\ \hline
$(0,0,0,0)$ & \multicolumn{1}{c|}{}                                                & \multicolumn{1}{c|}{}                         & \multicolumn{1}{c|}{}                         &                          & \multicolumn{1}{c|}{}                         & \multicolumn{1}{c|}{}                         & \multicolumn{1}{c|}{}                         &                          \\ 
\end{tabular}
\caption{Example of constructing extended $P'_m$ with $N=3$, $T= 2$, $n = 1$, $P_X'=(0.1, 0.3, 0.45, 0.15)$.}
\label{tab:construct_big_pm'}
\end{table}

\begin{table}[]
\centering
\begin{tabular}{c||ccc||ccc}

            & \multicolumn{3}{c||}{$P_m (m=1)$}                                                            & \multicolumn{3}{c}{$P_m (m=2)$}                                                                          \\ \hline
$\zeta$     & \multicolumn{1}{c|}{$x=1$}                                           & \multicolumn{1}{c|}{$x=2$}                    & $x=3$                    & \multicolumn{1}{c|}{$x=1$}                    & \multicolumn{1}{c|}{$x=2$}                    & $x=3$                    \\ \hline
$(1,2,0,0)$ & \multicolumn{1}{c|}{\cellcolor{gray!15}} & \multicolumn{1}{c|}{}                         &                          & \multicolumn{1}{c|}{}                         & \multicolumn{1}{c|}{\cellcolor{gray!15}} &                          \\ \hline
$(2,1,0,0)$ & \multicolumn{1}{c|}{}                                                & \multicolumn{1}{c|}{\cellcolor{gray!15}} &                          & \multicolumn{1}{c|}{\cellcolor{gray!15}} & \multicolumn{1}{c|}{}                         &                          \\ \hline
$(0,1,2,0)$ & \multicolumn{1}{c|}{}                                                & \multicolumn{1}{c|}{\cellcolor{gray!15}$0.3$} &                          & \multicolumn{1}{c|}{}                         & \multicolumn{1}{c|}{}                         & \cellcolor{gray!15}$0.3$\\ \hline
$(0,2,1,0)$ & \multicolumn{1}{c|}{}                                                & \multicolumn{1}{c|}{}                         & \cellcolor{gray!15}$0.3$& \multicolumn{1}{c|}{}                         & \multicolumn{1}{c|}{\cellcolor{gray!15}$0.3$} &                          \\ \hline
$(1,0,2,0)$ & \multicolumn{1}{c|}{\cellcolor{gray!15}$0.05$}                        & \multicolumn{1}{c|}{}                         &                          & \multicolumn{1}{c|}{}                         & \multicolumn{1}{c|}{}                         & \cellcolor{gray!15}$0.05$\\ \hline
$(2,0,1,0)$ & \multicolumn{1}{c|}{}                                                & \multicolumn{1}{c|}{}                         & \cellcolor{gray!15}$0.05$& \multicolumn{1}{c|}{\cellcolor{gray!15}$0.05$} & \multicolumn{1}{c|}{}                         &                          \\ \hline
$(2,0,0,1)$ & \multicolumn{1}{c|}{}                                                & \multicolumn{1}{c|}{}                         &                          $\mathbf{0.05}$& \multicolumn{1}{c|}{\cellcolor{gray!15}$0.05$} & \multicolumn{1}{c|}{}                         &                          \\ \hline
$(0,2,0,1)$ & \multicolumn{1}{c|}{}                                                & \multicolumn{1}{c|}{}                         &                          & \multicolumn{1}{c|}{}                         & \multicolumn{1}{c|}{\cellcolor{gray!15}} &                          \\ \hline
$(0,0,2,1)$ & \multicolumn{1}{c|}{}                                                & \multicolumn{1}{c|}{}                         & $\mathbf{0.1}$& \multicolumn{1}{c|}{}                         & \multicolumn{1}{c|}{}                         & \cellcolor{gray!15}$0.1$\\ \hline
$(1,0,0,2)$ & \multicolumn{1}{c|}{\cellcolor{gray!15}$0.05$}                        & \multicolumn{1}{c|}{}                         &                          & \multicolumn{1}{c|}{}                         & \multicolumn{1}{c|}{}                         &                          $\mathbf{0.05}$\\ \hline
$(0,1,0,2)$ & \multicolumn{1}{c|}{}                                                & \multicolumn{1}{c|}{\cellcolor{gray!15}} &                          & \multicolumn{1}{c|}{}                         & \multicolumn{1}{c|}{}                         &                          \\ \hline
$(0,0,1,2)$ & \multicolumn{1}{c|}{}                                                & \multicolumn{1}{c|}{}                         & \cellcolor{gray!15}$0.1$& \multicolumn{1}{c|}{}                         & \multicolumn{1}{c|}{}                         & $\mathbf{0.1}$\\ \hline
$(0,0,0,0)$ & \multicolumn{1}{c|}{}                                                & \multicolumn{1}{c|}{}                         &                          & \multicolumn{1}{c|}{}                         & \multicolumn{1}{c|}{}                         &                          \\ \hline
\end{tabular}
\caption{Example of constructing $P_m$ from $P'_m$ in Table \ref{tab:construct_big_pm'} with $P_X = (0.1, 0.3, 0.6)$.}
\label{tab:construct_pm_B}
\end{table}

\section{Discussion of the Constructions}

\subsection{Invalidity of the Construction by He et al.}
A purported construction of $(P_m,\gamma)$ was given in \cite{he2025distributional}  for the same problem under consideration. In their construction, the decoder $\gamma$ is any function for which both $\gamma(\cdot, \zeta)$ and $\gamma(x, \cdot)$ are bijective for every fixed $\zeta \in \cZ \setminus \{\mathbf{0}_N \}$ and $x \in \cX$, and the size of key set is chosen so that  $|\cZ| = |\cX|+1$. We will demonstrate, via a specific example, that this choice of decoder can not yield a valid solution to Theorem \ref{thm:forward} in general.

Consider the setting where $N=3$, $T=2$, $\alpha = 0.99$ and $P_X=(0.01,0.04,0.95)$. In this case, we choose the key set
$$\cZ_{bi} = \{(1,2,0),(0,1,2),(2,0,1),(0,0,0)\}.$$
It is straightforward to verify that, with this choice of key set $\cZ_{bi}$, the decoder $\gamma$ satisfies the required condition that both $\gamma(\cdot, \zeta)$ and $\gamma(x, \cdot)$ are bijective mappings for nonzero keys. Their matrix $P_m$ is shown in Table \ref{tab:construct_pm_he}. 

\begin{table}[]
\centering
\begin{tabular}{c||ccc||ccc}

            & \multicolumn{3}{c||}{$P_m (m=1)$}                                                            & \multicolumn{3}{c}{$P_m (m=2)$}                                                                          \\ \hline
$\zeta$     & \multicolumn{1}{c|}{$x=1$}                                           & \multicolumn{1}{c|}{$x=2$}                    & $x=3$                    & \multicolumn{1}{c|}{$x=1$}                    & \multicolumn{1}{c|}{$x=2$}                    & $x=3$                    \\ \hline
$(1,2,0)$ & \multicolumn{1}{c|}{\cellcolor{gray!15}$P_1(1,1)$} & \multicolumn{1}{c|}{$P_1(2,1)$}                         &   $P_1(3,1)$                       & \multicolumn{1}{c|}{$P_2(1,1)$}                         & \multicolumn{1}{c|}{\cellcolor{gray!15}$P_2(2,1)$} &  $P_2(3,1)$                        \\ \hline
$(0,1,2)$ & \multicolumn{1}{c|}{$P_1(1,2)$}                                                & \multicolumn{1}{c|}{\cellcolor{gray!15}$P_1(2,2)$} &          $P_1(3,2)$                & \multicolumn{1}{c|}{$P_2(1,2)$} & \multicolumn{1}{c|}{$P_2(2,2)$}                         &  $P_2(3,2)$\cellcolor{gray!15}                        \\ \hline
$(2,0,1)$ & \multicolumn{1}{c|}{$P_1(1,3)$}                                                & \multicolumn{1}{c|}{$P_1(2,3)$} &    \cellcolor{gray!15}$P_1(3,3)$                     & \multicolumn{1}{c|}{\cellcolor{gray!15}$P_2(1,3)$}                         & \multicolumn{1}{c|}{$P_2(2,3)$}                         & $P_2(3,3)$ \\ \hline
$(0,0,0)$ & \multicolumn{1}{c|}{$P_1(1,4)$}                                                & \multicolumn{1}{c|}{$P_1(2,4)$}                         &       $P_1(3,4)$                   & \multicolumn{1}{c|}{$P_2(1,4)$}                         & \multicolumn{1}{c|}{$P_2(2,4)$}                         &  $P_2(3,4)$                        \\ \hline
\end{tabular}
\caption{Construct the distribution $P_m$ using the He et al. method, given $P_X = (0.01, 0.04, 0.95)$.}
\label{tab:construct_pm_he}
\end{table}

We demonstrate that, with this particular choice of key set $\cZ_{bi}$, there does not exist a $P_m$ that can attain the optimal value $1-\sum_x \min(\tfrac{\alpha}{T}, P_X(x))$. \revision{To do so, we analyze the row-sum and column-sum constraints and derive a lower bound on the worst-case miss-detection probability; we show that this bound is greater than the optimal value.} A proof is provided in Appendix I. %\ref{sec:he_construction_proof}.

\subsection{Comparison Between Constructions A and B}
We give the two constructions of $P_m$ in Sections \ref{sec:construction_A} and \ref{sec:construction_B}. The primary distinction between these two constructions lies in how the decoder is defined. In the first construction, we use the key set $\cZ^N_T$ as the decoder, whereas in the second construction, we instead employ the key set $\cZ^{N+n}_T$ as the decoder. Consequently, the key set in the second construction may be significantly larger than in the first. The increase in the cardinality of the key set is given by
\begin{align}
    \frac{|\cZ^{N+n}_T|}{|\cZ^{N}_T|} \simeq \frac{(N+n)!(N-T)!}{N!(N+n-T)!}. 
\end{align}
However, if the size of the token set $N$ is large and the token distribution is relatively uniform, the value of $n$ becomes smaller, which in turn results in a smaller increase in the key size. Although the key set is significantly smaller in the first construction, Construction B is conceptually much more straightforward. 

\revision{In our proposed constructions, both A and B demand large key spaces when the token set $N$ is large. These expansive key spaces can also lead to significant computational and storage overhead when explicitly building the joint distributions \(P_m\), potentially requiring \(O(TN|\mathcal Z|)\) storage. While many entries in these distributions are often zero and sparsity can lower storage requirements, developing efficient large-scale implementations or approximations of the proposed constructions remains an open problem for future research.}

\revision{
\subsection{Sequence-Level Extensions and Limitations}
The proposed constructions can be extended to sequence generation in two ways. The first approach is to treat an entire token sequence of length \(L\) as a single element of the alphabet \(\mathcal{X}^L\), and directly apply our constructions to the joint sequence distribution \(P_{X^L}\). This preserves the theoretical optimality guarantee under the corresponding sequence-level formulation. However, the alphabet size increases from \(N\) to \(N^L\), making explicit construction and storage computationally impractical.}

\revision{
The second approach is to apply our constructions token by token during autoregressive generation. At each position \(i\), the conditional distribution \(P_{X_i|X^{i-1}}(\cdot|x^{i-1})\) can be treated as \(P_X\) in our formulation, and the corresponding joint distribution \(P_m\) can be constructed to generate the next token. This approach avoids explicitly enumerating the exponentially large sequence space. However, the optimality of each individual token-level construction does not necessarily imply sequence-level optimality, since the overall detection performance depends on how information across tokens is combined. A related constrained Markov decision process (CMDP) framework has been introduced for sequential LLM-based steganography, capturing the long-term impact of token generation on embedding efficiency \cite{huang2025relatively}. Although the optimization objectives and constraints differ, this framework could offer a helpful starting point for extending our constructions to sequence-level watermarking.}

\revision{We also note that the current formulation does not address robustness against adversarial modifications of the generated text. Incorporating such modifications would require additional attack-model assumptions and corresponding robustness constraints, potentially involving error-correcting codes or other sequence-level watermarking techniques.}

\section{Conclusion}
We consider the multi-bit generative watermarking under the worst-case false-alarm constraint, which was previously studied by  He et al. \cite{he2025distributional} . It is shown that the construction given in \cite{he2025distributional}  is invalid in general, and we provide two new optimal constructions. These new constructions, together with the converse result given in  \cite{he2025distributional} , fully characterize the optimal miss-detection probability.

It should be noted that the token set $\cX$ can also be viewed as the extended token sequence set, and the same result applies when the generative probability of the sequence is available. \revision{However, the number of possible sequences grows exponentially
with sequence length, which poses significant computational
and storage challenges when explicitly constructing the
corresponding joint distributions.
Alternatively, the proposed constructions can be applied
token by token during autoregressive generation, although
their optimality at the sequence level remains to be established.
}

In a private communication, the authors of \cite{he2024theoretically,he2025distributional} informed us that they recently proposed a new deterministic scheme, and kindly shared with us their preprint  \cite{he2026fundamental}. However, the scheme proposed there still induces a larger detection error and is thus not optimal, in contrast to the optimal constructions presented in this work. 

\bibliographystyle{IEEEtran}
\bibliography{watermark}

@inproceedings{huang2025relatively,
  title={Relatively-secure llm-based steganography via constrained markov decision processes},
  author={Huang, Yu-Shin and Tian, Chao and Narayanan, Krishna and Zheng, Lizhong},
  booktitle={2025 IEEE International Symposium on Information Theory (ISIT)},
  pages={1--6},
  year={2025},
  organization={IEEE}
}

@article{he2024theoretically,
  title={Theoretically grounded framework for llm watermarking: A distribution-adaptive approach},
  author={He, Haiyun and Liu, Yepeng and Wang, Ziqiao and Mao, Yongyi and Bu, Yuheng},
  journal={arXiv preprint arXiv:2410.02890},
  year={2024}
}

@inproceedings{kirchenbauer2023watermark,
  title={A watermark for large language models},
  author={Kirchenbauer, John and Geiping, Jonas and Wen, Yuxin and Katz, Jonathan and Miers, Ian and Goldstein, Tom},
  booktitle={International conference on machine learning},
  pages={17061--17084},
  year={2023},
  organization={PMLR}
}

@article{kuditipudi2023robust,
  title={Robust distortion-free watermarks for language models},
  author={Kuditipudi, Rohith and Thickstun, John and Hashimoto, Tatsunori and Liang, Percy},
  journal={arXiv preprint arXiv:2307.15593},
  year={2023}
}

@misc{he2026fundamental,
  author = {Haiyun He  and Yepeng Liu and Zhuoer Shen and Ziqiao Wang and Yongyi Mao and Yuheng Bu},
  title = {Fundamental Trade-Offs in Multi-Bit Watermarking of Stochastic Processes},
  date = {2026-04-02},
  howpublished = {Preprint via personal communication},
  year = {2026}
}

@article{cox2008digital,
  title={Digital watermarking},
  author={Cox, Ingemar J and Miller, Matthew L and Bloom, Jeffrey A and Fridrich, Jessica and Kalker, Ton},
  journal={Morgan Kaufmann Publishers},
  volume={54},
  number={56-59},
  pages={2},
  year={2008},
  publisher={Springer}
}

@article{liu2024survey,
  title={A survey of text watermarking in the era of large language models},
  author={Liu, Aiwei and Pan, Leyi and Lu, Yijian and Li, Jingjing and Hu, Xuming and Zhang, Xi and Wen, Lijie and King, Irwin and Xiong, Hui and Yu, Philip},
  journal={ACM Computing Surveys},
  volume={57},
  number={2},
  pages={1--36},
  year={2024},
  publisher={ACM New York, NY}
}

@article{yang2025watermarking,
  title={Watermarking for large language models: A survey},
  author={Yang, Zhiguang and Zhao, Gejian and Wu, Hanzhou},
  journal={Mathematics},
  volume={13},
  number={9},
  pages={1420},
  year={2025},
  publisher={MDPI}
}

@misc{aaronson2023watermarking,
  author       = {Scott Aaronson},
  title        = {Watermarking of Large Language Models},
  year         = {2023},
  month        = aug,
  day          = {17},
  howpublished = {\url{https://simons.berkeley.edu/talks/scott-aaronson-ut-austin-openai-2023-08-17}},
  note         = {Talk page, Simons Institute for the Theory of Computing. Accessed: 2026-03-28}
}

@article{zhao2023provable,
  title={Provable robust watermarking for ai-generated text},
  author={Zhao, Xuandong and Ananth, Prabhanjan and Li, Lei and Wang, Yu-Xiang},
  journal={arXiv preprint arXiv:2306.17439},
  year={2023}
}

@article{liu2024adaptive,
  title={Adaptive text watermark for large language models},
  author={Liu, Yepeng and Bu, Yuheng},
  journal={arXiv preprint arXiv:2401.13927},
  year={2024}
}

@article{hu2023unbiased,
  title={Unbiased watermark for large language models},
  author={Hu, Zhengmian and Chen, Lichang and Wu, Xidong and Wu, Yihan and Zhang, Hongyang and Huang, Heng},
  journal={arXiv preprint arXiv:2310.10669},
  year={2023}
}

@article{wu2023resilient,
  title={A resilient and accessible distribution-preserving watermark for large language models},
  author={Wu, Yihan and Hu, Zhengmian and Guo, Junfeng and Zhang, Hongyang and Huang, Heng},
  journal={arXiv preprint arXiv:2310.07710},
  year={2023}
}

@article{jiang2025stealthink,
  title={Stealthink: A multi-bit and stealthy watermark for large language models},
  author={Jiang, Ya and Wu, Chuxiong and Boroujeny, Massieh Kordi and Mark, Brian and Zeng, Kai},
  journal={arXiv preprint arXiv:2506.05502},
  year={2025}
}

@article{chao2024watermarking,
  title={Watermarking language models with error correcting codes},
  author={Chao, Patrick and Sun, Yan and Dobriban, Edgar and Hassani, Hamed},
  journal={arXiv preprint arXiv:2406.10281},
  year={2024}
}

@inproceedings{long2025optimized,
  title={Optimized Couplings for Watermarking Large Language Models},
  author={Long, Carol Xuan and Tsur, Dor and Verdun, Claudio Mayrink and Hsu, Hsiang and Permuter, Haim and Calmon, Flavio P},
  booktitle={2025 IEEE International Symposium on Information Theory (ISIT)},
  pages={1--6},
  year={2025},
  organization={IEEE}
}

@inproceedings{chen2025improved,
  title={Improved unbiased watermark for large language models},
  author={Chen, Ruibo and Wu, Yihan and Guo, Junfeng and Huang, Heng},
  booktitle={Proceedings of the 63rd Annual Meeting of the Association for Computational Linguistics (Volume 1: Long Papers)},
  pages={20587--20601},
  year={2025}
}

@inproceedings{zhu2024duwak,
  title={Duwak: Dual watermarks in large language models},
  author={Zhu, Chaoyi and Galjaard, Jeroen and Chen, Pin-Yu and Chen, Lydia},
  booktitle={Findings of the Association for Computational Linguistics: ACL 2024},
  pages={11416--11436},
  year={2024}
}

@article{wang2023towards,
  title={Towards codable watermarking for injecting multi-bits information to LLMs},
  author={Wang, Lean and Yang, Wenkai and Chen, Deli and Zhou, Hao and Lin, Yankai and Meng, Fandong and Zhou, Jie and Sun, Xu},
  journal={arXiv preprint arXiv:2307.15992},
  year={2023}
}

@article{takezawa2023necessary,
  title={Necessary and sufficient watermark for large language models},
  author={Takezawa, Yuki and Sato, Ryoma and Bao, Han and Niwa, Kenta and Yamada, Makoto},
  journal={arXiv preprint arXiv:2310.00833},
  year={2023}
}

@inproceedings{yoo2024advancing,
  title={Advancing beyond identification: Multi-bit watermark for large language models},
  author={Yoo, KiYoon and Ahn, Wonhyuk and Kwak, Nojun},
  booktitle={Proceedings of the 2024 Conference of the North American Chapter of the Association for Computational Linguistics: Human Language Technologies (Volume 1: Long Papers)},
  pages={4031--4055},
  year={2024}
}

@article{wouters2023optimizing,
  title={Optimizing watermarks for large language models},
  author={Wouters, Bram},
  journal={arXiv preprint arXiv:2312.17295},
  year={2023}
}

@inproceedings{christ2024undetectable,
  title={Undetectable watermarks for language models},
  author={Christ, Miranda and Gunn, Sam and Zamir, Or},
  booktitle={The Thirty Seventh Annual Conference on Learning Theory},
  pages={1125--1139},
  year={2024},
  organization={PMLR}
}

@article{boroujeny2024multi,
  title={Multi-bit distortion-free watermarking for large language models},
  author={Boroujeny, Massieh Kordi and Jiang, Ya and Zeng, Kai and Mark, Brian},
  journal={arXiv preprint arXiv:2402.16578},
  year={2024}
}

@article{dathathri2024scalable,
  title={Scalable watermarking for identifying large language model outputs},
  author={Dathathri, Sumanth and See, Abigail and Ghaisas, Sumedh and Huang, Po-Sen and McAdam, Rob and Welbl, Johannes and Bachani, Vandana and Kaskasoli, Alex and Stanforth, Robert and Matejovicova, Tatiana and others},
  journal={Nature},
  volume={634},
  number={8035},
  pages={818--823},
  year={2024},
  publisher={Nature Publishing Group UK London}
}

@inproceedings{moulin2000information,
  title={Information-theoretic analysis of watermarking},
  author={Moulin, Pierre and O'Sullivan, Joseph A},
  booktitle={2000 IEEE International Conference on Acoustics, Speech, and Signal Processing. Proceedings (Cat. No. 00CH37100)},
  volume={6},
  pages={3630--3633},
  year={2000},
  organization={IEEE}
}

@article{merhav2002random,
  title={On random coding error exponents of watermarking systems},
  author={Merhav, Neri},
  journal={IEEE Transactions on Information Theory},
  volume={46},
  number={2},
  pages={420--430},
  year={2002},
  publisher={IEEE}
}

@article{moulin2001role,
  title={The role of information theory in watermarking and its application to image watermarking},
  author={Moulin, Pierre},
  journal={Signal Processing},
  volume={81},
  number={6},
  pages={1121--1139},
  year={2001},
  publisher={Elsevier}
}

@article{steinberg2002identification,
  title={Identification in the presence of side information with application to watermarking},
  author={Steinberg, Yossef and Merhav, Neri},
  journal={IEEE Transactions on Information Theory},
  volume={47},
  number={4},
  pages={1410--1422},
  year={2002},
  publisher={IEEE}
}

@article{cohen2002gaussian,
  title={The Gaussian watermarking game},
  author={Cohen, Aaron S and Lapidoth, Amos},
  journal={IEEE transactions on Information Theory},
  volume={48},
  number={6},
  pages={1639--1667},
  year={2002},
  publisher={IEEE}
}

@article{kirchenbauer2023reliability,
  title={On the reliability of watermarks for large language models},
  author={Kirchenbauer, John and Geiping, Jonas and Wen, Yuxin and Shu, Manli and Saifullah, Khalid and Kong, Kezhi and Fernando, Kasun and Saha, Aniruddha and Goldblum, Micah and Goldstein, Tom},
  journal={arXiv preprint arXiv:2306.04634},
  year={2023}
}

@inproceedings{sun2023codemark,
  title={Codemark: Imperceptible watermarking for code datasets against neural code completion models},
  author={Sun, Zhensu and Du, Xiaoning and Song, Fu and Li, Li},
  booktitle={Proceedings of the 31st ACM joint European software engineering conference and symposium on the foundations of software engineering},
  pages={1561--1572},
  year={2023}
}

@article{block2025gaussmark,
  title={Gaussmark: A practical approach for structural watermarking of language models},
  author={Block, Adam and Sekhari, Ayush and Rakhlin, Alexander},
  journal={arXiv preprint arXiv:2501.13941},
  year={2025}
}

@article{zhao2025marktune,
  title={MarkTune: Improving the quality-detectability trade-off in open-weight llm watermarking},
  author={Zhao, Yizhou and Wu, Zhiwei Steven and Block, Adam},
  journal={arXiv preprint arXiv:2512.04044},
  year={2025}
}

@article{xu2024learning,
  title={Learning to watermark llm-generated text via reinforcement learning},
  author={Xu, Xiaojun and Yao, Yuanshun and Liu, Yang},
  journal={arXiv preprint arXiv:2403.10553},
  year={2024}
}

@article{gu2023learnability,
  title={On the learnability of watermarks for language models},
  author={Gu, Chenchen and Li, Xiang Lisa and Liang, Percy and Hashimoto, Tatsunori},
  journal={arXiv preprint arXiv:2312.04469},
  year={2023}
}

@article{liu2023unforgeable,
  title={An unforgeable publicly verifiable watermark for large language models},
  author={Liu, Aiwei and Pan, Leyi and Hu, Xuming and Li, Shu'ang and Wen, Lijie and King, Irwin and Yu, Philip S},
  journal={arXiv preprint arXiv:2307.16230},
  year={2023}
}

@inproceedings{an2026reinforcement,
  title={A reinforcement learning framework for robust and secure llm watermarking},
  author={An, Li and Liu, Yujian and Liu, Yepeng and Bu, Yuheng and Zhang, Yang and Chang, Shiyu},
  booktitle={Proceedings of the 19th Conference of the European Chapter of the Association for Computational Linguistics (Volume 1: Long Papers)},
  pages={7181--7198},
  year={2026}
}

@inproceedings{he2025distributional,
  title={Distributional information embedding: A framework for multi-bit watermarking},
  author={He, Haiyun and Liu, Yepeng and Wang, Ziqiao and Mao, Yongyi and Bu, Yuheng},
  booktitle={2025 Asia Pacific Workshop on Data Science and Information Theory (APWDSIT)},
  pages={1--6},
  year={2025},
  organization={IEEE}
}
\newpage
\appendix
\section{Correctness of Algorithm \ref{alg:T-hot-representable}} \label{sec:correctness_T-hot-representable}

To establish the correctness of Algorithm \ref{alg:T-hot-representable}, it suffices to show that in each step, the coefficient is non-negative and that the procedure terminates with $\boldsymbol{\nu} = \mathbf{0}$. We denote, in each iteration of Algorithm \ref{alg:T-hot-representable}, the before and after vectors  (of line \ref{alg-eq:update_v}), as $\boldsymbol{\nu}$ and $\boldsymbol{\nu}^+$, respectively.

\underline{\textbf{Claim:} At each iteration, $\boldsymbol{\nu}^{+}$ satisfied (\ref{eq:T-hot-representable}) and are non-negative.}
We prove by induction. First $\boldsymbol{a}$ clearly satisfies (\ref{eq:T-hot-representable}) and is non-negative. 
Assume it is true for $\boldsymbol{\nu}$, we first verify that $\boldsymbol{\nu}^{+}(j) \le \frac{1}{T}\sum_i \boldsymbol{\nu}^{+}(i)$ holds for every $j$. 
If $ j \in \supp(\boldsymbol{\omega})$, the algorithm updates $\boldsymbol{\nu}^{+}(j) = \boldsymbol{\nu}(j) - \lambda_{\boldsymbol{\omega}}\boldsymbol{\omega}(j) = \boldsymbol{\nu}(j) - \lambda_{\boldsymbol{\omega}}$, where 
\begin{align}
    \lambda_{\boldsymbol{\omega}} = \frac{1}{T}\min \left( \min_{j\in \supp(\boldsymbol{\omega})} \left( T\boldsymbol{\nu}(j)\right)  ,\min_{j\notin \supp(\boldsymbol{\omega})} \left( \sum^{N}_{i=1} \boldsymbol{\nu}(i)-T \boldsymbol{\nu}(j)\right) \right) . \label{eq:lambda_def_pf}
\end{align}
$\lambda_{\boldsymbol{\omega}}$ is non-negative, since the second term in (\ref{eq:lambda_def_pf}) is non-negative owing to the fact that $\boldsymbol{\nu}$ satisfies the inequality (\ref{eq:T-hot-representable}). It follows that for $ j \in \supp(\boldsymbol{\omega})$,
\begin{align}
    \boldsymbol{\nu}^{+}(j) = \boldsymbol{\nu}(j) - \lambda_{\boldsymbol{\omega}} \leq \frac{\sum_i \boldsymbol{\nu}(i)}{T} - \lambda_{\boldsymbol{\omega}} = \frac{1}{T} \left( \sum_i \boldsymbol{\nu}(i) - T\lambda_{\boldsymbol{\omega}}\right) = \frac{1}{T} \sum_i \boldsymbol{\nu}^+(i), \label{eq:at_in_A_pf1}
\end{align} 
where the inequality follows from the assumption that $\boldsymbol{\nu}$ satisfies inequality (\ref{eq:T-hot-representable}), and the last equation follows from 
\begin{align}
    \sum_i \boldsymbol{\nu}^+(i) = \sum_i\boldsymbol{\nu}(i) - \lambda_{\boldsymbol{\omega}}\boldsymbol{\omega}(i) = \sum_i\boldsymbol{\nu}(i) - \sum_i\lambda_{\boldsymbol{\omega}}\boldsymbol{\omega}(i) = \sum_i\boldsymbol{\nu}(i)- T\lambda{\boldsymbol{\omega}}. \label{eq:update_a_sum_is_sum_minus_theta} 
\end{align}

If $j \notin \supp(\boldsymbol{\omega})$, $\boldsymbol{\nu}^{+}(j)$ remains the same because $\boldsymbol{\nu}^{+}(j) = \boldsymbol{\nu}(j) - \lambda_{\boldsymbol{\omega}}\boldsymbol{\omega}(j) = \boldsymbol{\nu}(j)$. Therefore,
\begin{align}
    \boldsymbol{\nu}^{+}(j) = \boldsymbol{\nu}(j) \leq \frac{\sum_i \boldsymbol{\nu}(i)}{T} - \lambda_{\boldsymbol{\omega}} = \frac{1}{T} \sum_i \boldsymbol{\nu}^+(i) \label{eq:at_in_A_pf2},
\end{align}
where the inequality 
follows the definition of $\lambda_{\boldsymbol{\omega}}$, because (\ref{eq:lambda_def_pf}) gives
\begin{align}
     \lambda_{\boldsymbol{\omega}} \leq \frac{1}{T} \left( \sum_i \boldsymbol{\nu}(i) - T\boldsymbol{\nu}(j)\right), \forall j \notin \supp(\boldsymbol{\omega})~ \Rightarrow ~ \boldsymbol{\nu}(j) \leq \frac{\sum_i \boldsymbol{\nu}(i)}{T} - \lambda_{\boldsymbol{\omega}}, ~ \forall j \notin \supp(\boldsymbol{\omega}). \notag
\end{align}

Therefore, we have established that $\boldsymbol{\nu}^+(j) \leq \frac{1}{T} \sum_i \boldsymbol{\nu}^+(i)$ for all $j$, which implies that $\boldsymbol{\nu}^+$ satisfies (\ref{eq:T-hot-representable}). Moreover, from definition of $\lambda_{\boldsymbol{\omega}}$, we have $\lambda_{\boldsymbol{\omega}} \leq  \tfrac{1}{T}\min_{j \in \supp(\boldsymbol{\omega})}T\boldsymbol{\nu}(j)$. Therefore,   
\begin{align}
    \boldsymbol{\nu}^+(i)=\boldsymbol{\nu}(i) - \lambda_{\boldsymbol{\omega}} \geq \boldsymbol{\nu}(i) - \min_{j \in \supp(\boldsymbol{\omega})}\boldsymbol{\nu}(j) \geq 0, ~~\forall i \in \supp(\boldsymbol{\omega}). \label{eq:at_non-neg}
\end{align}
For $i\notin\supp(\boldsymbol{\omega})$, $\boldsymbol{\nu}^+(i)=\boldsymbol{\nu}(i)\geq 0$ by the assumption. The proof of the claim is complete.

\underline{\textbf{Claim:} $\boldsymbol{\nu}=\mathbf{0}_N$ at termination.}

We define the index set $I = \{i : \boldsymbol{\nu}(i) = 0 \text{ or } T\boldsymbol{\nu}(i) = \sum_j \boldsymbol{\nu}(j)\}$. The subset $I$ consists of all indices $i$ in $\boldsymbol{\nu}(i)$ for which  $\boldsymbol{\nu}(i)$ equals $0$ or where $\boldsymbol{\nu}(i)$ takes the value $\frac{1}{T}\sum_j \boldsymbol{\nu}(j)$.

We demonstrate that once an index $i$ belongs to $I$, it remains in $I$ for all subsequent iterations. Furthermore, as long as $|I| < N$, the set $I$ strictly grows with each iteration. More precisely, $ I \subset I^+$, which can be directly written as
\begin{align}
    \{ i : \boldsymbol{\nu}(i) = 0, \text{ or } T\boldsymbol{\nu}(i) = \sum_j \boldsymbol{\nu}(j) \} \subset \{ i : \boldsymbol{\nu}^+(i) = 0, \text{ or } T\boldsymbol{\nu}^+(i) = \sum_j \boldsymbol{\nu}^+(j) \}. \label{eq:I_increasing}
\end{align}
We first show if $i \in I$, then $i \in I^+$.
If $\boldsymbol{\nu}({i}) = 0$, then it follows that $\boldsymbol{\nu}^+({i}) = 0$, because at each iteration the component of $\boldsymbol{\nu}$ does not increase, and we have already established that $\boldsymbol{\nu}$ is nonnegative.  If $\boldsymbol{\nu}(i) = \frac{1}{T}\sum_j \boldsymbol{\nu}(j)$, the index $i$ will belong to the index set $A$ in the next iteration because $\boldsymbol{\nu}(i)$ is now the largest component of $\boldsymbol{\nu}$. This follows directly from the fact that any vector $\boldsymbol{\nu}$ satisfies (\ref{eq:T-hot-representable}), that is $\max_i \boldsymbol{\nu}(i) \leq\tfrac{1}{T} \sum_j \boldsymbol{\nu}(j)$. Thus, $i \in A = \supp(\boldsymbol{\omega})$, and  the update rule gives
\[
\boldsymbol{\nu}^{+}({i})
    = \boldsymbol{\nu}({i}) -\lambda_{\boldsymbol{\omega}}\boldsymbol{\omega}({i}) 
    = \frac{1}{T}\sum_j \boldsymbol{\nu}(j) - \lambda_{\boldsymbol{\omega}}
    = \frac{1}{T}\sum_j \boldsymbol{\nu}^{+}(j),
\]
where the last equality uses (\ref{eq:update_a_sum_is_sum_minus_theta}). Therefore, we conclude that $i \in I^+$. 

Next, suppose now $|I| < N$, we show that at each iteration, at least one index $i\! \notin \!I$ enters $\!I^+$. If \( \lambda_{\boldsymbol{\omega}}\!\!=\!\! \tfrac{1}{T}\min_{j \in \supp(\boldsymbol{\omega})}\!\!T\boldsymbol{\nu}(j)\), there exists $ i_{\star} \notin I, i_{\star} = \arg\min_{j \in \supp(\boldsymbol{\omega})} \boldsymbol{\nu}(j)$ such that 
\begin{align}
    \boldsymbol{\nu}^+({i_{\star}}) = \boldsymbol{\nu}({i_{\star}}) - \lambda_{\boldsymbol{\omega}}\boldsymbol{\omega}({i_{\star}}) = \boldsymbol{\nu}({i_{\star}}) - \min_{j \in \supp(\boldsymbol{\omega})} \boldsymbol{\nu}(j) = 0 \quad \Rightarrow {i_{\star}} \in I^+ .\label{eq:a_enters_0}
\end{align}

If $ \lambda_{\boldsymbol{\omega}} = \tfrac{1}{T}\min_{j \notin \supp(\boldsymbol{\omega})}(\sum_i \boldsymbol{\nu}(i) - T\boldsymbol{\nu}(j))$, we have
\begin{align}
    & \sum_i \boldsymbol{\nu}^+(i) = \sum_i \boldsymbol{\nu}(i) - \min_{j \notin \supp(\boldsymbol{\omega})}\left( \sum_i \boldsymbol{\nu}(i) - T\boldsymbol{\nu}(j)\right) = \max_{j \notin \supp(\boldsymbol{\omega})} T\boldsymbol{\nu}(j) \\
    & \Rightarrow \frac{\sum_i \boldsymbol{\nu}^+(i)}{T} = \max_{j \notin \supp(\boldsymbol{\omega})} \boldsymbol{\nu}(j) = \max_{j \notin \supp(\boldsymbol{\omega})} \boldsymbol{\nu}^+(j) \label{eq:a_enters_avg}
\end{align}
since $\boldsymbol{\nu}(j)$ will not be update if $j \notin \supp(\boldsymbol{\omega})$. Therefore, there exists $ j_{\star} \notin \supp(\boldsymbol{\omega}), \; {j_{\star}} \notin I$ such that $ \boldsymbol{\nu}^+({j_{\star}}) = \frac{\sum_i \boldsymbol{\nu}^+(i) }{T}$, which means ${j_{\star}} \in I^+$. Therefore, the set $I$ strictly grows in each iteration.

Finally, if $|I| = N$, then $\boldsymbol{\nu}$ have exactly $T$ indices $i$ such that $\boldsymbol{\nu}(i) = \frac{1}{T}\sum_j \boldsymbol{\nu}(j)$, and all others are $0$. Otherwise, $\boldsymbol{\nu}$ would not sum to $\sum_j \boldsymbol{\nu}(j)$. Therefore, the final step is to set $A = \{i:T\boldsymbol{\nu}(i) = \sum_j \boldsymbol{\nu}(j) \} = \supp(\boldsymbol{\omega})$, where
$\lambda_{\boldsymbol{\omega}} = \tfrac{1}{T}\sum_i \boldsymbol{\nu}(i)$.
With this choice, the update rule
$\boldsymbol{\nu}^+(i) = \boldsymbol{\nu}(i) - \lambda_{\boldsymbol{\omega}}\boldsymbol{\omega}(i) = 0,\ \forall i$
ensures that the algorithm terminates. This completes the proof.

\section{Proof of Proposition \ref{prop:split_px}}\label{appendix:prop3}
We first define a vector $\boldsymbol{a}$ derived from $P_X$:
    \begin{align}
        & \boldsymbol{a}(x) = \min \left( \frac{\alpha }{T}, P_{X}(x)\right) , \label{eq:def_of_a}
    \end{align}
    and then introduce $P_X^{(3)}$ as the residual part of $P_X$ after subtracting $\boldsymbol{a}$:
    \begin{align}
        P_X^{(3)}(x) = P_X(x) - \boldsymbol{a}(x) = \max \left( P_{X}(x) -\frac{\alpha }{T} ,\  0\right),
    \end{align}
    which satisfies the third point in the proposition.
     
    Consider the following two cases.
    \begin{itemize}
        \item Case 1: $\boldsymbol{a}$ is $T$-hot representable. In this case, we set 
            \begin{align}
                P_X^{(1)} = \boldsymbol{a},\qquad P_X^{(2)} = \mathbf{0}_N.
            \end{align} 
        \item Case 2: $\boldsymbol{a}$ is not $T$-hot representable.          There is a unique index $K \in [\tilde{K}:T-1]$ and a corresponding leveling value $y$ such that $P_X^{(1)}$ is $T$-hot representable and $P_X^{(2)}$ contains exactly $K$ positive entries at its tail. Specifically,
            \begin{align}
                & P_X^{(1)} = (\boldsymbol{a}(1), \dots, \boldsymbol{a}(N-K), y, \dots, y), \\
                & P_X^{(2)} = \boldsymbol{a} - P_X^{(1)} = (\mathbf{0}_{N-K}, \boldsymbol{a}(N-K+1)-y, \dots, \boldsymbol{a}(N)-y).
            \end{align}
            where $\tilde{K}$ denotes the number of elements in $P_X$ that are greater or equal to $\tfrac{\alpha}{T}$ and
            
            \begin{align*}
                y = \frac{1}{T-K} \sum_{i=1}^{N-K} \boldsymbol{a}(i), \quad  y & \geq \boldsymbol{a}({N-K}).
            \end{align*}
    \end{itemize}
We next show that $P_X^{(1)}$ and $P_X^{(2)}$ satisfy Properties 1 and 2. This is trivial for Case 1, and we can focus on Case 2. 
First observe that $\boldsymbol{a}$ is a nondecreasing vector because $P_X$ is non-decreasing. We decompose $P_X^{(1)}$ and $P_X^{(2)}$ from $\boldsymbol{a}$  via Algorithm \ref{alg:decompose_px1_px2} below. The idea of Algorithm \ref{alg:decompose_px1_px2} is to reduce the few largest components to $y$ so that the modified vector becomes $T$-hot representable with an equality. The algorithm terminates at the smallest index $k = K$ for which substituting the $K$ largest entries of $\boldsymbol{a}$ by $y$ yields a $T$-hot representable vector.

\begin{algorithm}[h]
\caption{$T$-hot representable extraction from $\boldsymbol{a}$}
\label{alg:decompose_px1_px2}
\begin{algorithmic}[1]
\State \textbf{Input:} Nondecreasing length $N$ vector $\boldsymbol{a}$; initialized $P_X^{(1)} = P_X^{(2)} = \mathbf{0}_N$; value $\tilde{K}$ and $T$
\For{$k = \tilde{K}:T-1$}
    \State $y \gets \tfrac{1}{T-k}\sum_{i=1}^{N-k}\boldsymbol{a}(i)$
    \If{$y \geq \boldsymbol{a}(N-k)$}
        \State \textbf{Break}
    \EndIf
\EndFor
\State $P_X^{(1)} \gets (\boldsymbol{a}(1), \dots, \boldsymbol{a}(N-k), y, \dots,y)$
\State $P_X^{(2)} \gets \boldsymbol{a}-P_X^{(1)}$ 
\end{algorithmic}
\end{algorithm}

Notice that the algorithm always terminates with a break on Line-5, because when $k=T-1$, we have
\begin{align}
y= \sum_{i=1}^{N-T+1}\boldsymbol{a}(i)\geq \boldsymbol{a}(N-T+1),
\end{align}
by the assumption that $N\geq T$. Therefore, we always have $y \geq \boldsymbol{a}(N-K)$. Moreover, the assignment of $y$ implies that, 
\begin{align}
    T y =\sum_{i=1}^{N-K}\boldsymbol{a}(i) + K y,  
    \label{eq:y_equation}
\end{align}
which is exactly (\ref{eqn:PX2}), and also confirms claimed Property 1 of the proposition.  

Next, we verify that $P_X^{(2)}$ is a non-negative vector and contains $K$ positive elements, which is  
 $\boldsymbol{a}(i) - y > 0$ for every $i \in [N-K+1:N]$. Because $\boldsymbol{a}$ is a nondecreasing vector, it suffices to prove that $\boldsymbol{a}(N-K+1) - y > 0$. 
 We first prove the inequality for the case when $K > \tilde{K}$. Suppose otherwise, i.e., $\boldsymbol{a}(N-K+1) \leq y$. Then it follows that $(T-K) \cdot \boldsymbol{a}(N-K+1) \leq (T-K)\cdot y$. 
By the definition of $y$, we have
\begin{align}
    & (T-K)\cdot\boldsymbol{a}(N-K+1) \leq (T-K)\cdot y = \sum_{i=1}^{N-K} \boldsymbol{a}(i) \\
    \Rightarrow& T\cdot \boldsymbol{a}(N-K+1) \leq \!\!\sum_{i=1}^{N-K} \boldsymbol{a}(i) + K\cdot \boldsymbol{a}(N-K+1) = \!\! \!\! \sum_{i=1}^{N-K+1} \boldsymbol{a}(i) + (K-1)\cdot \boldsymbol{a}(N-K+1) \label{eq:stop_at_k-1}
\end{align}
However, since the algorithm terminates at $k=K$, this implies that when $k=K-1$, 
\[
y = \tfrac{1}{T-(K-1)}\sum_{i=1}^{N-K +1} \boldsymbol{a}(i) < \boldsymbol{a}(N-K+1) \!\Rightarrow \!\! \sum_{i=1}^{N-K +1} \boldsymbol{a}(i) + (K-1)\boldsymbol{a}(N-K+1) < T \cdot \boldsymbol{a}(N-K+1),
\]
which contradicts (\ref{eq:stop_at_k-1}). It remains to consider the case when the algorithm terminates at $K = \tilde{K}$. We begin by explicitly writing  $\boldsymbol{a}$ as:
\[
\boldsymbol{a} = \bigl(\boldsymbol{a}(1), \dots, \boldsymbol{a}(N-\tilde{K}), \underbrace{\tfrac{\alpha}{T}, \dots, \tfrac{\alpha}{T}}_{\tilde{K}}\bigr). 
\]
To obtain a contradiction, suppose $y \geq \boldsymbol{a}(N-\tilde{K}+1) = \tfrac{\alpha}{T}$, then 
\begin{align}
    y = \frac{1}{T-\tilde{K}} \sum_{i=1}^{N-\tilde{K}} \boldsymbol{a}(i) \geq \frac{\alpha}{T} & \Rightarrow \sum_{i=1}^{N-\tilde{K}} \boldsymbol{a}(i) + \tilde{K}\frac{\alpha}{T} \geq T\cdot \frac{\alpha}{T}
    \Rightarrow \sum_{i=1}^{N} \boldsymbol{a}(i)  \geq T\cdot \boldsymbol{a}(N),
\end{align}
leading to a contradiction with the assumption that $\boldsymbol{a}$ is not $T$-hot representable. 

Hence, we can conclude that the suppositions are not true, and $\boldsymbol{a}(N-K+1) > y$ indeed holds.

Lastly, it is clear that the sum of the components equals $P_X$, i.e.,
\[ P_X^{(1)} + P_X^{(2)} + P_X^{(3)} = \boldsymbol{a} + P_X - \boldsymbol{a} = P_X. \]
The proof is complete.

\section{Proof of Proposition \ref{prop:construct_pm1}} \label{sec:proof_pm1}
In this section, we show that the construction $P_m^{(1)}$ obtained from Algorithm \ref{alg:iterative_construction} satisfies both $P_X^{(1)}$-column-sum and row-sum invariance. Moreover, we assign nonzero values to $P_m^{(1)}(x,\zeta)$ only for those $(x,\zeta)$ pairs that are decoded to message $m$ by the decoder $\gamma$ and those value adds up to $P_X^{(1)}$.

First, to verify $P_X^{(1)}$-column-sum invariance, we need to show that 
\begin{align}
    P_X^{(1)}(x) = \sum_{\zeta \in \cZ^N_T} P_m^{(1)}(x,\zeta), \; \forall x \in \cX. 
\end{align}
For each positive $\lambda_{\boldsymbol{\omega}}$ in Algorithm \ref{alg:iterative_construction}, we assign an amount $\tfrac{\lambda_{\boldsymbol{\omega}}}{(T-1)!}$ to all pairs $(x,\zeta)$ such that $\zeta \in \cZ_{\boldsymbol{\omega}}$ and $\zeta_x = m$, contributing to the joint distribution $P_m^{(1)}$. Consequently, for any fixed $x \in \cX$ and every $m \in [1:T]$,
\begin{align}
    \sum_{\zeta \in \cZ^N_T} P_m^{(1)}(x,\zeta) & =  \sum_{\zeta \in \cZ^N_T} \sum _{\boldsymbol{\omega} \in \Omega_T}\tfrac{\lambda_{\boldsymbol{\omega}}}{(T-1)!} \mathds{1}\{\zeta \in \cZ_{\boldsymbol{\omega}}, \zeta_x =m \} =  \sum_{\zeta \in \cZ^N_T} \sum _{\boldsymbol{\omega} \in \Omega_T}\tfrac{\lambda_{\boldsymbol{\omega}}{\boldsymbol{\omega}}(x)}{(T-1)!} \mathds{1}\{\zeta \in \cZ_{\boldsymbol{\omega}}, \zeta_x =m \}\label{eq:alg_pm1_column-sum-pf1}\\
    & = \sum_{\boldsymbol{\omega}\in \Omega_T}\tfrac{\lambda_{\boldsymbol{\omega}}{\boldsymbol{\omega}}(x)}{(T-1)!} \sum_{\zeta \in \cZ^N_T}\mathds{1}\{\zeta \in \cZ_{\boldsymbol{\omega}}, \zeta_x =m \} \\
    & = \sum_{\boldsymbol{\omega}\in \Omega_T}\tfrac{\lambda_{\boldsymbol{\omega}}{\boldsymbol{\omega}}(x)}{(T-1)!} (T-1)! = \sum_{\boldsymbol{\omega}\in \Omega_T} \lambda_{\boldsymbol{\omega}}\boldsymbol{\omega}(x) = P_X^{(1)}(x). \label{eq:alg_pm1_column-sum-pf3}
\end{align}

Second, to establish row-sum invariance, we must show that for any fixed $\zeta \in \cZ^N_T$,
\begin{align}
    \sum_{x \in \cX}P_m^{(1)}(x,\zeta) = \sum_{x \in \cX}P_{m'}^{(1)}(x,\zeta), \quad \forall\, m, m' \in [1:T],\; m\neq m'.
\end{align}

By construction, Algorithm \ref{alg:iterative_construction} updates the entry $(x,\zeta)$ of $P_m^{(1)}$ only when $\zeta_x = m$. Therefore, for any fixed $\zeta \in \cZ^N_T \setminus \{ \mathbf{0}_N\}$, there exists some $x_a$ with $\zeta_{x_a} = m$ and some $x_b$ with $\zeta_{x_b} = m'$, and we obtain
\begin{align}
    \sum_{x \in \cX} P^{(1)}_m(x, \zeta) & = P^{(1)}_m(x_a, \zeta), \label{eq:(a1)=(b1)-pf-goal-split1}\\
    \sum_{x \in \cX} P^{(1)}_{m'}(x, \zeta) & = P^{(1)}_{m'}(x_b, \zeta), \label{eq:(a1)=(b1)-pf-goal-split2}
\end{align}
for all $m \neq m'$. Since $\zeta$ is the same in both (\ref{eq:(a1)=(b1)-pf-goal-split1}) and (\ref{eq:(a1)=(b1)-pf-goal-split2}), Algorithm \ref{alg:iterative_construction} assigns them the same $\lambda_{\boldsymbol{\omega}}$ whenever $\zeta \in \cZ_{\boldsymbol{\omega}}$, so they each receive an identical increment of $\frac{\lambda_{\boldsymbol{\omega}}}{(T-1)!}$. Hence, the value of (\ref{eq:(a1)=(b1)-pf-goal-split1}) equals to (\ref{eq:(a1)=(b1)-pf-goal-split2}), establishing the row-sum invariance.

Finally, as shown in \eqref{eq:alg_pm1_column-sum-pf1}-\eqref{eq:alg_pm1_column-sum-pf3}, we obtain
\begin{align}
    \sum_{\zeta \in Z^N_T}P_m^{(1)}(x, \zeta) \mathds{1} \{ \gamma(x,\zeta) = m\} = P_X^{(1)}(x), \quad \forall x \in \cX. \label{eq:show_pm1_alpha/T}
\end{align}
We conclude that Algorithm \ref{alg:iterative_construction} assigns mass only to entries satisfying $(x,\zeta) \in \gamma^{-1}(m)$, and we present a construction example of $P_m^{(1)}$ on the next page.

\afterpage{
\begin{landscape}
\begin{table}[]
\resizebox{\linewidth}{!}{%
\begin{tabular}{c||c|c|c|c||c|c|c|c||c|c|c|c}
            & \multicolumn{4}{c||}{$P_1^{(1)} (m=1)$} & \multicolumn{4}{c||}{$P_2^{(1)} (m=2)$} & \multicolumn{4}{c}{$P_3^{(1)} (m=3)$} \\ \hline
$\zeta$     & $x=1$   & $x=2$   & $x=3$   & $x=4$   & $x=1$   & $x=2$   & $x=3$   & $x=4$   & $x=1$   & $x=2$   & $x=3$   & $x=4$   \\ \hline\hline
$(3,2,1,0)$ &         &         &    \cellcolor{gray!15}     &         &         &     \cellcolor{gray!15}    &         &         &   \cellcolor{gray!15}      &         &         &         \\\hline
$(2,3,1,0)$ &         &         &   \cellcolor{gray!15}      &         &  \cellcolor{gray!15}       &         &         &         &         &   \cellcolor{gray!15}      &         &         \\\hline
$(3,1,2,0)$ &         &     \cellcolor{gray!15}    &         &         &         &         &      \cellcolor{gray!15}   &        &     \cellcolor{gray!15}    &         &         &         \\\hline
$(1,3,2,0)$ &   \cellcolor{gray!15}      &         &         &         &         &         &     \cellcolor{gray!15}    &         &         &   \cellcolor{gray!15}      &         &         \\\hline
$(2,1,3,0)$ &         &   \cellcolor{gray!15}      &         &         &      \cellcolor{gray!15}   &         &         &         &         &         &      \cellcolor{gray!15}   &         \\\hline
$(1,2,3,0)$ &  \cellcolor{gray!15}       &         &         &         &         &      \cellcolor{gray!15}   &         &         &         &         &   \cellcolor{gray!15}      &         \\\noalign{\hrule height 1pt}
$(3,2,0,1)$ &         &         &         &    \cellcolor{gray!15}     &         &       \cellcolor{gray!15}  &         &         &     \cellcolor{gray!15}    &         &         &         \\\hline
$(2,3,0,1)$ &         &         &         &     \cellcolor{gray!15}    &   \cellcolor{gray!15}      &         &         &         &         &    \cellcolor{gray!15}     &         &         \\\hline
\textcolor{red}{$(3,0,2,1)$} &         &         &         &     \cellcolor{gray!15}\textcolor{red}{$\frac{\lambda_{(1,0,1,1)}}{(3-1)!} = \frac{0.05}{2}$}&         &         &      \cellcolor{gray!15}\textcolor{red}{$\frac{\lambda_{(1,0,1,1)}}{(3-1)!} = \frac{0.05}{2}$}&         &      \cellcolor{gray!15}\textcolor{red}{$\frac{\lambda_{(1,0,1,1)}}{(3-1)!} = \frac{0.05}{2}$}   &         &         &         \\\hline
\textcolor{Green}{$(0,3,2,1)$} &         &         &         &       \cellcolor{gray!15}\textcolor{Green}{$\frac{\lambda_{(0,1,1,1)}}{(3-1)!} = \frac{0.1}{2}$}&         &         &    \cellcolor{gray!15}\textcolor{Green}{$\frac{\lambda_{(0,1,1,1)}}{(3-1)!} = \frac{0.1}{2}$}&         &         &     \cellcolor{gray!15}\textcolor{Green}{$\frac{\lambda_{(0,1,1,1)}}{(3-1)!} = \frac{0.1}{2}$}    &         &         \\\hline
\textcolor{red}{$(2,0,3,1)$} &         &         &         &   \cellcolor{gray!15}\textcolor{red}{$\frac{\lambda_{(1,0,1,1)}}{(3-1)!} = \frac{0.05}{2}$}&  \cellcolor{gray!15}\textcolor{red}{$\frac{\lambda_{(1,0,1,1)}}{(3-1)!} = \frac{0.05}{2}$}       &         &         &         &         &         &    \cellcolor{gray!15}\textcolor{red}{$\frac{\lambda_{(1,0,1,1)}}{(3-1)!} = \frac{0.05}{2}$}&         \\\hline
\textcolor{Green}{$(0,2,3,1)$} &         &         &         &  \cellcolor{gray!15}\textcolor{Green}{$\frac{\lambda_{(0,1,1,1)}}{(3-1)!} = \frac{0.1}{2}$}&         &    \cellcolor{gray!15}\textcolor{Green}{$\frac{\lambda_{(0,1,1,1)}}{(3-1)!} = \frac{0.1}{2}$}     &         &         &         &         &    \cellcolor{gray!15}\textcolor{Green}{$\frac{\lambda_{(0,1,1,1)}}{(3-1)!} = \frac{0.1}{2}$}&         \\\noalign{\hrule height 1pt}
$(3,1,0,2)$ &         &    \cellcolor{gray!15}     &         &         &         &         &         &    \cellcolor{gray!15}     &   \cellcolor{gray!15}      &         &         &    \\\hline 
$(1,3,0,2)$ &    \cellcolor{gray!15}     &         &         &         &         &         &         &    \cellcolor{gray!15}     &         &   \cellcolor{gray!15}      &         &    \\\hline 
\textcolor{red}{$(3,0,1,2)$} &         &         &  \cellcolor{gray!15}\textcolor{red}{$\frac{\lambda_{(1,0,1,1)}}{(3-1)!} = \frac{0.05}{2}$}&         &         &         &         &   \cellcolor{gray!15}\textcolor{red}{$\frac{\lambda_{(1,0,1,1)}}{(3-1)!} = \frac{0.05}{2}$}&      \cellcolor{gray!15}\textcolor{red}{$\frac{\lambda_{(1,0,1,1)}}{(3-1)!} = \frac{0.05}{2}$}   &         &         &    \\\hline 
\textcolor{Green}{$(0,3,1,2)$} &         &         &  \cellcolor{gray!15}\textcolor{Green}{$\frac{\lambda_{(0,1,1,1)}}{(3-1)!} = \frac{0.1}{2}$}&         &         &         &         &    \cellcolor{gray!15}\textcolor{Green}{$\frac{\lambda_{(0,1,1,1)}}{(3-1)!} = \frac{0.1}{2}$}&         & \cellcolor{gray!15}\textcolor{Green}{$\frac{\lambda_{(0,1,1,1)}}{(3-1)!} = \frac{0.1}{2}$}        &         &    \\\hline 
\textcolor{red}{$(1,0,3,2)$} & \cellcolor{gray!15}\textcolor{red}{$\frac{\lambda_{(1,0,1,1)}}{(3-1)!} = \frac{0.05}{2}$}        &         &         &         &         &         &         & \cellcolor{gray!15}\textcolor{red}{$\frac{\lambda_{(1,0,1,1)}}{(3-1)!} = \frac{0.05}{2}$}&         &         &  \cellcolor{gray!15}\textcolor{red}{$\frac{\lambda_{(1,0,1,1)}}{(3-1)!} = \frac{0.05}{2}$}&    \\\hline 
\textcolor{Green}{$(0,1,3,2)$} &         &  \cellcolor{gray!15}\textcolor{Green}{$\frac{\lambda_{(0,1,1,1)}}{(3-1)!} = \frac{0.1}{2}$}       &         &         &         &         &         &\cellcolor{gray!15}\textcolor{Green}{$\frac{\lambda_{(0,1,1,1)}}{(3-1)!} = \frac{0.1}{2}$}&         &         &   \cellcolor{gray!15}\textcolor{Green}{$\frac{\lambda_{(0,1,1,1)}}{(3-1)!} = \frac{0.1}{2}$}&    \\\noalign{\hrule height 1pt}
$(2,1,0,3)$ &         &  \cellcolor{gray!15}       &         &         &      \cellcolor{gray!15}   &         &         &         &         &         &         &  \cellcolor{gray!15}  \\\hline 

$(1,2,0,3)$ &  \cellcolor{gray!15}       &         &         &         &         & \cellcolor{gray!15}        &         &         &         &         &         & \cellcolor{gray!15}   \\\hline

\textcolor{red}{$(2,0,1,3)$} &         &         &  \cellcolor{gray!15}\textcolor{red}{$\frac{\lambda_{(1,0,1,1)}}{(3-1)!} = \frac{0.05}{2}$}&         & \cellcolor{gray!15}\textcolor{red}{$\frac{\lambda_{(1,0,1,1)}}{(3-1)!} = \frac{0.05}{2}$}        &         &         &         &         &         &         & \cellcolor{gray!15}\textcolor{red}{$\frac{\lambda_{(1,0,1,1)}}{(3-1)!} = \frac{0.05}{2}$} \\\hline 

\textcolor{Green}{$(0,2,1,3)$} &         &         & \cellcolor{gray!15}\textcolor{Green}{$\frac{\lambda_{(0,1,1,1)}}{(3-1)!} = \frac{0.1}{2}$}&         &         &  \cellcolor{gray!15}\textcolor{Green}{$\frac{\lambda_{(0,1,1,1)}}{(3-1)!} = \frac{0.1}{2}$}       &         &         &         &         &         &\cellcolor{gray!15}\textcolor{Green}{$\frac{\lambda_{(0,1,1,1)}}{(3-1)!} = \frac{0.1}{2}$}  \\\hline 

\textcolor{red}{$(1,0,2,3)$} & \cellcolor{gray!15}\textcolor{red}{$\frac{\lambda_{(1,0,1,1)}}{(3-1)!} = \frac{0.05}{2}$}        &         &         &         &         &         &  \cellcolor{gray!15}\textcolor{red}{$\frac{\lambda_{(1,0,1,1)}}{(3-1)!} = \frac{0.05}{2}$}&         &         &         &         & \cellcolor{gray!15}\textcolor{red}{$\frac{\lambda_{(1,0,1,1)}}{(3-1)!} = \frac{0.05}{2}$}   \\\hline 

\textcolor{Green}{$(0,1,2,3)$} &         & \cellcolor{gray!15}\textcolor{Green}{$\frac{\lambda_{(0,1,1,1)}}{(3-1)!} = \frac{0.1}{2}$}        &         &         &         &         &   \cellcolor{gray!15}\textcolor{Green}{$\frac{\lambda_{(0,1,1,1)}}{(3-1)!} = \frac{0.1}{2}$}&         &         &         &         & \cellcolor{gray!15}\textcolor{Green}{$\frac{\lambda_{(0,1,1,1)}}{(3-1)!} = \frac{0.1}{2}$} \\\noalign{\hrule height 1pt}
$(0,0,0,0)$ &         &        &         &         &         &         &        &         &         &         &         & 
\end{tabular}}
\caption{Example of $P_m^{(1)}$ construction given $P_X^{(1)} = (0.05,0.1,0.15,0.15)$.}
\label{tab:pm1_example2}
\end{table}
\end{landscape}
}

\section{Proof of Proposition \ref{prop:construct_pm2}} \label{sec:proof_pm2_prop}

The cardinality of the set $\left\{\zeta: \zeta \in S(K) \cap \gamma_x^{-1}(m) \right\}$ in Algorithm \ref{alg:fill_eta_layered} can be determined as
\begin{align}
    \left\lvert \left\{\zeta \in  S(K) \cap \gamma_x^{-1}(m) \right\} \right\rvert 
    = \left( \begin{gathered}T-1\\ K-1\end{gathered} \right)  \times \frac{(N-K)!}{(N-T)!} \times (K-1)! := C_{N,T,K},
\end{align}
since, for any $x \in [N-j+1:N]$ and any $m \in [1:T]$, the set can be expressed as
\begin{align}
    & \left\{\zeta: \zeta \in  S(K) \cap \gamma_x^{-1}(m) \right\}
    = \left\{ \zeta: [N-K+1:N] \subset \supp(\zeta) \;\text{and}\; \zeta_x = m \right\} \notag \\
 & = \left\{\zeta: \zeta_i \in [1:T]\setminus \{m\},\; \forall i \in [N-K+1:N]\setminus \{x\},\;\text{and}\; \zeta_x = m \right\}.
\end{align}
Observe that this cardinality is independent of the specific choices of $x \in [N-j+1:N]$ and $m \in [1:T]$; it is identical for all such $x$ and $m$. Hence, we denote it by the constant $C_{N,T,K}$.

We show that Algorithm \ref{alg:fill_eta_layered}, which constructs $P^{(2)}_m$, preserves the $P_X^{(2)}$-column-sum invariance and that $P^{(1)}_m+P^{(2)}_m$ satisfies the $\tfrac{\alpha}{T}$-capped column sum condition. Let us first prove that
\begin{align}
    \sum_\zeta P_m^{(2)}(x, \zeta) = P_X^{(2)}(x) = \eta_x,\quad \forall x \in [N-K+1:N], \label{eq:goal_fill_layered_eta}
\end{align}
where $\eta_x$ denotes the positive entries of $P_X^{(2)}$ defined in (\ref{eq:express_eta_i_same_eta}). It suffices to verify this for $x \in [N-K+1:N]$, because we do not go through the entries with $x \in [1:N-K]$.

For any fixed $x' \in [N-K+1:N]$, there exists some $j' \in [1:K]$ such that $x' = N - j' + 1$. From algorithm \ref{alg:fill_eta_layered}, 
\begin{align}
    & \sum_{\zeta \in \cZ^N_T} P_m^{(2)}(x', \zeta) = \sum_{\zeta \in \cZ^N_T} \sum_{j=1}^{K} \sum_{x = N-j+1}^N \frac{\Delta \eta_{N-j+1}}{C_{N,T,K}} \mathds{1}\{x = x', ~ \zeta \in S(K) \cap \gamma^{-1}_{x'}(m) \}  \notag \\
    & =  \sum_{\zeta \in \cZ^N_T} \sum_{j=j'}^{K} \sum_{x = N-j+1}^N \frac{\Delta \eta_{N-j+1}}{C_{N,T,K}} \mathds{1}\{x = x', ~ \zeta \in S(K) \cap \gamma^{-1}_{x'}(m) \}  \label{eq:layered_eta_pf_0} \\ 
    & = \sum_{\zeta \in \cZ^N_T} \sum_{j=j'}^{K} \frac{\Delta \eta_{N-j+1}}{C_{N,T,K}}  \mathds{1}\{ \zeta \in S(K) \cap \gamma^{-1}_{x'}(m) \} =  \sum_{j=j'}^{K} \frac{\Delta \eta_{N-j+1}}{C_{N,T,K}}  \sum_{\zeta \in \cZ^N_T} \mathds{1}\{ \zeta \in S(K) \cap \gamma^{-1}_{x'}(m) \}\label{eq:layered_eta_pf_1}\\
    & = \sum_{j=j'}^{K} \Delta \eta_{N-j+1} = \eta_{N-j'+1} = \eta_{x'} \label{eq:layered_eta_pf_3} 
\end{align}
Equation \eqref{eq:layered_eta_pf_0} holds because, for every $j \in [1:j'-1]$, no $x \in [N-j+1:N]$ satisfies $x = x'$. Equation (\ref{eq:layered_eta_pf_1}) holds because there is exactly one $x \in [N-j+1:N]$ such that $x = x'$, and the second equality in (\ref{eq:layered_eta_pf_3}) follows immediately from the definition of $\Delta \eta_j$.

Next, we verify that $P^{(1)}_m + P^{(2)}_m$ satisfies the $\tfrac{\alpha}{T}$-capped column sum constraint, i.e., we show that
\begin{align}
    \sum_{\zeta \in \cZ^N_T} \bigl(P^{(1)}_m(x,\zeta) + P^{(2)}_m(x,\zeta)\bigr)\,\mathds{1}\{\gamma(x,\zeta) = m\} \;\ge\; \min\!\left(\frac{\alpha}{T},\,P_X(x)\right),\quad \forall m \in [1:T],\ \forall x \in \cX.
\end{align}
By Algorithm \ref{alg:fill_eta_layered}, each $P_m^{(2)}$ assigns mass only to pairs $(x,\zeta) \in \gamma^{-1}(m)$, and we have already shown that $P_m^{(2)}$ preserves the $P_X^{(2)}$-column-sum invariance. Hence,
\begin{align}
    \sum_{\zeta \in \cZ^N_T} P_m^{(2)}(x,\zeta)\,\mathds{1}\{\gamma(x,\zeta) = m\}= P_X^{(2)}(x), \quad \forall x \in \cX.
\end{align}
Combining this with the result established in \eqref{eq:show_pm1_alpha/T}, we obtain
\begin{align}
    & \sum_{\zeta \in \cZ^N_T} \bigl(P^{(1)}_m(x,\zeta) + P^{(2)}_m(x,\zeta)\bigr)\,\mathds{1}\{\gamma(x,\zeta) = m\} \notag \\
    = & P_X^{(1)}(x) + P_X^{(2)}(x)= P_X(x) - P_X^{(3)}(x)=\min\!\left(\frac{\alpha}{T},\,P_X(x)\right), \forall x \in \cX.
\end{align}

\section{Proof of Proposition \ref{prop:U_value}}\label{sec:prop_U_proof}

We first show that the imbalance will accumulate through the layer $j$ from $K$ to $1$ in Algorithm \ref{alg:fill_eta_layered}. Let us view $P_m^{(2)}$ as the accumulation of the step-vector allocations $P_{m,j}^{(2)}$, where $P_{m,j}^{(2)}$ denotes the allocation assigned in each layer $j$ of Algorithm \ref{alg:fill_eta_layered}, and $U_m^j(\zeta)$ represents the resulting row-sum imbalance induced by this allocation. Specifically,
\begin{align}
    P_m^{(2)} & = \sum_{j\in[1:K]} P_{m,j}^{(2)}, \\
    U_m^j(\zeta) & = \max_{\mu \in [1:T]}\left( \sum_{x \in \cX} P_{\mu,j}^{(2)}(x, \zeta) \right) - \sum_{x \in \cX} P_{m,j}^{(2)}(x, \zeta), \quad \forall m \in [1:T],\, \zeta \in \cZ^N_T. \label{eq:U^j_m_equation}
\end{align}

In each layer $j$, and for every $\zeta \in S(K)$, the row sum $\sum_{x \in \cX} P_{\mu,j}^{(2)}(x, \zeta)$ corresponding to message $\mu$ is either $0$ or $\tfrac{\Delta\eta_{N-j+1}}{C_{N,T,K}}$, with the allocation rule in Algorithm \ref{alg:fill_eta_layered}. Consequently, the first maximization term in \eqref{eq:U^j_m_equation} is always equal to $\tfrac{\Delta\eta_{N-j+1}}{C_{N,T,K}}$, since there is always a choice of message, namely $\mu \in \{\zeta_{N-j+1}, \dots, \zeta_{N}\}$, for which the row sums to $\tfrac{\Delta\eta_{N-j+1}}{C_{N,T,K}}$. Moreover, if $\zeta \in S(K) \cap \gamma_i^{-1}(m)$ for some $i \in [N-j+1:N]$, then the second term in \eqref{eq:U^j_m_equation} coincides with this same value $\tfrac{\Delta\eta_{N-j+1}}{C_{N,T,K}}$; otherwise, the second term is $0$. Consequently, $U^j_m(\zeta)$ can be non-zero only when $\zeta \in S(K) \setminus \bigcup_{i = N-j+1}^N \gamma_i^{-1}(m)$. 

Next, we need to show the following equality holds. For any fixed $\zeta \in S(K)$, 
\begin{align}
    \sum_{j = 1}^K \max_{\mu \in [1:T]}\left( \sum_{x \in \cX} P_{\mu,j}^{(2)}(x, \zeta) \right) = \max_{\mu \in [1:T]}\left( \sum_{j = 1}^K \sum_{x \in \cX} P_{\mu,j}^{(2)}(x, \zeta) \right). 
\end{align}
This equation holds only if there exists a $\mu^{\star} \in [1:T]$ such that, for this fixed $\mu^{\star}$, the term $\sum_x P_{\mu^{\star},j}^{(2)}(x, \zeta)$ is always maximal for every $j$. The index $\mu^{\star}$ is evidently equal to $\zeta_N$, i.e., the final component of the key $\zeta$, because for this choice of $\mu^{\star}$, the quantity $\sum_x P_{\mu^{\star},j}^{(2)}(x, \zeta)$ is always given by $\tfrac{\Delta\eta_{N-j+1}}{C_{N,T,K}}$. Therefore, we have
\begin{align}
    \sum_{j=1}^K U^j_m(\zeta) & = \sum_{j = 1}^K \max_{\mu \in [1:T]}\left( \sum_{x \in \cX} P_{\mu,j}^{(2)}(x, \zeta) \right) - \sum_{j=1}^K \sum_{x \in \cX} P_{m,j}^{(2)}(x, \zeta) \\
    & = \max_{\mu \in [1:T]}\left( \sum_{j = 1}^K \sum_{x \in \cX} P_{\mu,j}^{(2)}(x, \zeta) \right) - \sum_{j=1}^K \sum_{x \in \cX} P_{m,j}^{(2)}(x, \zeta) \\
    & = \max_{\mu \in [1:T]}\left(\sum_{x \in \cX} P_{\mu}^{(2)}(x, \zeta) \right) - \sum_{x \in \cX} P_{m}^{(2)}(x, \zeta) = U_m(\zeta).
\end{align}

We denote $U^j_m = \sum_{\zeta \in \cZ^N_T} U^j_m(\zeta)$, and proceed to compute this value: 
\begin{align}
    U^j_m & = \sum_{\zeta \in \cZ^N_T} U^j_m(\zeta) = \sum_{\zeta \in \cZ^N_T}  \frac{\Delta\eta_{N-j+1}}{C_{N,T,K}} \mathds{1} \{ \zeta \in S(K) \setminus \cup_{i = N-j+1}^N \gamma_i^{-1}(m) \} \\
    & = \frac{\Delta\eta_{N-j+1}}{C_{N,T,K}} \times  \left| S(K) \setminus \bigcup^{N}_{x=N-j+1} \left(S(K) \cap \gamma_x^{-1}(m) \right) \right| \label{eq:unbalanced_Uj-0} \\
    & =  \frac{\Delta\eta_{N-j+1}}{C_{N,T,K}}\times\left[ \underbrace{\left( \begin{gathered}T\\ K \end{gathered} \right)  \times K!\times \frac{(N-K)!}{(N-T)!} \  }_{\left| S(K)\right|  } \  -\  j \times \underbrace{\left( \begin{gathered}T-1\\ K-1\end{gathered} \right)  \times (K-1)!\times \frac{(N-K)!}{(N-T)!} }_{C_{N,T,K}  } \right]  \label{eq:unbalanced_Uj-1}  \\
    & =  \frac{\Delta\eta_{N-j+1}}{C_{N,T,K}}\times C_{N,T,K}(T-j) = \Delta\eta_{N-j+1}(T-j) \label{eq:unbalanced_Uj-2}.
\end{align}
The equality in (\ref{eq:unbalanced_Uj-1}) holds because the sets $\gamma_x^{-1}(m)$ are mutually disjoint for distinct values of $x$.

Last, we calculate the total row-sum imbalance value 
\begin{align}
    U_m = \sum_{j \in [1:K]} U^j_m = \sum_{j \in [1:K]} \Delta\eta_{N-j+1}(T-j) := U, \label{eq:total_U}
\end{align}
where we denote it as $U$ because it is the same value for all message $m$. This concludes the proof.  

\begin{table}[]
\resizebox{\linewidth}{!}{%
\begin{tabular}{c||c|c|c|c||c|c|c|c||c|c|c|c}
            & \multicolumn{4}{c||}{$P_1^{(2)} (m=1)$} & \multicolumn{4}{c||}{$P_2^{(2)} (m=2)$} & \multicolumn{4}{c}{$P_3^{(2)} (m=3)$} \\ \hline
$\zeta$     & $x=1$   & $x=2$   & $x=3$   & $x=4$   & $x=1$   & $x=2$   & $x=3$   & $x=4$   & $x=1$   & $x=2$   & $x=3$   & $x=4$   \\ \hline\hline
$(3,2,1,0)$ &         &         &    \cellcolor{gray!15}     &         &         &     \cellcolor{gray!15}    &         &         &   \cellcolor{gray!15}      &         &         &         \\\hline
$(2,3,1,0)$ &         &         &   \cellcolor{gray!15}      &         &  \cellcolor{gray!15}       &         &         &         &         &   \cellcolor{gray!15}      &         &         \\\hline
$(3,1,2,0)$ &         &     \cellcolor{gray!15}    &         &         &         &         &      \cellcolor{gray!15}   &        &     \cellcolor{gray!15}    &         &         &         \\\hline
$(1,3,2,0)$ &   \cellcolor{gray!15}      &         &         &         &         &         &     \cellcolor{gray!15}    &         &         &   \cellcolor{gray!15}      &         &         \\\hline
$(2,1,3,0)$ &         &   \cellcolor{gray!15}      &         &         &      \cellcolor{gray!15}   &         &         &         &         &         &      \cellcolor{gray!15}   &         \\\hline
$(1,2,3,0)$ &  \cellcolor{gray!15}       &         &         &         &         &      \cellcolor{gray!15}   &         &         &         &         &   \cellcolor{gray!15}      &         \\\noalign{\hrule height 1pt}
$(3,2,0,1)$ &         &         &         &    \cellcolor{gray!15}     &         &       \cellcolor{gray!15}  &         &         &     \cellcolor{gray!15}    &         &         &         \\\hline
$(2,3,0,1)$ &         &         &         &     \cellcolor{gray!15}    &   \cellcolor{gray!15}      &         &         &         &         &    \cellcolor{gray!15}     &         &         \\\hline
$(3,0,2,1)$ &         &         &         &     \cellcolor{gray!15}$\tfrac{0.1}{4}$&         &         &      \cellcolor{gray!15}$\tfrac{0.1}{4}$&         &      \cellcolor{gray!15}   &         &         &         \\\hline
$(0,3,2,1)$ &         &         &         &       \cellcolor{gray!15}$\tfrac{0.1}{4}$&         &         &    \cellcolor{gray!15}$\tfrac{0.1}{4}$&         &         &     \cellcolor{gray!15}    &         &         \\\hline
$(2,0,3,1)$ &         &         &         &   \cellcolor{gray!15}$\tfrac{0.1}{4}$&  \cellcolor{gray!15}       &         &         &         &         &         &    \cellcolor{gray!15}$\tfrac{0.1}{4}$&         \\\hline
$(0,2,3,1)$ &         &         &         &  \cellcolor{gray!15}$\tfrac{0.1}{4}$&         &    \cellcolor{gray!15}     &         &         &         &         &    \cellcolor{gray!15}$\tfrac{0.1}{4}$&         \\\noalign{\hrule height 1pt}
$(3,1,0,2)$ &         &    \cellcolor{gray!15}     &         &         &         &         &         &    \cellcolor{gray!15}     &   \cellcolor{gray!15}      &         &         &    \\\hline 
$(1,3,0,2)$ &    \cellcolor{gray!15}     &         &         &         &         &         &         &    \cellcolor{gray!15}     &         &   \cellcolor{gray!15}      &         &    \\\hline 
$(3,0,1,2)$ &         &         &  \cellcolor{gray!15}$\tfrac{0.1}{4}$&         &         &         &         &   \cellcolor{gray!15}$\tfrac{0.1}{4}$&      \cellcolor{gray!15}   &         &         &    \\\hline 
$(0,3,1,2)$ &         &         &  \cellcolor{gray!15}$\tfrac{0.1}{4}$&         &         &         &         &    \cellcolor{gray!15}$\tfrac{0.1}{4}$&         & \cellcolor{gray!15}        &         &    \\\hline 
$(1,0,3,2)$ & \cellcolor{gray!15}        &         &         &         &         &         &         & \cellcolor{gray!15}$\tfrac{0.1}{4}$&         &         &  \cellcolor{gray!15}$\tfrac{0.1}{4}$&    \\\hline 
$(0,1,3,2)$ &         &  \cellcolor{gray!15}       &         &         &         &         &         &\cellcolor{gray!15}$\tfrac{0.1}{4}$&         &         &   \cellcolor{gray!15}$\tfrac{0.1}{4}$&    \\\noalign{\hrule height 1pt}
$(2,1,0,3)$ &         &  \cellcolor{gray!15}       &         &         &      \cellcolor{gray!15}   &         &         &         &         &         &         &  \cellcolor{gray!15}  \\\hline 

$(1,2,0,3)$ &  \cellcolor{gray!15}       &         &         &         &         & \cellcolor{gray!15}        &         &         &         &         &         & \cellcolor{gray!15}   \\\hline 

$(2,0,1,3)$ &         &         &  \cellcolor{gray!15}$\tfrac{0.1}{4}$&         & \cellcolor{gray!15}        &         &         &         &         &         &         & \cellcolor{gray!15}$\tfrac{0.1}{4}$   \\\hline 

$(0,2,1,3)$ &         &         & \cellcolor{gray!15}$\tfrac{0.1}{4}$&         &         &  \cellcolor{gray!15}       &         &         &         &         &         &\cellcolor{gray!15}$\tfrac{0.1}{4}$   \\\hline 

$(1,0,2,3)$ & \cellcolor{gray!15}        &         &         &         &         &         &  \cellcolor{gray!15}$\tfrac{0.1}{4}$&         &         &         &         & \cellcolor{gray!15}$\tfrac{0.1}{4}$   \\\hline 

$(0,1,2,3)$ &         & \cellcolor{gray!15}        &         &         &         &         &   \cellcolor{gray!15}$\tfrac{0.1}{4}$&         &         &         &         & \cellcolor{gray!15}$\tfrac{0.1}{4}$ \\\noalign{\hrule height 1pt}
$(0,0,0,0)$ &         &        &         &         &         &         &        &         &         &         &         & 
\end{tabular}}
\caption{Example of simple case step-vector allocation: $P_X^{(2)} = (0,0,0.1,0.1)$.}
\label{tab:pm2_simple}
\end{table}

\section{Proof of Proposition \ref{prop:R_larger_than_U}} \label{sec:proof_R_larger_than_U}

We need to show that $R = \sum_{i=1}^N P_X^{(3)}$ is larger than the total imbalance \(U\). From the decomposition of \(P_X\), the vector \(P_X^{(2)}\) contains at most \(K \leq T-1\) positive entries, and in the proof of Proposition \ref{prop:split_px} we derived
\[
P_X^{(2)} = (\mathbf{0}_{N-K}, \boldsymbol{a}(N-K+1)-y, \dots, \boldsymbol{a}(N)-y).
\]
In that same proof, we introduced \(\tilde{K}\) as the number of entries of \(P_X\) that exceed \(\frac{\alpha}{T}\). From the definition of \(\boldsymbol{a}\) in (\ref{eq:def_of_a}), the last \(\tilde{K}\) entries of \(\boldsymbol{a}\) are all equal to \(\frac{\alpha}{T}\), that is,
\begin{align}
    \boldsymbol{a}(i) = \frac{\alpha}{T}, \quad \forall i \in [N-\tilde{K}+1:N], \label{eq:a_tail_is_alpha/T}
\end{align}
while for the remaining indices we have \(\boldsymbol{a}(i) = P_X(i)\) for all \(i \in [1:N-\tilde{K}]\).

Recall that in (\ref{eq:express_eta_i_same_eta}), we denoted the positive components of \(P_X^{(2)}\) by \(\eta_i\). Therefore, we can express
\begin{align}
    P_X^{(2)} & = (\mathbf{0}_{N-K}, P_X(N-K+1)-y, \dots, P_X(N-\tilde{K})-y, \tfrac{\alpha}{T}-y, \dots, \tfrac{\alpha}{T}-y) \notag \\
    & = (\mathbf{0}_{N-K}, \eta_{N-K+1}, \eta_{N-K+2}, \dots, \eta_N).
\end{align}
Consequently, for every index \(i \in [N-\tilde{K}+1:N]\), we obtain the same value \(\eta_i = \tfrac{\alpha}{T} - y\). In addition, for all element \(i \in [N-\tilde{K}+2:N]\), the increments satisfy \(\Delta\eta_i = \eta_i - \eta_{i-1} = 0\). 

Therefore, from (\ref{eq:total_U}), we further derive $U$ as 
\begin{align}
    U & = \sum_{j=1}^{\tilde{K}-1} \Delta\eta_{N-j+1}(T-j) + \sum_{j=\tilde{K}}^{K} \Delta\eta_{N-j+1}(T-j) = \sum_{j=\tilde{K}}^{K} \Delta\eta_{N-j+1}(T-j) \notag \\
    & = (\eta_{N-\tilde{K}+1} - \eta_{N-\tilde{K}}) (T-\tilde{K}) + (\eta_{N-\tilde{K}} - \eta_{N-\tilde{K}-1}) (T-\tilde{K}-1) + \cdots + (\eta_{N-{K}+1} - \eta_{N-K}) (T-K) \notag \\
    & = (T-\tilde{K})\eta_{N-\tilde{K}+1} - \sum_{j = N-K+1}^{N-\tilde{K}} \eta_j.
\end{align}

Next, we observe that $P_X^{(3)}$ contains only $\tilde{K}$ positive entries at its tail due to the way it is defined within the decomposition of $P_X$. Hence, we can establish that $R - U \geq 0$ as follows, when $U>0$:
\begin{align*}
    R-U & = \sum_{i=N-\tilde{K}+1}^{N} P_X^{(3)} - U \\
    & = \sum_{i=N-\tilde{K}+1}^{N} \left( P_X(i) - \frac{\alpha}{T}\right) - \left[  (T-\tilde{K})\eta_{N-\tilde{K}+1} - \sum_{j = N-K+1}^{N-\tilde{K}} \eta_j \right]  \\
    &=\sum_{i=N-\tilde{K}+1}^{N} P_X(i) - \tilde{K}\frac{\alpha}{T} - (T-\tilde{K})(\frac{\alpha}{T}-y) + \sum_{j = N-K+1}^{N-\tilde{K}}\left( P_{X}(j)-y\right)   \\
    & = \sum_{i=N-K+1}^{N} P_X(i) - T \frac{\alpha}{T}+(T-\tilde{K})y-(K-\tilde{K})y\\
    & = \sum_{i=N-K+1}^{N} P_X(i) - \alpha + (T-K)y \\
    &\stackrel{(a)}{=} \sum_{i=N-K+1}^{N} P_X(i) - \alpha + (T-K)\frac{1}{T-K}\sum_{i=1}^{N-K}P_X(i) \\
    & = 1 - \alpha \geq 0 \label{eq:U_smaller_than_sum_r},
\end{align*}
where in $(a)$ we apply the definition of $y$ when $U>0$. Since $\alpha\leq 1$, we have $R-U\geq 0$. When $U=0$, then clearly $R-U=R\geq 0$, and there is nothing to prove. This concludes the proof. 

\section{Proof of Proposition \ref{prop:pm3_contruction}} \label{sec:proof_pm3_prop}

In this section, we aim to establish the correctness of Algorithm \ref{alg:pm3_proportional_fill}, namely that all $P_m^{(3)}$ are non-negative, and that $P_m = P_m^{(1)} + P_m^{(2)} + P_m^{(3)}$ meets the required properties:  column-sum invariance, row-sum invariance, $\tfrac{\alpha}{T}$-capped column sum, and $\alpha$-bounded total sum. 

The only operation in Algorithm \ref{alg:pm3_proportional_fill} that could potentially produce a negative value is the final step $P_m^{(3)}(x,\mathbf{0}_N)$. We show in Appendix  \ref{sec:proof_R_larger_than_U} that $R = \sum_{x \in \cX}P_X^{(3)}(x) \geq U$, which in turn guarantees that $P_m^{(3)}(x,\mathbf{0}_N)$ is non-negative.

\underline{\textbf{$P_m$ satisfies column-sum invariance:}} \\
We first show $P_m^{(3)}$ satisfies $P_X^{(3)}$-column-sum invariance. Let $\tilde{K}$ denote the number of entries in $P_X$ that are larger than $\tfrac{\alpha}{T}$. It then suffices to check that, for any fixed $x \in [N-\tilde{K}+1 : N]$,
\begin{align}
    \sum_\zeta P_m^{(3)}(x,\zeta) = P_X^{(3)}(x),
\end{align}
because for all indices $x \in [1 : N-\tilde{K}]$ we have $P_X^{(3)}(x) = 0$, and we do not specify any values for $P_m^{(3)}(x,\cdot)$ for those $x$.

From algorithm \ref{alg:pm3_proportional_fill}, for any fixed $x' \in [N-\tilde{K}+1 : N]$, we have 
\begin{align}
    \sum_{\zeta \in \cZ^N_T} P_m^{(3)}(x',\zeta) & = \sum_{\zeta \in \cZ^N_T\setminus \{ \mathbf{0}_N\} } P_m^{(3)}(x',\zeta) + P_m^{(3)}(x',\mathbf{0}_N) \label{eq:rx_to_x_col_pf1}\\
    & = U\left(\frac{P_X^{(3)}(x')}{R} \right) + P_X^{(3)}(x')\left( 1- \frac{U}{R} \right)  \label{eq:rx_to_x_col_pf2}\\
    &= P_X^{(3)}(x'),
\end{align}
where (\ref{eq:rx_to_x_col_pf2}) can be seen as follows:
\begin{align}
    & \sum_{\zeta \in \cZ^N_T\setminus \{ \mathbf{0}_N\} } P_m^{(3)}(x',\zeta) \notag \\
    & =  \sum_{\zeta \in \cZ^N_T\setminus \{ \mathbf{0}\} } \sum_{j = 1}^K  \tfrac{\Delta\eta_{N-j+1}(T-j)}{|S(K)\setminus \cup_{l = N-j+1}^N \gamma_l^{-1}(m)|} \times \tfrac{P_X^{(3)}(x')}{R} \times \mathds{1} {\{ ~\zeta \in S(K)\setminus \cup_{l = N-j+1}^N \gamma_l^{-1}(m)\}} \label{eq:sum_xi_alg_3_w/o_0-0}\\
    & =   \tfrac{P_X^{(3)}(x')}{R} \sum_{j = 1}^K \Delta\eta_{N-j+1}(T-j)  \tfrac{1}{|S(K)\setminus \cup_{l = N-j+1}^N \gamma_l^{-1}(m)|} \times \sum_{\zeta \in \cZ^N_T\setminus \{ \mathbf{0}\} } \mathds{1} {\{\zeta \in S(K)\setminus \cup_{l = N-j+1}^N \gamma_l^{-1}(m)\}} \label{eq:sum_xi_alg_3_w/o_0-1}\\
    & = \tfrac{P_X^{(3)}(x')}{R} \sum_{j = 1}^K \Delta\eta_{N-j+1}(T-j)  =  \tfrac{P_X^{(3)}(x')}{R} U, \label{eq:sum_xi_alg_3_w/o_0-3}
\end{align}
where \eqref{eq:sum_xi_alg_3_w/o_0-0} is obtained from the allocation in Algorithm \ref{alg:pm3_proportional_fill}, and \eqref{eq:sum_xi_alg_3_w/o_0-1} is obtained by reordering the terms. The first equality in \eqref{eq:sum_xi_alg_3_w/o_0-3} holds because the last summation in \eqref{eq:sum_xi_alg_3_w/o_0-1} counts $\zeta$ elements in the set $\{S(K)\setminus \cup_{l = N-j+1}^N \gamma_l^{-1}(m) \}$, which exactly cancels the preceding fractional term.

Since we have already established that each $P_m^{(i)}$ satisfies the $P_X^{(i)}$-column-sum invariance, it follows that
\begin{align}
    \sum_{\zeta \in \cZ^N_T} P_m(x, \zeta) 
    &= \sum_{\zeta \in \cZ^N_T} \big(P_m^{(1)} + P_m^{(2)} + P_m^{(3)}\big)(x, \zeta) \notag \\
    &= P_X^{(1)}(x) + P_X^{(2)}(x) + P_X^{(3)}(x)= P_X(x), \quad \forall x \in \cX.
\end{align}
Hence, we have shown that $P_m$ satisfies the column-sum invariance condition.

\underline{\textbf{$P_m$ satisfies row-sum invariance:}} \\
Next, we show that $P_m^{(2)}+P_m^{(3)}$ satisfies row-sum invariance condition. It suffices to show that the row sum of $P_m^{(2)}+P_m^{(3)}$ is independent to $m \in [1:T]$ value. 
In algorithm \ref{alg:fill_eta_layered}, for an arbitrary $i \in \cX = [1:N]$, and fixed $m \in [1:T]$ and $\zeta \in \cZ^N_T \setminus \{\mathbf{0}_N \}$,
\begin{align}
    P^{(2)}_m(i, \zeta)  & = \sum_{j=1}^{K} \sum_{x=N-j+1}^N \frac{\Delta\eta_{N-j+1}}{C_{N,T,K}}\mathds{1}{\{x = i, \zeta \in S(K) \cap \gamma_{i}^{-1}(m)\}} \notag \\
    & = \left\{ \begin{array}{ll}
         0 &, i \in [1:N-K]  \\
         \sum_{j=j'}^{K}  \frac{\Delta\eta_{N-j+1}}{C_{N,T,K}}\mathds{1}{\{ \zeta \in S(K) \cap \gamma_{i}^{-1}(m)\}} & , i=N-j'+1\in [N-K+1:N]
    \end{array}\right. \\
    \Rightarrow & \sum_{i \in \cX} P^{(2)}_m(i, \zeta) = \sum_{i=N-K+1}^N \sum_{j=N-i+1}^{K}  \frac{\Delta\eta_{N-j+1}}{C_{N,T,K}}\mathds{1}{\{ \zeta \in S(K) \cap \gamma_{i}^{-1}(m)\}} .\label{eq:sum_x_p2_first_appear}
\end{align}

In algorithm \ref{alg:pm3_proportional_fill}, for an arbitrary $i \in \cX = [1:N]$, and fixed $m \in [1:T]$ and $\zeta \in \cZ^N_T\setminus \{\mathbf{0}_N \} $,
\begin{align}
    & P^{(3)}_m(i, \zeta) =\sum_{j=1}^K \sum_{x=N-\tilde{K}+1}^N \tfrac{\Delta\eta_{N-j+1}(T-j)}{|S(K) \setminus \cup_{l= N-j+1}^N \gamma^{-1}_l(m)|} \! \times \! \tfrac{P_X^{(3)}(i)}{R}\mathds{1} {\{ x = i, \zeta \in S(K) \setminus \cup_{l= N-j+1}^N \gamma^{-1}_l(m)\}} \notag \\
    & = \left\{ \begin{array}{ll}
         0 &, i \in [1:N-\tilde{K}]  \\
         \frac{P_X^{(3)}(i)}{R} \sum_{j=1}^K \frac{\Delta\eta_{N-j+1}(T-j)}{|S(K) \setminus \cup_{l= N-j+1}^N \gamma^{-1}_l(m)|}\mathds{1}{\zeta \in \{S(K) \setminus \cup_{l= N-j+1}^N \gamma^{-1}_l(m)\}} & , i\in [N-\tilde{K}+1:N]
    \end{array}\right.
\end{align}
This leads to 
\begin{align}
    \sum_{i \in \cX} P^{(3)}_m(i, \zeta) & = \underbrace{\sum_{i=N-\tilde{K}+1}^N \frac{P_X^{(3)}(i)}{R}}_{=1} \sum_{j=1}^K \tfrac{\Delta\eta_{N-j+1}(T-j)}{|S(K) \setminus \cup_{l= N-j+1}^N \gamma^{-1}_l(m)|}\mathds{1}{\{ \zeta \in S(K) \setminus \cup_{l= N-j+1}^N \gamma^{-1}_l(m)\}} \notag \\
    & = \sum_{j=1}^K \frac{\Delta\eta_{N-j+1}(T-j)}{|S(K) \setminus \cup_{l= N-j+1}^N \gamma^{-1}_l(m)|}\mathds{1}{\{ \zeta \in S(K) \setminus \cup_{l= N-j+1}^N \gamma^{-1}_l(m)\}} \\
    &  = \sum_{j=1}^K \frac{\Delta\eta_{N-j+1}}{C_{N,T,K}}\mathds{1}{\{ \zeta \in S(K) \setminus \cup_{l= N-j+1}^N \gamma^{-1}_l(m)\}}, \label{eq:sum_x_p3}
\end{align}
where (\ref{eq:sum_x_p3}) follows from the relation
\begin{align}
    \frac{(T-j)}{|S(K) \setminus \cup_{l= N-j+1}^N \gamma^{-1}_l(m)|} = \frac{(T-j)}{C_{N,T,K}(T-j)} = \frac{1}{C_{N,T,K}},
\end{align}
which is established in (\ref{eq:unbalanced_Uj-1})–(\ref{eq:unbalanced_Uj-2}).

Therefore, from \eqref{eq:sum_x_p2_first_appear} and \eqref{eq:sum_x_p3}, we have
\begin{align}
     & \sum_{i \in \cX} P^{(2)}_m(i, \zeta) +\sum_{i \in \cX}P^{(3)}_m(i, \zeta) \notag \\
     = & \sum_{i=N-K+1}^N \sum_{j=N-i+1}^{K}  \frac{\Delta\eta_{N-j+1}}{C_{N,T,K}}\mathds{1}{\{ \zeta \in S(K) \cap \gamma_{i}^{-1}(m)\}} +  \notag \\ 
     & \qquad \qquad \qquad \qquad  \sum_{j=1}^K \frac{\Delta\eta_{N-j+1}}{C_{N,T,K}}\mathds{1}{\{ \zeta \in S(K) \setminus \cup_{l= N-j+1}^N \gamma^{-1}_l(m)\}} \label{eq:sum_x_p2_p3-1}\\
    = & \sum_{j=1}^K \frac{\Delta\eta_{N-j+1}}{C_{N,T,K}} \mathds{1} {\{\zeta \in S(K)\}}, \; \forall \zeta \in \cZ^N_T \setminus \{\mathbf{0}_N \} \label{eq:sum_x_p2_p3-2}
\end{align}
which can be seen as follows. When $\zeta \notin S(K)$, both $\sum_{i \in \cX} P^{(2)}_m(i, \zeta)$ and $\sum_{i \in \cX} P^{(3)}_m(i, \zeta)$ are equal to $0$. When $\zeta \in S(K)$, we distinguish cases according to whether $m$ appears in the last $K$ components of $\zeta$. Specifically, consider the case where $\zeta \in \gamma^{-1}_{x'}(m)$ for some $x' \in [N-K+1:N]$. Write $x' = N - j' + 1$ with $j' \in [1:K]$. Then the first part of \eqref{eq:sum_x_p2_p3-1} becomes
\begin{align}
     & \sum_{i=N-K+1}^N \sum_{j=N-i+1}^{K}  \frac{\Delta\eta_{N-j+1}}{C_{N,T,K}}\mathds{1}{\{ \zeta \in S(K) \cap \gamma_{i}^{-1}(m)\}} \notag \\
     & =  \sum_{j=j'}^{K}  \frac{\Delta\eta_{N-j+1}}{C_{N,T,K}}\mathds{1}{\{ \zeta \in S(K) \cap \gamma_{x'}^{-1}(m)\}} 
       = \sum_{j=j'}^{K}  \frac{\Delta\eta_{N-j+1}}{C_{N,T,K}}, \label{eq:zeta1_pm2_total}
\end{align}
where \eqref{eq:zeta1_pm2_total} follows because there is exactly one $i \in [N-K+1:N]$ such that $\zeta \in \gamma^{-1}_{i}(m)$, namely $i = x' = N - j' + 1$.
The second part of \eqref{eq:sum_x_p2_p3-1} can be derived as 
\begin{align}
    \sum_{j=1}^K \frac{\Delta\eta_{N-j+1}}{C_{N,T,K}}\mathds{1}{\{ \zeta \in S(K) \setminus \cup_{l= N-j+1}^N \gamma^{-1}_l(m)\}} = \sum_{j=1}^{j'-1} \frac{\Delta\eta_{N-j+1}}{C_{N,T,K}},
\end{align}
since clearly $\zeta$ will not be in the set $\{ S(K) \setminus \cup_{l= N-j+1}^N \gamma^{-1}_l(m)\}$ if $j \geq j'$. Therefore, we have 
\begin{align}
     \eqref{eq:sum_x_p2_p3-1} = \sum_{j=j'}^{K}  \frac{\Delta\eta_{N-j+1}}{C_{N,T,K}} + \sum_{j=1}^{j'-1}  \frac{\Delta\eta_{N-j+1}}{C_{N,T,K}} = \sum_{j=1}^{K}  \frac{\Delta\eta_{N-j+1}}{C_{N,T,K}}.
\end{align}

If $\zeta \in S(K)$ does not belong to any of the sets $\gamma_x^{-1}(m)$ for $x \in [N-K+1:N]$, then the first term in \eqref{eq:sum_x_p2_p3-1} vanishes, and we obtain
\begin{align}
    \eqref{eq:sum_x_p2_p3-1}
    = \sum_{j=1}^K \frac{\Delta\eta_{N-j+1}}{C_{N,T,K}}\mathds{1}{\{ \zeta \in S(K) \setminus \cup_{l= N-j+1}^N \gamma^{-1}_l(m)\}}
    = \sum_{j=1}^K \frac{\Delta\eta_{N-j+1}}{C_{N,T,K}}.
\end{align}

For $\zeta=\mathbf{0}_N$, it holds that $P_m^{(2)}(i,\mathbf{0}_N)=0$ and
$$
P_m^{(3)}(i,\mathbf{0}_N)=\sum_{i\in\cX} P_X^{(3)}(i)\bigl(1-\tfrac{U}{R}\bigr)=R\bigl(1-\tfrac{U}{R}\bigr)=R-U.
$$
Hence, $P_m^{(2)}(i,\mathbf{0}_N)+P_m^{(3)}(i,\mathbf{0}_N)$ is also independent of $m$.

Combining the statements above, the row sums of $P_m^{(2)}+P_m^{(3)}$ are independent of $m$, and the sum matrix is indeed row-sum invariant. Since $P_m^{(1)}$ also satisfies row-sum invariance in Proposition \ref{prop:construct_pm1}, it follows that $P_m = P_m^{(1)}+P_m^{(2)}+P_m^{(3)}$ satisfies row-sum invariance as well.

\underline{\textbf{$P_m$ satisfies $\tfrac{\alpha}{T}$-capped column sum condition:}} 

For each $m \in [1:T]$ and for all $x \in \cX$
\begin{align}
    & \sum_{\zeta \in \cZ^N_T}P_m(x,\zeta) \mathds{1} \{ \gamma(x,\zeta) = m\} \notag \\
    & =  \sum_{\zeta \in \cZ^N_T} \left( (P_m^{(1)}+P_m^{(2)})(x,\zeta) \mathds{1} \{ \gamma(x,\zeta) = m\} + P_m^{(3)}(x,\zeta) \mathds{1} \{ \gamma(x,\zeta) = m\} \right) \notag \\
    & = \min \left(\frac{\alpha}{T} , P_X(x)\right), 
\end{align}
since the $P_m^{(3)}$ construction in Algorithm \ref{alg:pm3_proportional_fill} does not assign any value to $P_m^{(3)}(x,\zeta)$ for pairs satisfying $\gamma(x,\zeta) = m$. 

\underline{\textbf{$P_m$ satisfies $\alpha$-bounded total sum condition:}} \\

For any $x \in [1:N]$, 
\begin{align}
    & \sum_{m \in [1:T]} \sum_{\zeta \in \cZ^N_T} P_{Z}(\zeta) \mathds{1} {\{\gamma(x,\zeta) = m\}} = \sum_{m \in [1:T]} \sum_{\zeta \in \cZ^N_T} \left( \sum_{i=1}^N P_m(i,\zeta) \right) \mathds{1} {\{\zeta \in \gamma^{-1}_x(m)\}} \notag\\
    & = \sum_{m=1}^T \sum_{\zeta \in \cZ^N_T} \left( \sum_{i=1}^N P_m^{(1)}(i,\zeta)+ P_m^{(2)}(i,\zeta) + P_m^{(3)}(i,\zeta) \right) \mathds{1} {\{\zeta \in \gamma^{-1}_x(m)\}} \notag\\
    & = \sum_{m=1}^T \!\! \left(\!\! \sum_{\zeta \in \cZ^N_T} \underbrace{\sum_{i=1}^N P_m^{(1)}(i,\zeta) \mathds{1} {\{\zeta\! \in \!\gamma^{-1}_x(m)\}} }_{(a)} \!+ \! \! \!\sum_{\zeta \in \cZ^N_T} \underbrace{ \!\!\biggl[\sum_{i=1}^N P_m^{(2)}(i,\zeta) + \sum_{i=1}^N P_m^{(3)}(i,\zeta) \biggr]}_{(b)}\! \!\mathds{1}{\{\zeta \in \gamma^{-1}_x(m)\}}\!\!\right) \notag\\
    & = \sum_{m=1}^T \left( \sum_{\zeta \in \cZ^N_T} P_m^{(1)}(x,\zeta) \mathds{1} {\{\zeta \in \gamma^{-1}_x(m)\}} + \sum_{\zeta \in \cZ^N_T} \biggl[\sum_{j=1}^K \frac{\Delta\eta_{N-j+1}}{C_{N,T,K}} \mathds{1} \{ \zeta \in S(K)\}\biggr] \mathds{1} {\{\zeta \in \gamma^{-1}_x(m)\}}\!\right) \notag\\
    & = \sum_{m=1}^T \left( P_X^{(1)}(x)+  \sum_{j=1}^K \frac{\Delta\eta_{N-j+1}}{C_{N,T,K}}\sum_{\zeta \in \cZ^N_T} \mathds{1} \{ \zeta \in S(K)\cap \gamma^{-1}_x(m)\} \right) \notag \\
    & \leq \sum_{m=1}^T\left( P_X^{(1)}(x)+ \sum_{j=1}^K\Delta\eta_{N-j+1} \right) = \sum_{m=1}^T P_X^{(1)}(x) + \eta_N = \sum_{m \in [1:T]}(P_X^{(1)}(x) + P_X^{(2)}(N)) \label{eq:Pm_alpha_bound_pf_1} \\
    & \leq \sum_{m=1}^T (P_X^{(1)}(N) + P_X^{(2)}(N) ) = T\min (\frac{\alpha}{T}, P_X(N)) \leq \alpha \notag.
\end{align}
Here, the value $(b)$ is obtained from (\ref{eq:sum_x_p2_p3-2}), and the value $(a)$ follows from the fact that $P_m^{(1)}(i,\zeta)$ is positive only when $x = i$. The inequality in \eqref{eq:Pm_alpha_bound_pf_1} follows because, for $x \in [1:N-K]$, we have $\bigl|\{ \zeta \in S(K)\cap \gamma^{-1}_x(m)\}\bigr|=C_{N,T,K}\frac{T-K}{N-K}\le C_{N,T,K}$, whereas for $x \in [N-K+1:N]$ the cardinality equals $C_{N,T,K}$.

\afterpage{
\begin{landscape}
\begin{table}[]
\resizebox{\linewidth}{!}{%
\begin{tabular}{c||c|c|c|c||c|c|c|c||c|c|c|c}
            & \multicolumn{4}{c||}{$P_1 (m=1)$} & \multicolumn{4}{c||}{$P_2 (m=2)$} & \multicolumn{4}{c}{$P_3 (m=3)$} \\ \hline
$\zeta$     & $x=1$   & $x=2$   & $x=3$   & $x=4$   & $x=1$   & $x=2$   & $x=3$   & $x=4$   & $x=1$   & $x=2$   & $x=3$   & $x=4$   \\ \hline\hline
$(3,2,1,0)$ &         &         &    \cellcolor{gray!15}     &         &         &     \cellcolor{gray!15}    &         &         &   \cellcolor{gray!15}      &         &         &         \\\hline
$(2,3,1,0)$ &         &         &   \cellcolor{gray!15}      &         &  \cellcolor{gray!15}       &         &         &         &         &   \cellcolor{gray!15}      &         &         \\\hline
$(3,1,2,0)$ &         &     \cellcolor{gray!15}    &         &         &         &       &  \cellcolor{gray!15}       &        &     \cellcolor{gray!15}    &         &         &         \\\hline
$(1,3,2,0)$ &   \cellcolor{gray!15}      &         &         &         &         &         &     \cellcolor{gray!15}    &         &         &   \cellcolor{gray!15}      &         &         \\\hline
$(2,1,3,0)$ &         &   \cellcolor{gray!15}      &         &         &      \cellcolor{gray!15}   &         &         &         &         &         &      \cellcolor{gray!15}   &         \\\hline
$(1,2,3,0)$ &  \cellcolor{gray!15}       &         &         &         &         &      \cellcolor{gray!15}   &         &         &         &         &   \cellcolor{gray!15}      &         \\\noalign{\hrule height 1pt}
$(3,2,0,1)$ &         &         &         &    \cellcolor{gray!15}     &         &       \cellcolor{gray!15}  &         &         &     \cellcolor{gray!15}    &         &         &         \\\hline
$(2,3,0,1)$ &         &         &         &     \cellcolor{gray!15}    &   \cellcolor{gray!15}      &         &         &         &         &    \cellcolor{gray!15}     &         &         \\\hline
$(3,0,2,1)$ &         &         &         &     \cellcolor{gray!15}$\tfrac{0.05}{2}$+$\textcolor{blue}{\tfrac{0.1}{4}+\tfrac{0.05}{4}}$&         &         &      \cellcolor{gray!15}$\tfrac{0.05}{2}$+$\textcolor{blue}{\tfrac{0.1}{4}}$&  \textcolor{red}{$\tfrac{0.05}{4}$}        &      \cellcolor{gray!15}$\tfrac{0.05}{2}$   &         &         &     $\textcolor{red}{\tfrac{0.1}{4}+\tfrac{0.05}{4}}$    \\\hline
$(0,3,2,1)$ &         &         &         &       \cellcolor{gray!15}$\tfrac{0.1}{2}$+$\textcolor{blue}{\tfrac{0.1}{4}+\tfrac{0.05}{4}}$&         &         &    \cellcolor{gray!15}$\tfrac{0.1}{2}$+$\textcolor{blue}{\tfrac{0.1}{4}}$&   \textcolor{red}{$\tfrac{0.05}{4}$}       &         &     \cellcolor{gray!15}$\tfrac{0.1}{2}$    &         &       $\textcolor{red}{\tfrac{0.1}{4}+\tfrac{0.05}{4}}$   \\\hline
$(2,0,3,1)$ &         &         &         &   \cellcolor{gray!15}$\tfrac{0.05}{2}$+$\textcolor{blue}{\tfrac{0.1}{4}+\tfrac{0.05}{4}}$&  \cellcolor{gray!15}$\tfrac{0.05}{2}$       &         &         &    $\textcolor{red}{\tfrac{0.1}{4}+\tfrac{0.05}{4}}$     &         &         &    \cellcolor{gray!15}$\tfrac{0.05}{2}$+$\textcolor{blue}{\tfrac{0.1}{4}}$&   \textcolor{red}{$\tfrac{0.05}{4}$}      \\\hline
$(0,2,3,1)$ &         &         &         &  \cellcolor{gray!15}$\tfrac{0.1}{2}$+$\textcolor{blue}{\tfrac{0.1}{4}+\tfrac{0.05}{4}}$&         &    \cellcolor{gray!15}$\tfrac{0.1}{2}$     &         &     $\textcolor{red}{\tfrac{0.1}{4}+\tfrac{0.05}{4}}$    &         &         &    \cellcolor{gray!15}$\tfrac{0.1}{2}$+$\textcolor{blue}{\tfrac{0.1}{4}}$&    \textcolor{red}{$\tfrac{0.05}{4}$}     \\\noalign{\hrule height 1pt}
$(3,1,0,2)$ &         &    \cellcolor{gray!15}     &         &         &         &         &         &    \cellcolor{gray!15}     &   \cellcolor{gray!15}      &         &         &    \\\hline 
$(1,3,0,2)$ &    \cellcolor{gray!15}     &         &         &         &         &         &         &    \cellcolor{gray!15}     &         &   \cellcolor{gray!15}      &         &    \\\hline 
$(3,0,1,2)$ &         &         &  \cellcolor{gray!15}$\tfrac{0.05}{2}$+$\textcolor{blue}{\tfrac{0.1}{4}}$&  \textcolor{red}{$\tfrac{0.05}{4}$}       &         &         &         &   \cellcolor{gray!15}$\tfrac{0.05}{2}$+$\textcolor{blue}{\tfrac{0.1}{4}+\tfrac{0.05}{4}}$&      \cellcolor{gray!15}$\tfrac{0.05}{2}$   &         &         &  $\textcolor{red}{\tfrac{0.1}{4}+\tfrac{0.05}{4}}$   \\\hline 
$(0,3,1,2)$ &         &         &  \cellcolor{gray!15}$\tfrac{0.1}{2}$+$\textcolor{blue}{\tfrac{0.1}{4}}$&    \textcolor{red}{$\tfrac{0.05}{4}$}     &         &         &         &    \cellcolor{gray!15}$\tfrac{0.1}{2}$+$\textcolor{blue}{\tfrac{0.1}{4}+\tfrac{0.05}{4}}$&         & \cellcolor{gray!15}$\tfrac{0.1}{2}$        &         &  $\textcolor{red}{\tfrac{0.1}{4}+\tfrac{0.05}{4}}$   \\\hline 
$(1,0,3,2)$ & \cellcolor{gray!15}$\tfrac{0.05}{2}$        &         &         &  $\textcolor{red}{\tfrac{0.1}{4}+\tfrac{0.05}{4}}$       &         &         &         & \cellcolor{gray!15}$\tfrac{0.05}{2}$+$\textcolor{blue}{\tfrac{0.1}{4}+\tfrac{0.05}{4}}$&         &         &  \cellcolor{gray!15}$\tfrac{0.05}{2}$+$\textcolor{blue}{\tfrac{0.1}{4}}$ & $\textcolor{red}{\tfrac{0.05}{4}}$ \\\hline 
$(0,1,3,2)$ &         &  \cellcolor{gray!15}$\tfrac{0.1}{2}$       &         &   $\textcolor{red}{\tfrac{0.1}{4}+\tfrac{0.05}{4}}$      &         &         &         &\cellcolor{gray!15}$\tfrac{0.1}{2}$+$\textcolor{blue}{\tfrac{0.1}{4}+\tfrac{0.05}{4}}$&         &         &   \cellcolor{gray!15}$\tfrac{0.1}{2}$+$\textcolor{blue}{\tfrac{0.1}{4}}$&  \textcolor{red}{$\tfrac{0.05}{4}$}  \\\noalign{\hrule height 1pt}
$(2,1,0,3)$ &         &  \cellcolor{gray!15}       &         &         &      \cellcolor{gray!15}   &         &         &         &         &         &         &  \cellcolor{gray!15}  \\\hline 

$(1,2,0,3)$ &  \cellcolor{gray!15}       &         &         &         &         & \cellcolor{gray!15}        &         &        &         &         &         & \cellcolor{gray!15}   \\\hline 

$(2,0,1,3)$ &         &         &  \cellcolor{gray!15}$\tfrac{0.05}{2}$+$\textcolor{blue}{\tfrac{0.1}{4}}$& \textcolor{red}{$\tfrac{0.05}{4}$}        & \cellcolor{gray!15}$\tfrac{0.05}{2}$        &         &         &  $\textcolor{red}{\tfrac{0.1}{4}+\tfrac{0.05}{4}}$        &         &         &         & \cellcolor{gray!15}$\tfrac{0.05}{2}$+$\textcolor{blue}{\tfrac{0.1}{4}+\tfrac{0.05}{4}}$ \\\hline 

$(0,2,1,3)$ &         &         & \cellcolor{gray!15}$\tfrac{0.1}{2}$+$\textcolor{blue}{\tfrac{0.1}{4}}$&   \textcolor{red}{$\tfrac{0.05}{4}$}      &         &  \cellcolor{gray!15}$\tfrac{0.1}{2}$       &         &    $\textcolor{red}{\tfrac{0.1}{4}+\tfrac{0.05}{4}}$     &         &         &         &\cellcolor{gray!15}$\tfrac{0.1}{2}$+$\textcolor{blue}{\tfrac{0.1}{4}+\tfrac{0.05}{4}}$   \\\hline 

$(1,0,2,3)$ & \cellcolor{gray!15}$\tfrac{0.05}{2}$        &         &         &   $\textcolor{red}{\tfrac{0.1}{4}+\tfrac{0.05}{4}}$       &         &         &  \cellcolor{gray!15}$\tfrac{0.05}{2}$+$\textcolor{blue}{\tfrac{0.1}{4}}$&  \textcolor{red}{$\tfrac{0.05}{4}$}        &         &         &         & \cellcolor{gray!15}$\tfrac{0.05}{2}$+$\textcolor{blue}{\tfrac{0.1}{4}+\tfrac{0.05}{4}}$   \\\hline 

$(0,1,2,3)$ &         & \cellcolor{gray!15}$\tfrac{0.1}{2}$        &         &   $\textcolor{red}{\tfrac{0.1}{4}+\tfrac{0.05}{4}}$     &         &         &   \cellcolor{gray!15}$\tfrac{0.1}{2}$+$\textcolor{blue}{\tfrac{0.1}{4}}$&   \textcolor{red}{$\tfrac{0.05}{4}$}       &         &         &         & \cellcolor{gray!15}$\tfrac{0.1}{2}$+$\textcolor{blue}{\tfrac{0.1}{4}+\tfrac{0.05}{4}}$ \\ \noalign{\hrule height 1pt}
$(0,0,0,0)$ &         &        &         &   $\textcolor{red}{0.1}$       &         &         &        &    $\textcolor{red}{0.1}$     &         &         &         & $\textcolor{red}{0.1}$
\end{tabular}}
\caption{Illustration of Construction $P_m = P_m^{(1)} + P_m^{(2)} + P_m^{(3)}$ for $P_X= (0.05,0.1,0.25,0.6)$.}
\label{tab:pm1+pm2+pm3}
\end{table}
\end{landscape}
}

\section{Proof of Proposition \ref{prop:construct_big_Pm}} \label{sec:proof_big_Pm}
In this section, we show that $P_m$ produced by Algorithm \ref{alg:construction_B} satisfies column-sum invariance, row-sum invariance, $\tfrac{\alpha}{T}$-capped column sum, and $\alpha$-bounded total sum.

\underline{\textbf{$P_m$ satisfies column-sum invariance:}}

Since we generate $P'_m$ using Algorithm \ref{alg:iterative_construction} with input $P'_X$. By proposition \ref{prop:construct_pm1}, $P'_m$ satisfies $P'_X$-column-sum invariance, that is 
\begin{align}
   \sum_{\zeta } P^{\prime }_{m}(i,\zeta )=P^{\prime }_{X}(i)=\begin{cases}\min (\tfrac{\alpha }{T} ,P_{X}(i))&,\forall i\in [1:N]\\ \frac{R}{n} &,\forall i\in [N+1:N+n]\end{cases}  
\end{align}
For case $n=0$, Algorithm \ref{alg:construction_B} yields
\begin{align*}
   \sum_{\zeta } P_{m}(i,\zeta ) & = \sum_{\zeta \in \cZ^N_T \setminus \mathbf{0_N} } P_{m}(i,\zeta )+ P_{m}(i,\mathbf{0}_N)  \\
   & = P'_X(i)+ P_{m}(i,\mathbf{0}_N) = \min (\tfrac{\alpha }{T} ,P_{X}(i)) + \boldsymbol{r}(i) \\
   & = \min (\tfrac{\alpha }{T} ,P_{X}(i)) + P_X(i) - \min (\tfrac{\alpha }{T} ,P_{X}(i)) = P_X(i), \forall i \in [1:N] 
\end{align*} 

For case $n > 0$, for $x \in [1:N]$, we have 
\begin{align*}
   \sum_{\zeta } P_{m}(x,\zeta ) &= \sum_{\zeta } \left( P'_{m}(x,\zeta ) + \sum_{i=N+1}^{N+n} \tfrac{\boldsymbol{r}(x)}{R}P'_{m}(i,\zeta )\right) \\
   & = \sum_{\zeta }P'_{m}(x,\zeta ) +  \tfrac{\boldsymbol{r}(x)}{R} \sum_{i=N+1}^{N+n}\sum_{\zeta }P'_{m}(i,\zeta ) \\
   & = P'_X(x) +  \tfrac{\boldsymbol{r}(x)}{R} \sum_{i=N+1}^{N+n}P'_X(i) \\
   & = P'_X(x) +  \tfrac{\boldsymbol{r}(x)}{R}\times n \times \tfrac{R}{n} = \min (\tfrac{\alpha }{T} ,P_{X}(x)) + \boldsymbol{r}(x) = P_X(x). 
\end{align*}
Therefore, in both the cases $n=0$ and $n>0$, $P_m$ preserves the column-sum invariance.

\underline{\textbf{$P_m$ satisfies row-sum invariance:}}

In both the cases $n=0$ and $n>0$, $P'_m$ preserves the row-sum invariance because of proposition \ref{prop:construct_pm1}. Thus, for fixed $\zeta$, we have 
\begin{align}
    \sum_{i = 1}^{N+n} P'_m(i,\zeta) = \sum_{i = 1}^{N+n} P'_{m'}(i,\zeta), \; \forall m' \in [1:T] \;  \text{and } \; m \neq m'.
\end{align}
For the case $n=0$, we already have $P_m = P'_m$ on the subset $\zeta \in \cZ^N_T \setminus \mathbf{0}_N$, so it remains to verify row-sum invariance when $\zeta = \mathbf{0}_N$. For every $m \in [1:T]$,
\begin{align}
    \sum_{i=1}^N P_m(i,\mathbf{0}_N) = \sum_{i=1}^N \boldsymbol{r}(i) = R.
\end{align}
Hence, for $\zeta = \mathbf{0}_N$, row-sum invariance still holds, because the sum does not depend on $m$.

For the case $n>0$, for any fixed $\zeta \in \cZ^{N+n}_T$, 
\begin{align*}
    \sum_{x=1}^N P_m(x,\zeta) & = \sum_{x=1}^N \left( P'_m(x,\zeta) + \sum_{i=N+1}^{N+n} \tfrac{\boldsymbol{r}(x)}{R}P'_m(i,\zeta) \right) \\&= \sum_{x=1}^N P'_m(x,\zeta) + \sum_{x=1}^N \tfrac{\boldsymbol{r}(x)}{R} \sum_{i=N+1}^{N+n}P'_m(i,\zeta) \\
    & = \sum_{x=1}^N P'_m(x,\zeta) +\sum_{i=N+1}^{N+n}P'_m(i,\zeta) = \sum_{x=1}^{N+n}P'_m(x,\zeta) \\
    &= \sum_{x=1}^{N+n}P'_{m'}(x,\zeta) = \sum_{x=1}^N P_{m'}(x,\zeta), \; \forall m' \in [1:T] \;  \text{and } \; m \neq m'.
\end{align*}
Thus, $P_m$ satisfies the row-sum invariance for both cases $n = 0$ and $n >0$. 

\underline{\textbf{$P_m$ satisfies $\tfrac{\alpha}{T}$-capped column sum:}}

We need to show that for each $x \in [1:N]$, 
\begin{align*}
    \sum_{\zeta \in \cZ^{N+n}_T} P_m(x,\zeta)\,\mathds{1}\{\gamma(x,\zeta) = m \} \geq \min\left(\frac{\alpha}{T}, P_X(x)\right). 
\end{align*}
Since Algorithm \ref{alg:iterative_construction} only allocates value to $(x,\zeta)$ pairs that $\gamma(x,\zeta) = m $ , we have 
\begin{align*}
    \sum_{\zeta \in \cZ^{N+n}_T} P'_m(x,\zeta)\,\mathds{1}\{\gamma(x,\zeta) = m \} = P'_X(x) = \min\left(\frac{\alpha}{T}, P_X(x)\right). 
\end{align*}
Thus, it is obvious that $P_m$ satisfies the $\tfrac{\alpha}{T}$-capped column-sum property for case $n=0$. 

For $n > 0$ and each fixed $x \in [1:N]$, 
\begin{align}
    \sum_{\zeta \in \cZ^{N+n}_T} P_m(x,\zeta)\,\mathds{1}\{\gamma(x,\zeta) = m \} & = \sum_{\zeta \in \cZ^{N+n}_T} \left(P'_m(x,\zeta) + \tfrac{\boldsymbol{r}(x)}{R} \sum_{i=N+1}^{N+n}P'_m(i,\zeta) \right) \,\mathds{1}\{\zeta \in \gamma_x^{-1}(m)\} \label{eq:pf_Pm_alpha/T_capped_n>0}\\
    & = \sum_{\zeta \in \cZ^{N+n}_T} P'_m(x,\zeta) \,\mathds{1}\{\zeta \in \gamma_x^{-1}(m)\} = \min\left(\frac{\alpha}{T}, P_X(x)\right).
\end{align}
The term $\sum_{i=N+1}^{N+n}P'_m(i,\zeta)$ in \eqref{eq:pf_Pm_alpha/T_capped_n>0} vanishes because, whenever $\zeta \in \gamma_x^{-1}(m)$ for some $x\in [1:N]$, the algorithm assigns no value to $P'_m(i,\zeta)$ for any $i \in [N+1:N+n]$, because the Algorithm \ref{alg:iterative_construction} will only assign one non-zero value on each row.

Thus, $P_m$ also satisfies $\tfrac{\alpha}{T}$-capped column sum for case $n>0$.

\underline{\textbf{$P_m$ satisfies $\alpha$-bounded total sum:}}

For any $x\in[1:N]$, 
\begin{align}
    & \sum_{m \in [1:T]} \sum_{\zeta \in \cZ^{N+n}_T} P_\cZ(\zeta)\mathds{1} {\{ \gamma(x,\zeta) = m\}}  = \sum_{m \in [1:T]} \sum_{\zeta \in \cZ^{N+n}_T} \sum_{i=1}^{N} P_m(i,\zeta)\mathds{1} {\{ \gamma(x,\zeta) = m\}}\\
    & = \sum_{m \in [1:T]} \sum_{\zeta \in \cZ^{N+n}_T} \sum_{i=1}^{N+n} P'_m(i,\zeta)\mathds{1} {\{ \gamma(x,\zeta) = m\}}  \label{eq:pf_Pm_B_alpha_bounded1} \\
    & = \sum_{m \in [1:T]} \sum_{\zeta \in \cZ^{N+n}_T} P'_m(x,\zeta)\mathds{1} {\{ \gamma(x,\zeta) = m\}} = \sum_{m \in [1:T]} P'_X(x) \leq T\frac{\alpha}{T} = \alpha,  \label{eq:pf_Pm_B_alpha_bounded2} 
\end{align}
where the equality in \eqref{eq:pf_Pm_B_alpha_bounded1} follows from the equations used in proving that $P_m$ satisfies the row-sum invariance property, and the first equality in \eqref{eq:pf_Pm_B_alpha_bounded2} holds because for $\zeta$ satisfying $\gamma(x,\zeta) = m$, Algorithm \ref{alg:iterative_construction} only assigns a value to $P'_m(x,\zeta)$. This concludes the proof. 

\section{Proof of Invalidity of He et al.'s Construction} \label{sec:he_construction_proof}

 We first formulate the optimization problem under their choice of  $(P_m,\gamma)$ in Table \ref{tab:construct_pm_he} as below. The variables are those elements in $P_m$, which are chosen to minimize the objective function. 
\begin{align}
   \min_{P_1,P_2, P_{\cZ}}  \; & \; \max\left( \sum_{x=1}^3 \sum_{ \zeta \notin \gamma^{-1}_x(1)} P_1(x,\zeta), \sum_{x=1}^3 \sum_{ \zeta \notin \gamma^{-1}_x(2)} P_2(x,\zeta)\right) \\
    \text{s.t.} & \sum_{m=1}^2 \sum_{\zeta \in \gamma^{-1}_x(m)} P_{\cZ}(\zeta) \leq \alpha, \forall x  \\
    & \sum_{\zeta \in \cZ_{bi}} P_m(x,\zeta) = P_X(x), \forall x , \forall m \label{eq:column_sum_constraint_counter_ex_1} \\
    & \sum_{x=1}^3 P_m(x,\zeta) = P_{\cZ}(\zeta), \forall \zeta \in \cZ_{bi}, \forall m \label{eq:row_sum_constraint_counter_ex} \\ 
    & P_m(x,\zeta) \geq 0, \forall x, \forall \zeta \in \cZ_{bi}, \forall m 
\end{align}

In this counterexample, we take $P_X = (0.01, 0.04, 0.95)$, set $T=2$, and choose $\alpha=0.99$. Let the decoding error probabilities for messages $1$ and $2$ be denoted by $\beta_1$ and $\beta_2$, respectively, given by
$$
\beta_1 = \sum_{x=1}^3 \sum_{\zeta \notin \gamma^{-1}_x(1)} P_1(x,\zeta), \qquad 
\beta_2 = \sum_{x=1}^3 \sum_{\zeta \notin \gamma^{-1}_x(2)} P_2(x,\zeta),
$$
so that the resulting objective is $\max(\beta_1,\beta_2)$.

Suppose we assign $q_1$ to the entry $P_1(3,(2,0,1))$, i.e., $P_1(3,(2,0,1)) = q_1$. To meet the column-sum constraint in \eqref{eq:column_sum_constraint_counter_ex_1}, the remaining entries in the same column, namely $P_1(3,\zeta)$ for $\zeta \neq (2,0,1)$, must together sum to $0.95 - q_1$. Hence,
\begin{align}
    \beta_1 \geq \sum_{\zeta \neq (2,0,1)} P_1(3, \zeta) = 0.95 - q_1.
\end{align}
The requirement for the value of $q_1$ is that it must not exceed either the token probability or the key probability; that is, $q_1 \leq 0.95$ and $q_1 \leq P_Z((2,0,1))$.
For message $2$, let $q_2$ denote the entry $P_2(1,(2,0,1))$, i.e., set $P_2(1,(2,0,1))=q_2$. By the row-sum constraint in \eqref{eq:row_sum_constraint_counter_ex}, the remaining entries in that row must add up to $P_Z((2,0,1)) - q_2$. Hence,
\begin{align}
    \beta_2 \ge \sum_{x \neq 1} P_2(x,(2,0,1)) = P_Z((2,0,1)) - q_2.
\end{align}
Moreover, $q_2$ must satisfy the same type of condition as $q_1$, namely $q_2 \le 0.01$ and $q_2 \le P_Z((2,0,1))$. Putting the constraints on $q_1$ and $q_2$ together gives
\begin{align}
    \beta_2 \ge P_Z((2,0,1)) - q_2 \ge q_1 - q_2 \ge q_1 - 0.01.
\end{align}

We can then derive a lower bound on the objective:
\begin{align}
    \max(\beta_1,\beta_2) \ge \frac{\beta_1+\beta_2}{2} \ge \frac{1}{2}\bigl(0.95-q_1 + q_1-0.01\bigr) = 0.47.
\end{align}
This bound is larger than the value reported in He et al., namely $1-\sum_x \min\!\bigl(\tfrac{\alpha}{T}, P_X(x)\bigr)=0.455$. Therefore, the reported $0.455$ cannot be the optimal value. This concludes the proof.

\vfill

\end{document}